\documentclass[acmsmall]{acmart} 
\renewcommand\footnotetextcopyrightpermission[1]{}  
\fancypagestyle{firstpagestyle}{%
  \fancyhf{}%
}
\usepackage{subcaption} %
\usepackage{amsmath,stmaryrd,bm,faktor}
\usepackage[arrow,matrix,frame,curve,cmtip]{xy}\CompileMatrices
\usepackage{graphicx,epsfig,color}
\usepackage{enumerate}
\usepackage{tikz}
\usetikzlibrary{automata,arrows,shapes.geometric,positioning}
\tikzstyle{enode}=[rectangle, draw]
\tikzstyle{anode}=[diamond, shape aspect=2, draw]
\usepackage{wrapfig}

\usepackage{bussproofs}
\alwaysRootAtTop

\newcommand{\nat}{\ensuremath{\mathbb{N}}}
\renewcommand{\phi}{\varphi}

\renewcommand{\vec}[1]{\ensuremath{\boldsymbol{#1}}}
\newcommand{\subvec}[3]{\ensuremath{\vec{#1}_{{#2},{#3}}}}

\newcommand{\sol}[2][]{\ensuremath{\mathit{sol}_{#1}({#2})}}

\newcommand{\Emoves}[1]{\ensuremath{\mathbf{E}({#1})}}
\newcommand{\Amoves}[1]{\ensuremath{\mathbf{A}({#1})}}
\newcommand{\Amovesred}[1]{\ensuremath{\mathbf{A}_r({#1})}}

\newcommand{\semlog}[1]{\ensuremath{\llbracket{#1}\rrbracket}}

\newcommand{\ord}[1]{\ensuremath{\mathit{ord}({#1})}}

\newcommand{\interval}[1]{\ensuremath{\underline{#1}}}

\newcommand{\down}[1]{\ensuremath{[{#1}]}}
\newcommand{\lift}[2]{\ensuremath{\down{#1}^{#2}_{\err}}}
\newcommand{\err}{\ensuremath{\star}}

\newcommand{\subst}[3]{\ensuremath{{#1}[{#2}:={#3}]}}

\newcommand{\fsubst}[3]{\ensuremath{{#1}[{#2}:={#3}]}}

\newcommand{\propsubst}[3]{\ensuremath{{#1}[\faktor{#3}{#2}]}}

\newcommand{\lub}{\ensuremath{\bigsqcup}}
\newcommand{\glb}{\ensuremath{\bigsqcap}}

\newcommand{\twoheaddownarrow}{\mathrel{\rotatebox[origin=c]{270}{$\twoheadrightarrow$}}}
\newcommand{\cone}[1]{\ensuremath{\mathop{\twoheaddownarrow\!{#1}}}}
\newcommand{\filter}[1]{\ensuremath{\mathop{\uparrow\!{#1}}}}

\newcommand{\Pow}[1]{\ensuremath{\mathbf{2}^{#1}}}

\newcommand{\asc}[1]{\ensuremath{\lambda_{#1}}}

\newcommand{\decr}{\mathit{decrease}}

\newcommand{\atom}[2]{\ensuremath{[{#1},{#2}]}}

\newcommand{\true}{\ensuremath{\mathbf{t}}}
\newcommand{\false}{\ensuremath{\mathbf{f}}}
\newcommand{\PVar}{\ensuremath{\mathit{PVar}}}
\newcommand{\Prop}{\ensuremath{\mathit{Prop}}}
\newcommand{\sem}[1]{\ensuremath{|\!|{#1}|\!|}}

\newenvironment{mytab}{\begin{tabbing} xxx \= xxx \= xxx \= xxx \= xxx \=
    xxx \= xxx \= xxx \= xxx \= xxx \= xxx \= xxx \= xxx \= xxx 
    \kill}{\end{tabbing}}

\newcommand{\prg}[1]{\texttt{#1}}
\newcommand{\blk}[2]{$[$#2$]^{#1}$}

\usepackage[normalem]{ulem}

\begin{document}

\title[Fixpoint Games on Continuous Lattices]{Fixpoint Games on Continuous Lattices}

\subtitle{(full version)}

\author{Paolo Baldan}
\orcid{nnnn-nnnn-nnnn-nnnn}             
\affiliation{
  \department{Dipartimento di Matematica ``Tullio Levi-Civita''}       
  \institution{Universit\`a di Padova}      
  \streetaddress{Via Trieste, 63}
  \city{Padova}
  \postcode{I-35121}
  \country{Italy}               
}
\email{baldan@math.unipd.it}    

\author{Barbara K\"onig}
\orcid{nnnn-nnnn-nnnn-nnnn}             
\affiliation{
  \department{Fakult\"at f\"ur Ingenieurwissenschaften, Abteilung Informatik und Angewandte Kognitionswissenschaft}             
  \institution{Universit\"at Duisburg-Essen}            
  \streetaddress{Lotharstra{\ss}e 65}
  \city{Duisburg} 
  \postcode{47048}
  \country{Germany}                    
}
\email{barbara_koenig@uni-due.de}      

\author{Christina Mika-Michalski}
\orcid{nnnn-nnnn-nnnn-nnnn}             
\affiliation{
  \department{Fakult\"at f\"ur Ingenieurwissenschaften, Abteilung Informatik und Angewandte Kognitionswissenschaft}             
  \institution{Universit\"at Duisburg-Essen}    
  \streetaddress{Lotharstra{\ss}e 65}
  \city{Duisburg} 
  \postcode{47048}
  \country{Germany}    
}
\email{	christina.mika-michalski@uni-due.de}

\author{Tommaso Padoan}
\orcid{nnnn-nnnn-nnnn-nnnn}             
\affiliation{
  \department{Dipartimento di Matematica ``Tullio Levi-Civita''}      
  \institution{Universit\`a di Padova}       
  \streetaddress{Via Trieste, 63}
  \city{Padova}
  \postcode{I-35121}
  \country{Italy}
}
\email{padoan@math.unipd.it}    

\begin{abstract}
  Many analysis and verifications tasks, such as static program
  analyses and model-checking for temporal logics, reduce to the
  solution of systems of equations over suitable lattices.
  Inspired by recent work on lattice-theoretic progress measures, we
  develop a game-theoretical approach to the solution of systems of
  monotone equations over lattices, where for each single equation
  either the least or greatest solution is taken. A simple parity
  game, referred to as fixpoint game, is defined that provides a
  correct and complete characterisation of the solution of systems of
  equations over continuous lattices, a quite general class of
  lattices widely used in semantics.
  For powerset lattices the fixpoint game is intimately connected with
  classical parity games for $\mu$-calculus model-checking, whose
  solution can exploit as a key tool Jurdzi\'nski's small progress
  measures. We show how the notion of progress measure can be
  naturally generalised to fixpoint games over
  continuous lattices and we prove the existence of small progress
  measures.
  Our results lead to a constructive formulation of progress
  measures as (least) fixpoints. We refine this characterisation by
  introducing the notion of selection that allows one to constrain the
  plays in the parity game, enabling an effective (and possibly
  efficient) solution of the game, and thus of the associated
    verification problem. We also propose a logic for specifying the
  moves of the existential player that can be used to systematically
  derive simplified equations for efficiently computing
  progress measures.
  We discuss potential applications to the model-checking of latticed
  $\mu$-calculi and to the solution of fixpoint equations
    systems over the reals.
\end{abstract}



\maketitle

\section{Introduction}

Systems of fixpoint equations are ubiquitous in formal analysis and
verification. For instance, program
analysis~\cite{nnh:program-analysis} uses the flow graph of a
program to generate a set of constraints specifying how the information of interest at the different program points is interrelated.
The set of constraints can be viewed as a system of fixpoint
equations, whose (least or greatest) solution provides a sound
approximation of the properties of the  program. Invariant/safety properties
can be characterised as greatest fixpoints, while
liveness/reachability properties as least fixpoints. Behavioural
equivalences (for instance for process calculi) are typically defined
as the solution of a fixpoint equation. The most famous example is
bisimilarity that can be characterised as the greatest fixpoint of a
suitable operator over the lattice of binary relations on the set of states (see,
e.g.,~~\cite{s:bisimulation-coinduction}).

Almost invariably, in the mentioned settings, the involved functions
are monotone and the domains of interest are complete lattices where
the key result for deriving the existence of (least or greatest)
fixpoints is Knaster-Tarski's fixpoint
theorem~\cite{t:lattice-fixed-point}.

Least and greatest fixpoint can be profitably mixed, in order to
obtain expressive specification logics, among which the
$\mu$-calculus~\cite{k:prop-mu-calculus} is a classical example. The
$\mu$-calculus is very expressive, but the nesting of fixpoints
increases the
complexity of model-checking. Common approaches to the model-checking
problem rely on an encoding in terms of
parity
games (see, e.g.,~\cite{bw:mu-calculus-modcheck,s:local-modcheck-games,ej:tree-automata-mu-determinacy}).
The seminal paper~\cite{j:progress-measures-parity} provides an
algorithm for the solution of parity games which is polynomial in the
number of states and exponential in (half of) the alternation depth,
recently improved to quasi-polynomial in~\cite{CJKLS:DPGQPT}. A
detailed discussion of the complexity of $\mu$-calculus model-checking
can be found in~\cite{bw:mu-calculus-modcheck}.

It has been recently observed in~\cite{hsc:lattice-progress-measures}
that progress measures, a key ingredient in Jurdzis\'nki's algorithm for
solving parity games, are amenable to a generalisation to
systems of fixpoint equations over general lattices.
A constructive characterisation of such progress measures is given in
the case of powerset lattices and used to derive model-checking
procedures for (branching and linear) coalgebraic logic.
For general lattices, however, the notion of progress measure
in~\cite{hsc:lattice-progress-measures} does not exactly correspond
to Jurdzinski's notion. In particular, there is no algorithm for
actually computing such progress measures, they rather play the role
of invariants respectively ranking functions that have somehow to be
provided.
While the possibility of deriving generic algorithms for solving
systems of equations is very appealing, the restriction to powerset
lattices limits the applicability of the technique. Often program
analysis relies on lattices which are not powerset lattices (and
neither distributive, hence they cannot be seen as sublattices of
powerset lattices). Moreover also settings involving fuzziness,
probabilities or in general quantitative information are not captured
by restricting to powerset lattices.

Inspired by the mentioned work, in this paper we devise a
game-theoretical approach to the solution of systems of fixpoint
equations over a vast class of lattices, the so-called continuous
lattices. Originally studied by Scott in connection with the semantics
of the $\lambda$-calculus~\cite{Scott:CL}, they have later been
recognised as a fundamental structure, with a plethora of applications
in the semantics of programming languages and, more generally, in the
theory of
computation~\cite{ghklms:continuous-lattices-domains,AJ:DT}. They
include discrete structures, such as most domains used in program
analysis, and continuous structures, such as the real interval $[0,1]$
or the lattice of open sets of a locally compact Hausdorff space.

The possibility of characterising the least or the greatest fixpoint
of a (single) monotone function over a powerset lattice in terms of a
game between an existential and an universal player is probably
folklore and has been observed in~\cite{v:lectures-mu-calculus} where
the game is referred to as an unfolding game.
As a first result, here we show how the unfolding game can be extended to
work for a \emph{system} of fixpoint equations over lattices,
resulting in a surprisingly simple game that we refer to as a
\emph{fixpoint game}.
Mixing least and greatest fixpoint equations requires a non-trivial
winning condition, which however arises as a natural adaptation to our
setting of the one for parity games.

For the simpler case of powerset lattices the interaction
between the players in the fixpoint game fundamentally relies on the possibility of testing
the presence of elements in the image of a set and on the fact that a subset
is completely determined by the elements that belongs to
it. When moving to a more general class of lattices we need to ensure that this kind of interaction can be suitably mimicked.
We argue in the paper that continuous lattices provide an extremely
natural setting for this extension, providing exactly the necessary machinery
for stating results in a way which is analogous to the powerset case.
In fact, they come equipped with a notion of ``finitary'' approximation based on the
\emph{way-below} relation and
each element arises as the join of the elements%
(possibly restricted to a selected basis) which are way-below
it, in the same way as a subset is the union of its singletons.

The proof that our fixpoint game provides a correct and complete
characterisation of the solution of a system of fixpoint equations
over a continuous lattice relies on $\mu$- and $\nu$-approximants that
provide a clear notion of approximation of the solution. In
particular, $\mu$-approximants turn out to be closely related to the
progress measures of~\cite{hsc:lattice-progress-measures}, a
  connection that we will make precise in
  Appendix~\ref{sec:hasuo}.

We  show how Jurdzi\'nski's approach for solving parity
games~\cite{j:progress-measures-parity} can be generalised to systems
of fixpoint equations over continuous lattices. In particular we
introduce a notion of progress measure for fixpoint games over
continuous lattices. Intuitively, given an element $b$ of the basis
of the lattice and an equation index $i$, the progress measure
provides a vector of ordinals, specifying how many iterations are
needed for each equation to cover $b$ in the $i$-th component of the
solution. Then
we prove the existence of suitably defined small progress
measures. This result enables a constructive characterisation of
progress measures as (least) fixpoints and provides a recipe for
computing the progress measure that can be straightforwardly
implemented, at least for finite lattices.

We refine the fixpoint characterisation of progress
measures by introducing the notion of selection, which basically
constrains the moves of the existential player in the parity game,
still preserving correctness and completeness, thus enabling a
more efficient solution of the game.  We also
define a logic for providing a symbolic representation of the moves
of the existential player that can be directly
translated into a system of fixpoint equations describing the progress
measure.
In
particular, we discuss selections and logical formulae that
are needed to handle $\mu$-calculus model-checking.

As an example of application beyond standard $\mu$-calculus
model-checking we will discuss the case of  latticed $\mu$-calculi,
where the evaluation of a formula for a state gives a lattice element,
generalising the standard truth values $0,1$ (see, e.g.,
~\cite{kl:latticed-simulation,GLLS:DNMC,ekn:bisim-heyting}). This
happens naturally also when $\mu$-calculus formulae are evaluated over
weighted transition systems or over probabilistic automata~\cite{hk:quantitative-analysis-mc}.

The lattice under consideration -- and hence its basis --
might be
infinite and in this case it is not even guaranteed that the fixpoint
iteration terminates in $\omega$ steps. In fact, it is known
that, despite the functions involved in the equations of the system
being continuous, due to alternation of least and greatest fixpoints,
discontinuous functions may arise and we possibly have to refer to
ordinals beyond $\omega$~\cite{ms:lukasiewicz-mu-arxiv,f:continuous-fragment-mu}. In order to
solve these cases, one has to restrict to a finite part of the
lattice, approximate or resort to symbolic representations.
We take some preliminary steps in this direction proposing a technique
to deal with infinite lattices that is always correct, and
complete under certain conditions (for instance, on well-orders
without other restrictions, or on the reals, suitably restricting
the functions). In
particular, inspired by~\cite{MS:MS,ms:lukasiewicz-mu-arxiv}, we
show how the game for solving fixpoint equations over the reals can
be encoded into an SMT formula of fixed size, capturing the winning
condition of the existential
player.

Summing up, our main contributions are the following:
\begin{itemize}
\item We propose a game-theoretical characterisation of the
  solution of systems of fixpoint equations over lattices and we
  identify continuous lattices as a general and appropriate setting
  for such theory.
\item We develop a theory of progress measures \`{a} la Jurdzi\'nski
  in this general framework, with a clear recipe for their
  computation. This can be seen as a generalisation of the MC progress
  measures, proposed in~\cite{hsc:lattice-progress-measures} for
  coalgebraic logics over powerset lattices.
\item We devise strategies for the computation of such progress
  measures based on selections and a logic for the symbolic
  representation of players' moves,
  along with
  a
  complexity analysis.

\item We explicitly discuss two application scenarios, beyond
  standard $\mu$-calculus over powerset lattices: model-checking of
  latticed $\mu$-calculi via progress measures and the solution of
  fixpoint equation systems over the reals via SMT solvers.

\end{itemize}

We believe that due to the generality of our results, there is the
potential for several more interesting applications for different
lattices.

The rest of the paper is structured as follows.
In \S~\ref{sec:preliminaries} we recap the basics of continuous
lattices and introduce some notation that will be used throughout the
paper.
In \S~\ref{sec:system} we introduce the systems of fixpoint
equations over a lattice, we define their solution and devise a
corresponding notion of approximation.
In \S~\ref{sec:fp-equation-games} we present a game-theoretical
approach to the solution of a system of equations over a continuous
lattice, together with several case studies.
In \S~\ref{sec:progress} we introduce the notion of
progress measure for (the game associated with) systems of fixpoint
equations over a continuous lattice.
In \S~\ref{sec:applications-latticed} we discuss the application of
our framework to the model-checking of latticed $\mu$-calculi.
In \S~\ref{sec:solving-fp-equations-smt} we present some results for
the solutions of systems of fixpoint equations over infinite lattices,
with special focus on real intervals.
In \S~\ref{sec:conclusions} we conclude the paper and outline future
research.
All proofs,
  further details on the encoding of $\mu$-calculus formulae into
  fixpoint equation systems (and vice versa) and a detailed comparison
  to~\cite{hsc:lattice-progress-measures} can be found in the
  appendix.

\section{Preliminaries on Ordered Structures}
\label{sec:preliminaries}

In this section we provide the basic order theoretic notions that will
be used throughout the paper. In particular, we define continuous
lattices and we provide some notation about tuples of elements that will be useful for compactly describing the solution of systems of equation.

\subsection{Lattices}

A preordered or partially ordered set $\langle P, \sqsubseteq \rangle$
is often denoted simply as $P$, omitting the (pre)order
relation.
It is \emph{well-ordered} if every non-empty subset
$X \subseteq P$ has a minimum.
The \emph{join} and the \emph{meet} of a
subset $X \subseteq P$ (if they exist) are denoted  $\bigsqcup X$
and $\bigsqcap X$, respectively.
 
\begin{definition}[complete lattice, basis, irreducibles]
  A \emph{complete lattice} is a partially ordered set
  $(L, \sqsubseteq)$ such that each subset $X \subseteq L$ admits a
  join $\lub X$ and a meet $\glb X$. A complete lattice
  $(L, \sqsubseteq)$ always has a least element
  $\bot = \lub \emptyset$ and a greatest element
  $\top = \glb \emptyset$.
  Given an element $l \in L$ we define its upward-closure
  $\filter{l} = \{ l' \mid l' \in L\ \land\ l \sqsubseteq l' \}$.
  A \emph{basis} for a lattice is a subset $B_L \subseteq L$ such that
  for each $l \in L$ it holds that
  $l = \lub \{ b \in B_L \mid b \sqsubseteq l \}$.
  An element $l \in L$ is \emph{completely join-irreducible} if whenever
  $l = \bigsqcup X$ for some $X \subseteq L$ then $l \in X$.
\end{definition}

Since all lattices in this paper will be complete, we will often omit
the qualification ``complete''. Similarly, since we are only interested in completely join-irreducible elements we will often omit
the qualification ``completely''. Note that $\bot$ is never an irreducible since $\bot = \bigsqcup \emptyset$ and $\bot \not\in \emptyset$.

\begin{example}
  Three simple examples of lattices, that we will refer to later,
  are:
  \begin{itemize}    
  \item The powerset of any set $X$, ordered by subset inclusion
    $(\Pow{X}, \sqsubseteq)$.  Join is union, meet is intersection,
    top is $X$ and bottom is $\emptyset$. A basis is the set of
    singletons $B_{\Pow{X}} = \{ \{x\} \mid x \in X\}$. These are
    also the the join-irreducible elements. Any set $Y \subseteq X$
    with $|Y|>1$ is not irreducible, since
    $Y = \bigsqcup_{x \in Y} \{x\}$ but clearly $Y \neq \{x\}$ for
    any $x \in Y$.
    
  \item The real interval $[0,1]$ with the usual order $\leq$. Join
    and meet are the sup and inf over real numbers, $0$ is bottom
    and $1$ is top. The rationals $\mathbb{Q} \cap (0,1]$ are a
    basis. There are no irreducible elements (in fact, for any
    $x \in [0,1]$ we have that $x = \bigsqcup \{ y \mid y < x \}$
    and clearly $x \not\in \{ y \mid y < x \}$).
    
  \item Consider the partial order
    $W = \mathbb{N} \cup \{ \omega, a \}$ depicted in
    Fig.~\ref{fi:non-continuous}. It is easy to see that it is a
    lattice. All elements are irreducible apart from the bottom $0$
    and the top $\omega$. For the latter notice that, e.g.,
    $\omega = \bigsqcup \{1,a\}$.
    
  \end{itemize}
\end{example}

\begin{figure}
  \begin{center}
    \begin{tikzpicture}[node distance=6mm, >=stealth',x=10mm,y=6mm]
      \node at (0,0) (n0)  {$0$};
      \node at (0,1) (n1)  {$1$};
      \node at (0,2) (n2)  {$2$};
      \node at (0,4) (w)  {$\omega$};
      \node at (1,2) (a)  {$a$};
      \draw [-] (n0) -- (n1);
      \draw [-] (n1) -- (n2);
      \draw [dotted] (n2) -- (w);
      \draw [-] (n0) -- (a);
      \draw [-] (a) -- (w);
    \end{tikzpicture}
  \end{center}
  \caption{A complete lattice $W$ which is not continuous.}
  \label{fi:non-continuous}
\end{figure}
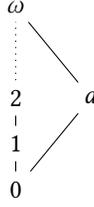

A lattice is \emph{completely distributive} if
\[ \bigsqcup_{k\in K} \bigsqcap_{j\in J_k} l_{k,j} = \bigsqcap
  \{\bigsqcup_{k\in K} r_{k,j^k} \mid j^k\in J_k, k\in K \} \] where
$K, J_k, K$ are index sets and $l_{k,j} \in L$.

A function $f\colon L \to L$ is \emph{monotone} if for all
$l, l' \in L$, if $l \sqsubseteq l'$ then
$f(l) \sqsubseteq f(l')$. By Knaster-Tarski's
theorem~\cite{t:lattice-fixed-point}, any monotone
function on a complete lattice has a least and a greatest fixpoint,
denoted respectively $\mu f$ and $\nu f$, characterised as the meet
of all pre-fixpoints respectively the join of all post-fixpoints:
\begin{center}
  $\mu f = \glb \{ l \mid f(l) \sqsubseteq l \}$ \qquad
  $\nu f = \lub \{ l \mid l \sqsubseteq f(l) \}$
\end{center}
The least and greatest fixpoint can also be obtained by iterating the
function on the bottom and top elements of the lattice. This is often
referred to as Kleene's theorem (at least for continuous functions)
and it is one of the pillars of abstract
interpretation~\cite{CC:CCVTFP}. Given a lattice $L$, define its
\emph{height} $\asc{L}$ as the supremum of the length of any strictly
ascending, possibly transfinite, chain. Then we have the following result.

\begin{theorem}[Kleene's iteration~\cite{CC:CCVTFP}]
  Let $L$ be a lattice and let $f \colon L \to L$ be a monotone
  function.  Consider the (transfinite) ascending chain
  $(f^{\beta}(\bot))_{\beta}$ where $\beta$ ranges over the ordinals,
  defined by $f^0(\bot) = \bot$,
  $f^{\alpha+1}(\bot) = f(f^\alpha(\bot))$ for any ordinal $\alpha$
  and $f^{\alpha}(\bot) = \lub_{\beta < \alpha} f^{\beta}(\bot)$ for
  any limit ordinal $\alpha$. Then $\mu f = f^{\gamma}(\bot)$ for some
  ordinal $\gamma \leq \asc{L}$. The greatest fixpoint $\nu f$ can be
  characterised dually, via the (transfinite) descending chain
  $(f^{\alpha}(\top))_{\alpha}$.
\end{theorem}
Note also that $f^\alpha(\bot)$ is
always a post-fixpoint and $f^\alpha(\top)$ is always a pre-fixpoint.

We will focus on special lattices where elements are generated by
suitably defined approximations.  Given a lattice $L$, a subset
$X \subseteq L$ is \emph{directed} if $X \neq \emptyset$ and every
pair of elements in $X$ has an upper bound in $X$.

\begin{definition}[way-below relation, continuous lattices]
  \label{de:continuous-lattice}
  Let $L$ be a lattice. Given two elements $l, l' \in L$ we say that
  $l$ is \emph{way-below} $l'$, written $l \ll l'$ when for every
  directed set $D \subseteq L$, if $l' \sqsubseteq \lub D$ then there
  exists $d \in D$ such that $l \sqsubseteq d$.  We denote by
  $\cone{l}$ the set of elements way-below $l$, i.e.,
  $\cone{l} = \{ l' \mid l' \in L\ \land\ l' \ll l \}$.

  The lattice $L$ is called \emph{continuous} if $l = \lub \cone{l}$
  for all $l \in L$.
\end{definition}

Intuitively, the way-below relation captures a form of finitary
approximation: if one imagines that $\sqsubseteq$ is an order on the
information content of the elements, then $x \ll y$ means that
whenever $y$ can be ``covered'' by joining (possibly small) pieces of
information, then $x$ is covered by one of those pieces. Then a
lattice is continuous if any element can be built by joining its
approximations.

Concerning the origin of the name ``continuous lattice'', we can
quote~\cite{Scott:CL} that says that ``One of the justifications (by
euphony at least) of the term \emph{continuous lattice} is the fact
that such spaces allow for so many continuous functions.'' For
instance, one indication is the fact that meet and join are both
continuous in such lattices.

It can be shown that if $L$ is a continuous lattice and $B_L$ is a basis, for all $l \in L$, it holds that $l = \lub (\cone{l} \cap B_L)$. 

Various lattices that are commonly used in semantics enjoy a property stronger  than continuity, defined below.

\begin{definition}[compact element, algebraic lattice]
  Let $L$ be a lattice. An element $l \in L$ is called \emph{compact}
  whenever $l \ll l$. The lattice $L$ is \emph{algebraic} if the set of
  compact elements is a basis.
\end{definition}

\begin{example}
  Some examples are as follows:
  \begin{itemize}
  \item All finite lattices are continuous (since every finite
    directed set has a maximum). More generally, all algebraic
    lattices (which include all finite lattices) are continuous. The
    way-below relation is $x \ll y$ if $x$ compact and
    $x \sqsubseteq y$.
    
  \item Given a set $X$, the powerset lattice $\Pow{X}$, ordered by
    inclusion, is an algebraic lattice. The compact elements are the
    finite subsets. In fact, any set $Y$ is the union of its finite
    subsets, i.e.,
    $Y = \bigcup \{ F \mid F \subseteq Y\ \land\ F \text{
      finite}\}$. Since
    $\{ F \mid F \subseteq Y\ \land\ F \text{ finite}\}$ is directed
    set, compactness requires that $Y \subseteq F$ for some finite
    $F \subseteq Y$, hence $Y=F$.
    Therefore $Y \ll Z$ holds when $Y$ is finite and $Y \subseteq Z$.

  \item The interval $[0,1]$ with the usual order $\leq$ is a
    continuous lattice.  For $x, y \in [0,1]$, we have $x \ll y$ when
    $x < y$ or $x=0$.  In fact, each $\emptyset \neq Y\subseteq [0,1]$
    is directed. Imagine that $y\le \bigsqcup Y$ for such a $Y$. Then
    by standard properties of the reals there always exists a
    $y'\in Y$ such that $x\le y'$ if and only if $x<y$ or $x=0$. Note
    that this lattice is not algebraic since the only compact element
    is $0$.
    
  \item The lattice $W$ in Fig.~\ref{fi:non-continuous} is not
    continuous. In fact, $a \not\ll a$ since
    $a \sqsubseteq \bigsqcup \mathbb{N}$ but $a \not\sqsubseteq i$ for
    all $i \in \mathbb{N}$. 
    Therefore $\cone{a} = \{0\}$ and thus $a \neq \bigsqcup \cone{a}$.

  \end{itemize}
\end{example}

\subsection{Tuples and Ordinals}

We will often consider tuples of elements. 
Given a set $A$, an $n$-tuple in $A^n$ will be denoted by a
boldface letter $\vec{a}$. The components of a tuple $\vec{a}$ will be
denoted by using the same name of the tuple, not in boldface style and
with an index, i.e., $\vec{a} = (a_1, \ldots, a_n)$. For an index
$n \in \mathbb{N}$ we use the notation $\interval{n}$ to denote the
integer interval $\{1, \ldots, n\}$. Given $\vec{a} \in A^n$ and
$i, j \in \interval{n}$ we write $\subvec{a}{i}{j}$ for the subtuple
$(a_i, a_{i+1}, \ldots, a_j)$.

\begin{definition}[pointwise order]
  Given a lattice $(L,\sqsubseteq)$ we will denote by
  $(L^n,\sqsubseteq)$ the set of $n$-tuples endowed with the
  \emph{pointwise order} defined, for $\vec{l}, \vec{l'} \in L^n$, by
  $\vec{l} \sqsubseteq \vec{l'}$ if $l_i \sqsubseteq l_i'$ for all
  $i \in \interval{n}$.
\end{definition}

The structure $(L^n, \sqsubseteq)$ is a lattice and it is continuous
if $L$ is continuous, with the way-below relation given by
$\vec{l} \ll \vec{l}'$ iff $l_i \ll l_i'$ for all
$i \in \interval{n}$~\cite[Proposition
I-2.1]{ghklms:continuous-lattices-domains}. More generally, for any
set $X$, the set of functions $L^X = \{ f \mid f\colon X \to L \}$, endowed with pointwise order, is a lattice (continuous when $L$ is).

\begin{definition}[lexicographic order]
  Given a partial order $(P, \sqsubseteq)$ we will denote by
  $(P^n, \preceq)$ the set of $n$-tuples endowed with the
  \emph{lexicographic order} (where the last component is the most
  relevant), i.e., inductively, for $\vec{l}, \vec{l}' \in P^n$, we
  let $\vec{l} \preceq \vec{l}'$ if either $l_n \sqsubset l_n'$ or
  $l_n = l_n'$ and $\subvec{l}{1}{n-1} \preceq \subvec{l'}{1}{n-1}$.
\end{definition}
When $(L, \sqsubseteq)$ is a lattice also $(L^n,\preceq)$
is a lattice. Given a set $X \subseteq L^n$, the meet of $X$ with respect to $\preceq$ can be obtained by taking the meet of
the single components, from the last to the first, i.e., it is defined
inductively as $\bigsqcap X = \vec{l}$ where
$l_i = \bigsqcap \{ l_i' \mid \vec{l}' \in X\ \land\ \subvec{l'}{i+1}{n} =
\subvec{l}{i+1}{n} \}$. The join can be defined analogously.
Similarly, one can show that $\preceq$ is a well-order whenever
$\sqsubseteq$ is.

As in~\cite{j:progress-measures-parity,hsc:lattice-progress-measures}, we will also need to consider tuples with a preorder arising from the
lexicographic order, when some components are considered irrelevant.

\begin{definition}[truncated lexicographic order]
  \label{de:trunc-ord}
  Let $(P,\sqsubseteq)$ be a partial order and let $n \in \nat$. For
  $i \in \interval{n}$ we define a preorder $\preceq_i$ on $P^n$
  by letting, for $\vec{x}, \vec{y} \in P^n$,
  $\vec{x} \preceq_{i} \vec{y}$ if
  $\subvec{x}{i}{n} \preceq \subvec{y}{i}{n}$. We write $=_i$ for the
  equivalence induced by $\preceq_i$ and $\vec{x} \prec_{i} \vec{y}$ for $\vec{x} \preceq_{i} \vec{y}$ and $\vec{x} \not=_{i} \vec{y}$.
  Whenever $\sqsubseteq$ is a well-order, given $X \subseteq P^n$ with
  $X\neq \emptyset$ and $i \in \interval{n}$, we write
  $\min\nolimits_{\preceq_i} X$ for the vector
  $\vec{x} = (\bot, \ldots, \bot, x_i, \ldots, x_n)$ where
  $\subvec{x}{i}{n} = \min_\preceq \{ \subvec{l}{i}{n} \mid \vec{l}
  \in X \}$.
\end{definition}

In words, $\preceq_i$ is the
lexicographic order restricted to the components
$i, i+1, \ldots, n$. 
For instance, if $P = \mathbb{N}$ with the usual order, then $(6,1,4,7) \prec_2 (5,2,4,7)$, while  $(6,1,4,7) =_3 (5,2,4,7)$.

We denote \emph{ordinals} by Greek letters
$\alpha, \beta, \gamma, \ldots$ and their order by $\leq$. The
collection of all ordinals is well-ordered. Given any ordinal
$\alpha$, the collection of ordinals dominated by $\alpha$ is a set
$\down{\alpha} = \{ \lambda \mid \lambda \leq \alpha \}$, which, seen
as an ordered structure, is a lattice. Meet and join of a set $X$ of
ordinals will be denoted by $\inf X$ (which equals $\min X$ if
$X\neq\emptyset$) and $\sup X$.
The lattice $\down{\alpha}$ is completely distributive, which follows
from classical results. In fact, the complete join-irreducibles are
all ordinals which are not limit ordinals. Hence, from~\cite[Theorems
1 and 2]{Ran:CDCL}, since every element is the join of completely
join-irreducible elements, we can conclude that $\down{\alpha}$ is
completely distributive.
A similar argument shows that, for a fixed
  $n \in \nat$ and ordinal $\alpha$, the lattice of $n$-tuples of
  ordinals, referred to as \emph{ordinal vectors}, endowed with the
  lexicographic order $(\down{\alpha}^n, \preceq)$ is completely
  distributive. In fact, the only elements that are not complete
  join-irreducibles are vectors of the kind
  $(0,\ldots,0,\alpha,\beta_i, \ldots, \beta_n)$ where $\alpha$ is a
  limit ordinal and such vectors can be obtained as the join of the
  vectors $(0, \ldots , 0, \beta, \beta_i, \ldots, \beta_n)$, with
  $\beta < \alpha$ and $\beta$ a successor ordinal.

\section{Fixpoint Equations: Solutions and Approximants}
\label{sec:system}

In this section we introduce the systems of fixpoint equations we will work with in the paper. We define the solution of a system and we devise some results concerning its approximations that will play a major role later.

\subsection{Systems of Fixpoint Equations}

We focus on systems of (fixpoint) equations over some lattice, where,
for each equation one can be interested either in the least or in
the greatest solution.

\begin{definition}[system of equations]
  \label{de:system}
  Let $L$ be a lattice. A system of equations $E$ over $L$ is a list of
  equations of the following form
  \begin{eqnarray*}
    x_1 & =_{\eta_1} & f_1(x_1,\ldots,x_m) \\
        & \ldots & \\
    x_m & =_{\eta_m} & f_m(x_1,\ldots,x_m) 
  \end{eqnarray*}
  where $f_i\colon L^m\to L$ are monotone functions and
  $\eta_i\in\{\mu,\nu\}$.
  The system will often be denoted as
  $\vec{x} =_{\vec{\eta}} \vec{f} (\vec{x})$, where $\vec{x}$,
  $\vec{\eta}$ and $\vec{f}$ are the obvious tuples.
  We denote by
  $\emptyset$ the system with no equations.

\end{definition}

Systems of equations of this kind have been considered by
various authors,
e.g.,~\cite{cks:faster-modcheck-mu,s:fast-simple-nested-fixpoints,hsc:lattice-progress-measures}.
In particular,~\cite{hsc:lattice-progress-measures} works on general
lattices.

We next define the pre-solutions of a system as tuples of lattice elements that, replacing the variables, satisfy all the equations of the system. The solution will be a suitably chosen pre-solution, taking into account also the $\eta_i$ annotations that specify for each equation whether the least or greatest solution is required.

\begin{definition}[pre-solution]
  \label{le:pre-solution}
  Let $L$ be a lattice and let $E$ be a system of equations over $L$
  of the kind $\vec{x} =_{\vec{\eta}} \vec{f}(\vec{x})$. A
  \emph{pre-solution} of $E$ is any tuple $\vec{u} \in L^m$ such that
  $\vec{u} = \vec{f}(\vec{u})$.
\end{definition}

Note that $\vec{f}$ can be seen as a function $\vec{f}\colon L^m \to L^m$. In this view, pre-solutions are the fixpoints of $\vec{f}$.
Since all components $f_i$ are monotone, also $\vec{f}$ is monotone over $(L^m, \sqsubseteq)$. Then, it is well-known that the set of fixpoints of $\vec{f}$, i.e., the pre-solutions of the system, are a sublattice.
In order to define the solution of a system  we need some further notation.

\begin{definition}[substitution]
  \label{de:substitution}
  Given a system $E$ of $m$ equations over a lattice $L$ of the kind
  $\vec{x} =_{\vec{\eta}} \vec{f} (\vec{x})$, an index
  $i \in \interval{m}$ and $l \in L$ we write
  $\subst{E}{x_i}{l}$ for the system of $m-1$ equations obtained from
  $E$ by removing the $i$-th equation and replacing $x_i$ by $l$ in the other equations,
  i.e., if  $\vec{x} = \vec{x}' x_i \vec{x}''$,
  $\vec{\eta} = \vec{\eta}' \eta_i \vec{\eta}''$ and
  $\vec{f} = \vec{f}' f_i \vec{f}''$ then $\subst{E}{x_i}{l}$ is
  $\vec{x}' \vec{x}'' =_{\vec{\eta}', \vec{\eta}''} \vec{f}'
  \vec{f}''(\vec{x}', l, \vec{x}'')$.

  Let $\fsubst{f}{x_i}{l}\colon L^{m-1} \to L$ be defined by
  $\fsubst{f}{x_i}{l}(\vec{x}',\vec{x}'') = f(\vec{x}',l,\vec{x}'')$.
  Then, explicitly, the system $\subst{E}{x_i}{l}$ has $m-1$ equations, 
  \begin{center}
    $x_j =_{\eta_j} \fsubst{f_j}{x_i}{l}(\vec{x}',\vec{x}'')$ \qquad
    $j \in \{ 1, \ldots, i-1,i+1, \ldots, n \}$
  \end{center}
\end{definition}

We can now recursively define the solution of a system of
equations. The notion is the same as
in~\cite{hsc:lattice-progress-measures}, although we find it
convenient to adopt a more succinct formulation
(an explicit proof of the equivalence of the two
notions can be found in Appendix~\ref{ssec:comparison-solution-hsc}).

\begin{definition}[solution]
  \label{de:solution}
  Let $L$ be a lattice and let $E$ be a system of $m$ equations on
  $L$ of the kind $\vec{x} =_{\vec{\eta}} \vec{f}(\vec{x})$.
  The \emph{solution} of $E$, denoted $\sol{E} \in L^m$, is defined
  inductively as follows:
  \begin{center}
    $
    \begin{array}{lll}
      \sol{\emptyset} & = & () \\
      \sol{E} & = & (\sol{\subst{E}{x_m}{u_m}},\, u_m) \mbox{ where } u_m
                      = \eta_m (\lambda x.\, f_m(\sol{\subst{E}{x_m}{x}},x))
    \end{array}
    $
  \end{center}
  The $i$-th component of the solution will be denoted $\sol[i]{E}$.
\end{definition}

In words, for solving a system of $m$ equations, the last
variable is considered as a fixed parameter $x$ and the system of
$m-1$ equations that arises from dropping the last equation is
recursively solved. This produces an $(m-1)$-tuple parametric on $x$,
i.e., we get $\subvec{u}{1}{m-1}(x) = \sol{\subst{E}{x_m}{x}}$. Inserting
this parametric solution into the last equation, 
we get an equation in a single variable
\begin{center}
  $x  =_{\eta_m} f_m(\subvec{u}{1}{m-1}(x), x)$
\end{center}
that can be solved by taking for the 
function
$\lambda x.\, f_m(\subvec{u}{1}{m-1}(x), x)$, the least or greatest
fixpoint, depending on whether the last equation is a $\mu$- or
$\nu$-equation.
This provides the $m$-th component of the solution
$u_m = \eta_m (\lambda x.\, f_m(\subvec{u}{1}{m-1}(x), x))$. The
remaining components of the solution are obtained inserting $u_m$ in
the parametric solution $\subvec{u}{1}{m-1}(x)$ previously computed, i.e.,
$\subvec{u}{1}{m-1} = \subvec{u}{1}{m-1}(u_m)$.

The next  lemma will be helpful in several places. In
particular, it shows that the definition above is well-given, since we
are taking (least or greatest) fixpoints of monotone functions.

\begin{lemma}[solution is monotone]
  \label{le:solution-monotone}
  Let $E$ be a system of $m > 0$ equations of the
  kind $\vec{x} =_{\vec{\eta}} \vec{f} (\vec{x})$ over a lattice $L$.  For
  $i \in \interval{m}$ the function $g\colon L \to L^{m-1}$ defined
  by $g(x) = \sol{\subst{E}{x_i}{x}}$ is monotone.
\end{lemma}

\begin{proof}
  The proof proceeds by induction on $m$. The base case $m=1$ holds
  trivially since necessarily $i=1$ and for any $x \in L$, the system
  $\subst{E}{x_i}{x}$ is empty, with empty solution.

  Let us assume $m>1$. We distinguish two subcases according to
  whether $i=m$ or $i < m$. If $i=m$ then by definition of solution
  \begin{equation}
    \label{eq:g}
    g(x) = \sol{\subst{E}{x_m}{x}} =
    (\sol{\subst{\subst{E}{x_m}{x}}{x_{m-1}}{u_{m-1}(x)}}, u_{m-1}(x))
  \end{equation}
  where
  $u_{m-1}(x) = \eta_{m-1} (\lambda y.\,
  f_{m-1}(\sol{\subst{\subst{E}{x_m}{x}}{x_{m-1}}{y}},y,x)$.

  Next observe that the function $h\colon L^2 \to L^{m-2}$ defined by
  $h(x,y) = \sol{\subst{\subst{E}{x_m}{x}}{x_{m-1}}{y}}$ is
  monotone. In fact, it is monotone in $y$ by inductive hypothesis,
  and also in $x$, again by inductive hypothesis, since
  $\subst{\subst{E}{x_m}{x}}{x_{m-1}}{y} =
  \subst{\subst{E}{x_{m-1}}{y}}{x_m}{x}$.
  Observe that $u_{m-1}$ can be written as
  \begin{center}
    $u_{m-1}(x) =  \eta_{m-1}(\lambda y.\, f_{m-1}(h(x,y),y,x))$
  \end{center}
  Recalling that also $f_{m-1}$ is monotone, we deduce that $u_{m-1}$ is
  monotone.

  Finally, using the definition of $g$ and $u_{m-1}$, from
  (\ref{eq:g}) we can derive
  \begin{center}
    $g(x) = (h(x, u_{m-1}(x)), u_{m-1}(x))$
  \end{center}
  which allows us to conclude that $g$ is monotone.

  \medskip

  If instead, $i < m$, just note that
  \begin{equation}
    g(x) = \sol{\subst{E}{x_i}{x}} =
    (\sol{\subst{\subst{E}{x_i}{x}}{x_{m}}{u_{m}(x)}}, u_{m}(x))
  \end{equation}
  where
  $u_{m}(x) = \eta_{m} (\lambda y.f_m\,
  (\sol{\subst{\subst{E}{x_i}{x}}{x_{m}}{y}},y,x)$. Then the proof
  proceeds as in the previous case.
\end{proof}

It can be easily proved that the solution of a system is, as
anticipated, a special pre-solution.

\begin{lemma}[solution is pre-solution]
  \label{le:solution-is-pre}
  Let $E$ be a system of $m$ equations over a lattice $L$ of the kind
  $\vec{x} =_{\vec{\eta}} \vec{f} (\vec{x})$ and let $\vec{u}$ be its
  solution. Then $\vec{u}$ is a pre-solution, i.e.,
  $\vec{u} = \vec{f}(\vec{u})$.
\end{lemma}

\begin{proof}
  The proof proceeds by induction on $m$. The base case $m=0$
  trivially holds. For any $m > 0$, let $\vec{u} = \vec{u'} u_m$,
  $\vec{f} = \vec{f}' f_m$ and $\vec{x} = \vec{x}'x_m$. Since
  $\vec{u'} = \sol{\subst{E}{x_m}{u_m}}$, by inductive hypothesis, we
  have that
  \begin{equation}
    \label{eq:sol1}
    \vec{u'} = \subst{\vec{f'}}{x_m}{u_m}(\vec{u}') = \vec{f'}(\vec{u}).
  \end{equation}
  Moreover, again by definition of solution, we have that
  $u_m = \eta_m (\lambda x.\, f_m(\sol{\subst{E}{x_m}{x}}, x))$. Hence
  $u_m = f_m(\sol{\subst{E}{x_m}{u_m}}, u_m))$. Recalling that
  $\sol{\subst{E}{x_m}{u_m}} = \vec{u'}$ we deduce
  $u_m = f_m(\vec{u'}, u_m) = f_m(\vec{u})$, that together with
  (\ref{eq:sol1}) gives $\vec{u} = \vec{f}(\vec{u})$ as desired.
\end{proof}

\subsection{A Prototypical Example: the $\mu$-Calculus}
\label{ssec:mu-calc}

As a prototypical example, we discuss how $\mu$-calculus formulae can
be equivalently seen as systems of fixpoint equations.
We focus on a standard $\mu$-calculus syntax. For fixed disjoint sets
$\PVar$ of propositional variables, ranged over by $x, y, z, \ldots$
and $\Prop$ of propositional symbols, ranged over by $p,q,r, \ldots$,
formulae are defined by
\begin{center}
  $\varphi \ ::=\ \true\ \mid\ \false\ \mid\ p \mid\ x\ \mid\ \varphi
  \land \varphi\ \mid\ \varphi \lor \varphi\ \mid\ \Box \varphi\ \mid\
  \Diamond \varphi\ \mid\ \eta x.\, \varphi$
\end{center}
where $p \in \Prop$, $x \in \PVar$ and $\eta \in \{ \mu, \nu
\}$. Formulae of the kind $\eta x.\, \phi$ are called fixpoint
formulae.

The semantics of a formula is given with respect to an unlabelled
transitions system (or Kripke structure) $(\mathbb{S}, \to)$ where
$\mathbb{S}$ is the set of states and
$\to\ \subseteq \mathbb{S} \times \mathbb{S}$ is the transition
relation. Given a formula $\varphi$ and an environment
$\rho\colon \Prop \cup \PVar \to \Pow{\mathbb{S}}$ mapping each
proposition or propositional variable to the set of states where it
holds, we denote by $\sem{\varphi}_\rho$ the semantics of $\varphi$
defined as usual (see, e.g.,~\cite{bw:mu-calculus-modcheck}).

First note that any $\mu$-calculus formula can be expressed in
equational form, by inserting an equation for each propositional
variable (see also~\cite{cks:faster-modcheck-mu,s:fast-simple-nested-fixpoints}). The
reverse translation is also possible, hence these specification
languages are equally expressive. Here, we will only depict the
relation via an example, the formal treatment is given in
Appendix~\ref{sec:fp-systems-mu}.
\begin{figure}[h]
  \begin{subfigure}[b]{.3\textwidth}
    \centering
    $
    \begin{array}{lcl}
      x_1 & =_{\mu} & p \lor \Diamond x_1\\
      x_2 & =_{\nu} & x_1 \land \Box x_2
    \end{array}
    $
    \caption{}
    \label{fi:running-eqf}
  \end{subfigure}      
  \begin{subfigure}[b]{.3\textwidth}
    \centering      
    \begin{tikzpicture}[->, >=stealth, nodes={draw, circle, minimum size=1.8em}, node distance=0.5cm]
      \node (a) {$a$};
      \node (b) [label=above:$p$,right=of a] {$b$};
      \path
      (a) edge (b)
      (a) edge [loop left] ()
      (b) edge [loop right] ();
    \end{tikzpicture}
    \caption{}
    \label{fi:running-ts}
  \end{subfigure}
  \begin{subfigure}[b]{.3\textwidth}
    \centering            
    $
    \begin{array}{lcl}
      x_1 & =_{\mu} & \{ b\} \cup \Diamond x_1\\
      x_2 & =_{\nu} & x_1 \cap \Box x_2
    \end{array}
    $
    \caption{}
    \label{fi:running-eq}
  \end{subfigure}      
  \caption{}
\end{figure}

\begin{example}
  \label{ex:running}
  Let $\phi = \nu x_2.((\mu x_1.(p \lor\Diamond x_1))\land\Box x_2)$
  be a formula requiring that from all reachable states there exists a
  path that eventually reaches a state where $p$ holds. The equational
  form is quite straightforward and is reported in
  Fig.~\ref{fi:running-eqf}.
  Consider a transition system $(\mathbb{S}, \to)$ where
  $\mathbb{S} = \{a,b\}$ and $\to$ is as depicted in Fig.~\ref{fi:running-ts},
  with $p$ that holds only on state $b$.
  The resulting system of equations on the lattice $\Pow{\mathbb{S}}$
  is given in Fig.~\ref{fi:running-eq}, where
  $\Diamond,\Box\colon \Pow{\mathbb{S}}\to \Pow{\mathbb{S}}$ are
  defined as
  $\Diamond(S) = \{s\in \mathbb{S}\mid \exists s'\in \mathbb{S}.(s\to
  s' \land s'\in S)\}$,
  $\Box(S) = \{s\in \mathbb{S}\mid \forall s'\in \mathbb{S}.(s\to
  s'\Rightarrow s'\in S)\}$ for $S\subseteq \mathbb{S}$.

  The solution is $x_1 = x_2 = \mathbb{S}$.  In particular,
  $x_2 = \mathbb{S}$ corresponds to the fact that the formula
  $\varphi$ holds in every state.
\end{example}

\begin{example}
  Consider the formula $\phi' = \nu x_2.(\Box x_2 \land \mu x_1.((p
  \land \Diamond x_2) \lor \Diamond x_1))$ requiring that from all reachable
  states there is a path along which $p$ holds infinitely often.
  The equational form of $\phi'$ is:
  \begin{center}
    $
    \begin{array}{lcl}
      x_1 & =_{\mu} & (p \land \Diamond x_2) \lor \Diamond x_1\\
      x_2 & =_{\nu} & \Box x_2 \land x_1
    \end{array}
    $
  \end{center}
  On the same transition system of the previous example
  (Fig.~\ref{fi:running-ts}), the solution of the corresponding
  system is $x_1 = x_2 = \mathbb{S}$.  Notice that this time the
  order of the equations is relevant, while in the previous example it
  was not. Indeed, if we swap the two equations in the system, the
  solution becomes $x_1 = x_2 = \emptyset$.  In general, the
  order of the equations is important whenever there is alternation of
  fixpoints (mutual dependencies between least and greatest fixpoint
  equations).  
\end{example}

\subsection{Data-Flow Analysis}
\label{ss:data-flow}

In order to give further intuition, we revisit another area
where fixpoints play a major role, namely data-flow analysis of
programs. One can easily state
a program analysis question
in this setting as a system of fixpoint equations, based on the flow
graph of the program under consideration.

We focus on the well-known \emph{constant propagation
analysis} (see, e.g.,~\cite{nnh:program-analysis}).
Its aim
is to show that the value
of a variable is always constant at a certain program point, allowing
us to optimise the program by replacing the variable by the constant.
Consider for instance the while program in
Fig.~\ref{fig:constant-propagation}, where variables contain integer
values and blocks are numbered in order to easily reference them. The
condition for the while loop (block~$3$) is irrelevant and is hence
replaced by $\prg{*}$. Note that variable $\prg{x}$ always has
value~$7$ in block~$4$ and hence the assignment in this block
could be replaced by $\prg{y:=7+y}$.

\begin{figure}[h]
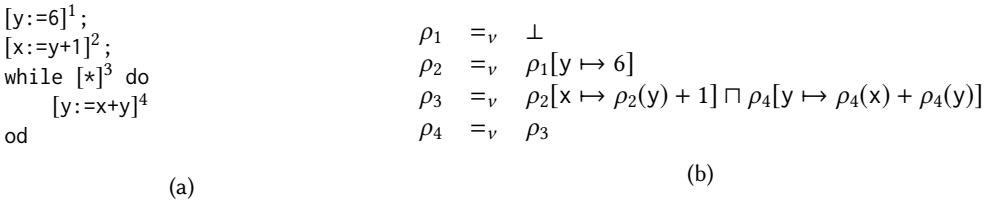

  \begin{subfigure}[b]{.34\textwidth}
    {\tt\small
      \begin{mytab}
        \blk{1}{y:=6}; \\
        \blk{2}{x:=y+1}; \\
        while \blk{3}{*} do \\
        \> \blk{4}{y:=x+y} \\
        od
      \end{mytab}}
    \caption{}
  \label{fig:constant-propagation}
\end{subfigure}      
\begin{subfigure}[b]{.64\textwidth}
  \centering
  $
  \begin{array}{lcl}
    \rho_1 & =_\nu & \bot \\
    \rho_2 & =_\nu & \rho_1[\prg{y}\mapsto 6] \\
    \rho_3 & =_\nu & \rho_2[\prg{x}\mapsto \rho_2(\prg{y})+1] \sqcap
                     \rho_4[\prg{y}\mapsto \rho_4(\prg{x})+\rho_4(\prg{y})] \\
    \rho_4 & =_\nu & \rho_3
  \end{array}
  $ 
  \caption{}
  \label{fig:constant-analysis}
\end{subfigure}      
\caption{(a) A simple while program and (b) the equation system for
  the corresponding constant propagation analysis.}
\end{figure}

Following~\cite{nnh:program-analysis} we analyse such programs by
setting up an instance of a monotone framework. In particular we will
use the following lattice to record the results of the analysis:
\[ L = (\mathbb{Z}\cup \{\bot\})^\mathit{Var}\cup \{\top\} \] where
$\mathit{Var}$ is the set of variables.
That is, a lattice element is
either $\top$ or a function
$\rho\colon \mathit{Var}\to \mathbb{Z}\cup \{\bot\}$ that assigns a
variable $\prg{x}$ to a value in $\mathbb{Z}$ (if $\prg{x}$ is known
to have constant value $\rho(\prg{x})$ at this program point) or to
$\bot$ (to indicate that $\prg{x}$ is possibly non-constant). As
usual, we are allowed to over-approximate and $\bot$ might be
assigned although the value of the variable is actually constant.

The lattice order is defined as follows: two assignments
$\rho_1,\rho_2\colon \mathit{Var}\to \mathbb{Z}\cup \{\bot\}$ are
ordered, i.e.\ $\rho_1\sqsubseteq \rho_2$, if for each
$\prg{x}\in \mathit{Var}$ either $\rho_1(\prg{x}) = \rho_2(\prg{x})$
or $\rho_1(\prg{x}) = \bot$. That is, we consider a flat order where
$\bot$ is smaller than the integers and the integers themselves are
incomparable, and extend it pointwise to functions. Clearly,  $\top$ is
the largest lattice element and we use some overloading
and denote by $\bot$ the function that maps every variable to
$\bot$. Note that this order deviates from the usual convention in
program analysis which states that smaller values should be more
precise than larger values. We do this since our game characterises
whether a lattice element is \emph{below} the solution. Since we want
to check that the solution is more precise than a given threshold, we
have to reverse the order.

Let us write $\rho' = \rho[\prg{x}\mapsto z]$ for $z\in\mathbb{Z}$ to
denote function update, that is $\rho'(\prg{x}) = z$ and
$\rho'(\prg{y}) = \rho(\prg{y})$ for $\prg{y}\neq\prg{x}$. When
$\rho=\top$ we define $\rho[\prg{x}\mapsto z] = \top$. 

Observe that $L$ is algebraic (and hence continuous).  The compact
elements are $\top$ and those functions which have finite support,
i.e., functions of the kind
$\bot[\prg{x1} \mapsto z_1, \ldots, \prg{xn} \mapsto z_n]$ where
only finitely many variables are not mapped to $\top$. In particular
we can use as a basis the functions $\bot [\prg{x} \mapsto z]$ for some $\prg{x} \in \mathit{Var}$ and $z \in \mathbb{Z}$.
Note also that $L$ is not distributive. For instance if $\rho_i = \bot[\prg{x} \mapsto i]$ for $i \in \{1,2,3\}$ then $(\rho_1 \sqcap \rho_2) \sqcup \rho_3 = \bot \sqcup \rho_3 = \rho_3$ while $(\rho_1 \sqcup \rho_3) \sqcap (\rho_2 \sqcup \rho_3)  = \top \sqcap \top = \top$.

From the program in Fig.~\ref{fig:constant-propagation} we can
easily derive the system of fixpoint equations in Fig.~\ref{fig:constant-analysis}, where we use $\rho_i$
to record the lattice value for the entry of block~$i$. 

At the beginning, no variable is constant. Then the equation system
mimics the control flow of the program. In block~$3$ we have to take
the meet to obtain an analysis result that is less precise than the
results coming from block~$2$ respectively block~$4$. Furthermore,
since precision increases with the order, we are interested in the
greatest solution, which means that we have only $\nu$-equations.

The expected solution is $\rho_1 = \bot$,
$\rho_2 = \bot[\prg{y}\mapsto 6]$,
$\rho_3 = \rho_4 = \bot[\prg{x}\mapsto 7]$ witnessing that at block 2 variable \prg{y} has constant value $6$ and at blocks 3 and 4 variable \prg{x} has constant value $7$.

\subsection{Approximating the Solution}

The game theoretical characterisation of the solution of a system of
fixpoint equations discussed later will rely on a notion of
approximation of the solution that is reminiscent of the lattice
progress measure in~\cite{hsc:lattice-progress-measures}.

\begin{definition}[approximants]
  \label{de:approximants}
  Let $E$ be a system of $m$ equations over the lattice $L$ of the
  kind $\vec{x} =_{\vec{\eta}} \vec{f}(\vec{x})$. Given any tuple
  $\vec{l} \in L^m$, let $f_{i,{\vec{l}}}\colon L \to L$ be the function
  defined by
  \begin{center}
    $f_{i,{\vec{l}}}(x) = f_i(\sol{\subst{\subst{E}{\subvec{x}{i+1}{m}}{\subvec{l}{i+1}{m}}}{x_i}{x}}, x, \subvec{l}{i+1}{m})$.
  \end{center}
  We say that a tuple  $\vec{l} \in L^m$ is a $\mu$-\emph{approximant}
  when for all $i \in \interval{m}$,
  if $\eta_i = \nu$ then $l_i = \nu (f_{i,\vec{l}})$, else
  if $\eta_i = \mu$ then
    $l_i = f_{i,{\vec{l}}}^\alpha(\bot)$ for some ordinal $\alpha$.
  Dually,  $\vec{l} \in L^m$ is a $\nu$-\emph{approximant} when for all $i \in \interval{m}$,
  if $\eta_i = \nu$ then $l_i = f_{i,{\vec{l}}}^\alpha(\top)$  for some ordinal $\alpha$,
  else if $\eta_i = \mu$ then
    $l_i =
    \mu (f_{i,\vec{l}})$.

  Whenever $\vec{l}$ is a $\mu$-approximant we write $\ord{\vec{l}}$
  to denote the least $m$-tuple of ordinals $\vec{\alpha}$ such that
  for any $i \in \interval{m}$, if $\eta_i = \mu$ then
  $l_{i} = f_{i,{\vec{l}}}^{\alpha_i}(\bot)$ else, if $\eta_i = \nu$,
  $l_{i} = f_{i,{\vec{l}}}^{\alpha_i}(\top) = \nu (f_{i,{\vec{l}}})$.
\end{definition}

Observe that, spelling out the definition of the solution of a system
of equations, it can be easily seen that
$\sol[i]{\subst{E}{\subvec{x}{i+1}{m}}{\subvec{l}{i+1}{m}}} = \eta_i
(f_{i,{\vec{l}}})$.
Then a $\mu$-approximant is obtained by taking
under-approximations for the least fixpoints and the exact value for
greatest fixpoints.
In fact, in the case of $\mu$-approximants, for each
$i \in \interval{m}$, if $\eta_i=\nu$, the $i$-th component is set to
$\nu (f_{i,{\vec{l}}})$ which is $i$-th component
$\sol[i]{\subst{E}{\subvec{x}{i+1}{m}}{\subvec{l}{i+1}{m}}}$ of the
solution. Instead, if $\eta_i=\mu$ the component $l_i$ is set to
$f_{i,{\vec{l}}}^\alpha(\bot)$ for some ordinal $\alpha$, which is an
underapproximation of
$\mu(f_{i,{\vec{l}}}) =
\sol[i]{\subst{E}{\subvec{x}{i+1}{m}}{\subvec{l}{i+1}{m}}}$, obtained
by iterating $f_{i,{\vec{l}}}$ over $\bot$ up to
ordinal $\alpha$. For $\nu$-approximants the situation is dual.

We remark that the function $f_{i,\vec{l}}$ depends only on the
subvector $\subvec{l}{i+1}{m}$. In particular
$f_{m,\vec{l}}$ does not depend on $\vec{l}$. In fact, $f_{m,\vec{l}}  = \lambda x.\, f_m(\sol{\subst{E}{x_m}{x}},x)$. Using
$\vec{l}$ as subscript instead of the subvector is a slight abuse of
notation that makes the notation lighter.

Approximants can be given an inductive characterisation. Besides
shedding some light on the notion of approximant, the following easy
result will be useful at a technical level.

\begin{lemma}[inductive characterisation of approximants]
  \label{le:inductive-approximant}
  Let $E$ be a system of $m>0$ equations over the lattice $L$, of the
  kind $\vec{x} =_{\vec{\eta}} \vec{f}(\vec{x})$ and let $g_m\colon L \to L$ be the function $g_m(x) = f_m(\sol{\subst{E}{x_m}{x}},x)$.  A tuple $\vec{l} \in L^m$ is a $\mu$-approximant iff the following conditions hold
  \begin{enumerate}
  \item either $\eta_m=\mu$ and $l_m = g_m^\alpha(\bot)$ for some ordinal $\alpha$, or $\eta_m=\nu$ and $l_m = \nu g_m$
    
  \item 
    $\subvec{l}{1}{m-1}$ is a $\mu$-approximant of $\subst{E}{x_m}{l_m}$.
  \end{enumerate}
\end{lemma}

\begin{proof}
  Immediate.
\end{proof}

As mentioned above, $\mu$-approximants are closely related to lattice
progress measures in the sense
of~\cite{hsc:lattice-progress-measures}.
In fact, from Lemma~\ref{le:inductive-approximant} we can infer that, given a vector
$\vec{\alpha}$ of ordinals,  the $\mu$- or $\nu$-approximant
$\vec{l}\in L^m$ with $\ord{\vec{l}} = \vec{\alpha}$ is uniquely
determined.
More precisely, a $\mu$-approximant $\vec{l}$ is determined by the
subvector of $\ord{\vec{l}}$ consisting only of the $m$-tuple of
components of $\ord{\vec{l}}$ corresponding to $\mu$-indices. Hence we
can define a function that maps each such $m$-tuple of ordinals to the
corresponding $\mu$-approximant and this turns out to be a lattice
progress measures in the sense
of~\cite{hsc:lattice-progress-measures}.  Actually, as proved in
Appendix~\ref{ssec:hsc-measures}, it is the
greatest one.  It can be seen to coincide with the measure used
in~\cite[Theorem 2.13]{hsc:lattice-progress-measures} (completeness
part).

We next observe that the name approximant is appropriate, i.e.,
$\mu$-approximants provide an approximation of the solution from
below, while $\nu$-approximants from above.
The solution is thus the only pre-solution which is both a $\mu$- and a $\nu$-approximant.

\begin{lemma}[solution and approximants]
  \label{le:solution-is-approximant}
  Let $E$ be a system of $m$ equations over the lattice $L$, of the
  kind $\vec{x} =_{\vec{\eta}} \vec{f}(\vec{x})$.  The solution of 
  $E$ is the greatest $\mu$-approximant and the least
  $\nu$-approximant.
\end{lemma}

\begin{proof}
  The solution $\vec{u}$ is clearly a $\mu$-approximant. We prove that it is the greatest
  $\mu$-approximant by induction on $m$.
  If $m=0$ the thesis is vacuously true. If $m>0$, consider another
  $\mu$-approximant $\vec{l}$. We distinguish two subcases according
  to whether $\eta_m=\mu$ or $\eta_m=\nu$. If $\eta_m = \mu$, we know
  that $l_m = f_{m,\vec{l}}^\alpha(\bot)$ for some ordinal $\alpha$.
  Observe that
  $f_{m,\vec{l}} = \lambda x.\, f_m(\sol{\subst{E}{x_m}{x}}, x))$ is
  the function for which $u_m$ is the least fixpoint, hence
  \begin{equation}
    \label{eq:smallest1}
    l_m \sqsubseteq u_m.
  \end{equation}
  Moreover, by Lemma~\ref{le:inductive-approximant}, $\subvec{l}{1}{m-1}$ is a
  $\mu$-approximant for the system $\subst{E}{x_m}{l_m}$.
  Hence, by inductive hypothesis
  \begin{equation}
    \label{eq:smallest2}
    \subvec{l}{1}{m-1} \sqsubseteq
    \sol{\subst{E}{x_m}{l_m}}
  \end{equation}
  Moreover, by monotonicity of the solution
  (Lemma~\ref{le:solution-monotone}), since $l_m \sqsubseteq u_m$, we
  get
  $\sol{\subst{E}{x_m}{l_m}} \sqsubseteq \sol{\subst{E}{x_m}{u_m}} =
  \subvec{u}{1}{m-1}$. Therefore, combined with (\ref{eq:smallest1})
  and (\ref{eq:smallest2}), we conclude $\vec{l} \sqsubseteq \vec{u}$.
  
  \medskip
  
  The proof for $\nu$-approximants is dual.  
\end{proof}

We conclude with a technical lemma that will be used to locally modify approximations in the game.

\begin{lemma}[updating approximants]
  \label{le:approx-update}
  Let $E$ be a system of $m$ equations over the lattice $L$, of the
  kind $\vec{x} =_{\vec{\eta}} \vec{f}(\vec{x})$ and let $\vec{l}$ be
  a $\mu$-approximant with $\ord{\vec{l}} = \vec{\alpha}$. For any
  $i \in \interval{m}$ and ordinal $\alpha \leq \alpha_i$
  \begin{enumerate}
  \item if $\eta_i = \mu$, then
    $\vec{l}'
    =(\sol{\subst{\subst{E}{\subvec{x}{i+1}{m}}{\subvec{l}{i+1}{m}}}{x_i}{l_i'}},
    l_i', \subvec{l}{i+1}{m})$, with
    $l_i' = f_{i,\vec{l}}^\alpha(\bot)$ for some ordinal $\alpha$, is
    a $\mu$-approximant
    
  \item if $\eta_i = \nu$, then
    $\vec{l}' = (\sol{\subst{E}{\subvec{x}{i+1}{m}}{\subvec{l}{i+1}{m}}},
    \subvec{l}{i+1}{m})$ is
    a $\mu$-approximant
  \end{enumerate}
  and in both cases $\ord{\vec{l}'} \preceq_i \ord{\vec{l}}$.
  A dual result holds for $\nu$-approximants.
\end{lemma}

\begin{proof}
  Let us focus on (1). In order to show that $\vec{l}' = (\sol{\subst{\subst{E}{\subvec{x}{i+1}{m}}{\subvec{l}{i+1}{m}}}{x_i}{l_i'}},
    l_i', \subvec{l}{i+1}{m})$ is a
      $\mu$-approximant, first observe that the components
      $l_{i+1}, \ldots, l_m$ do not change. Component $l_i'$ is of the
      desired shape by definition. Finally, for
      $j<i$ the component $l_j'$ is defined as \linebreak
      $\sol[j]{
        \subst{\subst{E}{\subvec{x}{i+1}{m}}{\subvec{l}{i+1}{m}}}{x_i}{l_i'}}$
      and thus, by definition of solution of a system, if $\eta_j = \nu$ then
      $l_j' = \nu (f_{j,\vec{l}'})$  and if $\eta_j = \mu$ then
      $l_j' = \mu (f_{j,\vec{l}'}) = f_{j,\vec{l}'}^\beta(\bot)$ for
      some ordinal $\beta$, as desired.
      Finally observe that since $\vec{l}$ and $\vec{l}'$
      coincide on components $i+1, \ldots, m$, and
      $l_i = f_{i,\vec{l}}^{\alpha_i}(\bot)$, while
      $l_i' = f_{i,\vec{l}}^\alpha(\bot)$, with
      $\alpha \leq \alpha_i$, clearly
      $\ord{\vec{l}'} \preceq_i \ord{\vec{l}}$.

      \medskip
      
      The proof for (2) is analogous. In fact, also in this case the
      components ${i+1} \ldots, m$ are unchanged and finally, for
      $j\leq i$ the component $l_j'$ is defined as
      $\sol[j]{
        \subst{\subst{E}{\subvec{x}{i+1}{m}}{\subvec{l}{i+1}{m}}}{x_i}{l_i'}}$,
      thus the same reasoning as above applies.
      Both can easily dualised for $\nu$-approximants.
\end{proof}

\section{Fixpoint Games}
\label{sec:fp-equation-games}

In this section we present a game-theoretical approach to the solution
of a system of fixpoint equations over a continuous lattice. More
precisely, given a lattice with a fixed basis, the game allows us to
check whether an element of the basis is smaller (with respect to
$\sqsubseteq$) than the solution of a selected equation. This
  corresponds to solving the associated verification problem. For
  instance, when model-checking the $\mu$-calculus, one is interested
  in establishing whether a system satisfies a formula $\varphi$,
  which amounts to check whether $\{s_0\} \subseteq u_\varphi$ where
  $s_0$ is the initial state and $u_\varphi$ is the solution of the
  system of equations associated with $\varphi$.

\subsection{Definition of the Game}

The fixpoint game that we introduce has been inspired by the unfolding
game described in~\cite{v:lectures-mu-calculus}, that works for a single
fixpoint equation over the powerset lattice. We adopted the name
\textit{fixpoint game}, analogously to~\cite{hkmv:parity-games-automata-game-logic}.

\begin{definition}[fixpoint game]
  \label{def:fp-game}
  Let $L$ be a continuous lattice and let $B_L$ be a basis of $L$ such
  that $\bot \not \in B_L$. Given a system $E$ of $m$ equations over
  $L$ of the kind $\vec{x} =_{\vec{\eta}} \vec{f} (\vec{x})$, the
  corresponding fixpoint game is a parity game, with an existential
  player $\exists$ and a universal player $\forall$, defined as
  follows:
  
  \begin{itemize}
  \item The positions of $\exists$ are pairs $(b, i)$ where $b \in B_L$
    and $i \in \interval{m}$ and those of $\forall$ are
    tuples $\vec{l} \in L^m$.

  \item From $(b, i)$ the possible moves of $\exists$ are 
    $\Emoves{b,i} = \{ \vec{l} \mid \vec{l} \in L^m\ \land\ 
    b \sqsubseteq f_i(\vec{l})\}$.
    
  \item From $\vec{l} \in L^m$ the possible moves of $\forall$ are
    $\Amoves{\vec{l}} = \{ (b, i) \mid i \in \interval{m}\ \land\ b
    \in B_L\ \land\ b \ll l_i \}$.
  \end{itemize}
  The game is schematised in Table~\ref{tab:game}.
  For a finite play, the winner is the player whose opponent is unable
  to move. For an infinite play, let $h$ be the highest index that
  occurs infinitely often in a pair $(b, i)$. If $\eta_h = \nu$ then
  $\exists$ wins, else $\forall$ wins.
\end{definition}

\begin{table}[h]
  \begin{center}
    \begin{tabular}{l|c|l}
      Position & Player & Moves \\ \hline 
      $(b,i)$ & $\exists$ & $(l_1,\dots,l_m)$  such that $b \sqsubseteq f_i(l_1,\dots,l_m)$ \\[1mm]
      $(l_1,\dots,l_m)$ & $\forall$ & $(b',j)$ such that $b' \ll l_j$
    \end{tabular}
  \end{center}
  \caption{The fixpoint game}
  \label{tab:game}
\end{table}

Observe that the fixpoint game is a parity
game~\cite{ej:tree-automata-mu-determinacy,Zie:IGFCG} (on an infinite
graph) and the winning condition is the natural formulation of the
standard winning condition in this setting.

Hereafter, whenever we consider a continuous lattice $L$, we assume
that a basis $B_L$ is fixed such that $\bot \not \in B_L$. Elements of
the basis will be denoted by letters $b$ with super or subscripts.

We will prove correctness and completeness of the game, i.e., we will show
that if $\vec{u}$ is the solution of the system, given a basis element
$b \in B_L$ and $i \in \interval{m}$, if $b \sqsubseteq u_i$ then
starting from $(b,i)$ the existential player has a winning strategy,
otherwise the universal player has a winning strategy.

\begin{example}
  \label{ex:running-a}
  As an example, consider the equation system of
  Example~\ref{ex:running}, as depicted in Fig.~\ref{fi:running-eq},
  corresponding to the $\mu$-calculus formula
  $\phi = \nu x_2.((\mu x_1.(p \lor\Diamond x_1))\land\Box x_2)$.
  Recall that the lattice is $(\Pow{\mathbb{S}}, \subseteq)$ and let us
  fix as a basis the set of singletons
  $B_{\Pow{\mathbb{S}}} = \{ \{a\},\{b\}\}$.

  A portion of the fixpoint game is graphically represented as a
  parity game in Fig.~\ref{fi:game}. Diamond nodes correspond to
  positions of player $\exists$ and the box nodes to positions of
  player $\forall$. Only a subset of the possible positions for
  $\forall$ are represented. The positions which are missing, such as
  $(\{a,b\}, \{a,b\})$, can be shown to be redundant, in a sense
  formalised later (see \S~\ref{ssec:selections}), so that the subgame is
  equivalent to the full game.
  Numbers in the diamond nodes correspond to priorities. Box nodes do
  not have priorities (or we can assume priority~$0$). Since index $1$
  and $2$ corresponds to a $\mu$ and a $\nu$ equation, respectively, in
  this specific case the winning condition for player $\exists$ is
  exactly the same as for parity games: either the play is finite and
  $\exists$ plays last or the play is infinite and the highest priority that
  occurs infinitely often is even (in this case $2$).

  Let $(u_1, u_2)~$ be the solution of the system. We can check if
  $a \in u_2$, i.e., if $a$ satisfies $\varphi$, by playing the game
  from the position $(\{a\},2)$. In fact, $\{a\} \sqsubseteq u_2$ amounts to
  $a \in u_2$.
  Indeed player $\exists$ has a winning strategy that we can represent
  as a function $\varsigma$ from the positions of the game (for any
  play) to the corresponding moves of player $\exists$, i.e.,
  $\varsigma\colon B_{\Pow{\mathbb{S}}} \times \interval{2} \to
  \Pow{\mathbb{S}} \times \Pow{\mathbb{S}}$.  A winning strategy for
  $\exists$ is given by $\varsigma(\{a\},1)= (\{b\},\emptyset)$,
  $\varsigma(\{a\},2) = (\{a\},\{a,b\})$,
  $\varsigma(\{b\},1) = (\emptyset,\emptyset)$ and
  $\varsigma(\{b\},2) = (\{b\},\{b\})$. In Fig.~\ref{fi:game} we
  depict by bold arrows the choices prescribed by the strategy.

  A possible
  play of the game could be the following, where
  $\stackrel{x}{\leadsto}$ denotes a move of $x\in\{\exists,\forall\}$:
  \[ (\{a\},2) \stackrel{\exists}{\leadsto} (\{a\},\{a,b\})
    \stackrel{\forall}{\leadsto} (\{a\},1)
    \stackrel{\exists}{\leadsto} (\{b\},\emptyset)
    \stackrel{\forall}{\leadsto} (\{b\},1)
    \stackrel{\exists}{\leadsto} (\emptyset,\emptyset)
    \stackrel{\forall}{\not\leadsto}, \] hence $\exists$ wins. Another
  (infinite) play is the following. It is also won by $\exists$ since
  the highest index that occurs infinitely often is $2$, which is a $\nu$-index:
  \[ (\{a\},2) \stackrel{\exists}{\leadsto} (\{a\},\{a,b\})
    \stackrel{\forall}{\leadsto} (\{a\},2)
    \stackrel{\exists}{\leadsto} (\{a\},\{a,b\})
    \stackrel{\forall}{\leadsto} \dots \]
  Note that if $\exists$ always plays as specified by $\varsigma$, she
  will always win.
  
  \begin{figure}
  \scalebox{0.8}{
  \begin{tikzpicture}[node distance=9mm,>=stealth',y=8mm, x=8mm]
    \node (a2) at (0,2) [anode]{$(\{a\},2)$};
    \node (a2-a) at (3,2) [enode]{$(\{a\},\{a,b\})$};
    \node (a1) at (3,0) [anode]{$(\{a\},1)$};
    \node (b2) at (7,2) [anode]{$(\{b\},2)$};
    \node (a1-a) at (0,0) [enode]{$(\{a\},\emptyset)$};
    \node (a1-b) at (7,0) [enode]{$(\{b\},\emptyset)$};
    \node (b2-a) at (10,2) [enode]{$(\{b\},\{b\})$};
    \node (b1) at (10,0) [anode]{$(\{b\},1)$};
    \node (b1-a) at (13,0) [enode]{$(\emptyset,\emptyset)$};

    \draw  [->, very thick, bend left] (a2) to node {} (a2-a);
    \draw  [->, bend left] (a2-a) to node {} (a2);
    \draw  [->] (a2-a) to node {} (a1);
    \draw  [->] (a2-a) to node {} (b2);
    \draw  [->, bend left] (a1) to node {} (a1-a);
    \draw  [->,very thick] (a1) to node {} (a1-b);
    \draw  [->, bend left] (a1-a) to node {} (a1);
    \draw  [->] (a1-b) to node {} (b1);
    \draw  [->, very thick, bend left] (b2) to node {} (b2-a);
    \draw  [->, bend left] (b2-a) to node {} (b2);
    \draw  [->, very thick] (b1) to node {} (b1-a);
    \draw  [->] (b2-a) to node {} (b1);
  \end{tikzpicture}
}
\caption{Graphical representation of a fixpoint game}
  \label{fi:game}
\end{figure}
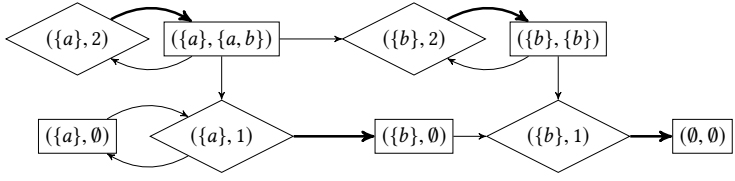
\end{example}

\subsection{Correctness and Completeness}

  Before proving correctness and completeness of the game in the general
  case, as a warm up, we give some intuition and outline the proof for
  the case of a single equation.
  Let $f\colon L\to L$ be a monotone function on a continuous lattice
  $L$ and consider the equation $x =_\eta f(x)$, where
  $\eta\in\{\nu,\mu\}$, with solution $u = \eta f$. In this case the
  positions for $\exists$ are simply basis elements $b\in B_L$ and
  $\exists$ must choose $l\in L$ such that $b\sqsubseteq
  f(l)$. Positions of $\forall$ are lattice elements $l \in L$ and
  moves are elements of the basis $b \in B_L$, with $b \ll l$.
  In the case of $\eta=\mu$, player $\forall$ wins infinite plays and
  in the case of $\eta=\nu$, player $\exists$ wins infinite plays.
    
  When $\eta = \mu$, if $b \sqsubseteq u$, then
  $b \sqsubseteq f^\alpha(\bot)$ for some ordinal $\alpha$. The idea
  is that $\exists$ can win by descending the chain
  $f^\beta(\bot)$. E.g., if $\beta = \gamma+1$ is a successor
  ordinal, then $\exists$ can play $f^{\gamma}(\bot)$.
  If instead, $\eta = \nu$, then the existential player can win just
  by identifying some post-fixpoint $l$ such that $b \sqsubseteq
  l$. In fact, if $l$ is a post-fixpoint, i.e., $l \sqsubseteq f(l)$ we
  know that $l \sqsubseteq u$. Moreover, if $b \sqsubseteq l$ then
  $b \sqsubseteq f(l)$ and thus $\exists$ can cycle on $l$ and win.
  More formally:

  \smallskip
  
  \subsubsection*{(Case $\eta = \mu$)}
  In this case $u = f^\alpha(\bot)$ for some
  ordinal $\alpha$.
  \begin{itemize}
  \item \emph{Completeness:} We show that whenever
    $b\sqsubseteq f^\beta(\bot)$, for some ordinal $\beta$ (i.e.,
    $b$ is below some $\mu$-approximant), then $\exists$ has a
    winning strategy, by transfinite induction on $\beta$.
    First observe that $\beta > 0$. In fact, otherwise
    $b \sqsubseteq f^0(\bot) = \bot$, hence $b = \bot$, while
    $\bot \not\in B_L$ by hypothesis.
    Hence we have two possibilities:
    \begin{itemize}
    \item If $\beta$ is a limit ordinal, player $\exists$ plays
      $l = f^\beta(\bot)$, which is a post-fixpoint and hence
      $b\sqsubseteq f^\beta(\bot)\sqsubseteq f(f^\beta(\bot))$. Then
      $\forall$ chooses
      $b' \ll f^\beta(\bot) = \bigsqcup_{\gamma < \beta}
      f^\gamma(\bot)$. Since this is a directed join, by
        definition of the way-below relation there exists
      $\gamma < \beta$ with $b'\sqsubseteq f^\gamma(\bot)$.
    \item If $\beta = \gamma+1$, $\exists$ plays
      $l = f^\gamma(\bot)$ and $\forall$ chooses
      $b'\ll f^\gamma(\bot)$, hence $b'\sqsubseteq f^\gamma(\bot)$.
    \end{itemize}
    Note that $\exists$ always has a move and the answer of
    $\forall$ is some $b'\sqsubseteq f^\gamma(\bot)$,
    with $\gamma < \beta$, from which there exists a winning
    strategy for $\exists$ by the inductive hypothesis.
    
  \item \emph{Correctness:} We show that whenever
    $b\not\sqsubseteq u$, player $\forall$ has a winning strategy.
    
    Observe that a move of $\exists$ will be some $l$ such that
    $b\sqsubseteq f(l)$. Note that there must be a $b' \ll l$ with
    $b'\not\sqsubseteq u$. In fact, otherwise, if for all $b'\ll l$
    it holds that $b'\sqsubseteq u$, since $L$ is a continuous
      lattice, we would have
    $l = \bigsqcup \{b'\mid b'\ll l\} \sqsubseteq u$ and furthermore
    $b\sqsubseteq f(l) \sqsubseteq f(u) = u$, which is a
    contradiction.
    
    Hence $\forall$ can choose such a $b' \ll l$ with
    $b'\not\sqsubseteq u$ and the game can continue. Then either
    $\exists$ runs out of moves at some point or we end up in an
    infinite play. In both cases $\forall$ wins.
  \end{itemize}
  
  \subsubsection*{(Case $\eta = \nu$)}
  In this case $u = f^\alpha(\top)$ for some
  ordinal $\alpha$.
  \begin{itemize}
  \item \emph{Completeness:} We show that when $b\sqsubseteq u$,
    then $\exists$ has a winning strategy.
    In fact, in this case $\exists$ simply plays $l = u$, which
    satisfies $b \sqsubseteq u = f(u)$ and $\forall$ answers with
    some $b\ll u$, hence $b\sqsubseteq u$. The game can thus
    continue forever, leading to an infinite play which is won by
    $\exists$.
    
  \item \emph{Correctness:} We show that whenever
    $b\not\sqsubseteq f^\beta(\top)$, for some ordinal $\beta$
    (i.e., $b$ is not below some $\nu$-approximant), then $\forall$
    has a winning strategy, by transfinite induction on $\beta$.
    First observe
    that $\beta > 0$. In fact, otherwise
    $b \not\sqsubseteq f^0(\top) = \top$ would be a contradiction.
    Hence we distinguish two cases:
    \begin{itemize}
    \item If $\beta$ is a limit ordinal
      $b\not\sqsubseteq f^\beta(\top) = \bigsqcap_{\gamma<\beta}
      f^\gamma(\top)$, which means that there exists $\gamma<\beta$
      such that $b\not\sqsubseteq f^\gamma(\top)$.
      
      Now any move of $\exists$ is some $l$ with
      $b\sqsubseteq f(l)$. Therefore
      $l\not\sqsubseteq f^\gamma(\top)$, since otherwise
      $b\sqsubseteq f(l)\sqsubseteq f(f^\gamma(\top)) =
      f^{\gamma+1}(\top) \sqsubseteq f^\beta(\top)$ (since
      $\gamma+1<\beta$). Hence there must be $b'\ll l$ with
      $b'\not\sqsubseteq f^\gamma(\top)$. Otherwise, as above, if for all
        $b' \ll l$ we had $b' \sqsubseteq f^\gamma(\top)$, then by
        continuity of the lattice, we would conclude
        $l = \bigsqcup \{b' \mid b' \ll l\} \sqsubseteq
        f^\gamma(\top)$. Such a $b'$ can be chosen by $\forall$,
      and the game continues.
    \item If $\beta = \gamma+1$ we know that
      $b\not\sqsubseteq f^\beta(\top) = f(f^\gamma(\top))$.
      
      Any move of $\exists$ is $l$ with $b\sqsubseteq f(l)$, which as
      above implies that $l\not\sqsubseteq f^\gamma(\top)$ and thus
      the existence of $b'\ll l$ with
      $b'\not\sqsubseteq f^\gamma(\top)$. The basis element $b'$ is
      chosen by $\forall$ and the game continues.
    \end{itemize}
    Hence $\forall$ always has a move, ending up in
    $b'\not\sqsubseteq f^\gamma(\top)$, from which there exists a
    winning strategy for $\forall$ by the induction hypothesis.
  \end{itemize}

  Observe that cases of a $\mu$- and a $\nu$-equation are not
  completely symmetric. In the completeness part, for showing that $l\sqsubseteq \nu f$ we use
  the fact that $\nu f$ is the greatest post-fixpoint. Instead, for
  showing that $l\sqsubseteq \mu f$ we use the fact that
  $l\sqsubseteq f^\alpha(\bot)$ for some $\alpha$ and provide a proof
  that we can descend to $\bot$, similarly to what happens for ranking functions in
  termination analysis. 
  Note that in order to guarantee that we truly descend, also below
  limit ordinals, we require that $\forall$ plays $b$ with $b \ll
  l$. Then we can use the fact that whenever $b$ is way-below a
  directed join, then it is smaller than one of the elements over
  which the join is taken.  We remark that choosing $b$ with
  $b\sqsubseteq l$ instead would not be sufficient (see
  Proposition~\ref{pr:counter-way-below}).
  In the correctness part, despite the asymmetry, both proofs use the
  fact that each element is the join of all elements way-below it, for
  which it is essential to be in a continuous lattice (see
  Proposition~\ref{pr:counter-continuity}). Instead, for completeness,
  the continuity hypothesis does not play a role.

\medskip

For the general case, correctness and completeness of the game are proved by relying on the notions of $\mu$- and $\nu$-approximant.
We prove the two properties separately. Completeness exploits
a result that shows how $\exists$ can play descending along a chain of
$\mu$-approximants and, as in the case of a single equation, it can be proved for general lattices, without assuming the continuity hypothesis.

\begin{lemma}[descending on $\mu$-approximants]
  \label{le:descend-mu}
  Let $E$ be a system of $m$ equations over a
  lattice $L$
  of the kind $\vec{x} =_{\vec{\eta}} \vec{f} (\vec{x})$. For each
  $\mu$-approximant $\vec{l} \in L^m$ and $(b,i)
  \in \Amoves{\vec{l}}$ there exists a
  $\mu$-approximant $\vec{l}' \in \Emoves{b,i}$ such that
  $\ord{\vec{l}} \succeq_i \ord{\vec{l}'}$.
  Moreover, if $\eta_i = \mu$, the $i$-th component strictly decreases
  and thus the inequality is strict.
\end{lemma}

\begin{proof}
  Let $\vec{l} \in L^m$ be a $\mu$-approximant and let
  $(b, i)$ in $\Amoves{\vec{l}}$, i.e., $b \in B_L$ and
  $i \in \interval{m}$ with $b \ll l_i$. We distinguish various cases:

  \begin{enumerate}
    
  \item ($\eta_i = \mu$) This means that
    $l_i = f_{i,\vec{l}}^\alpha(\bot)$ for some ordinal
    $\alpha$. Since $f_{i,\vec{l}}^0(\bot)= \bot$ and $b \ll \bot$
    would imply $b = \bot$, while $\bot \not\in B_L$, necessarily
    $\alpha \neq 0$. We distinguish two subcases:
    \begin{enumerate}
    \item \emph{$\alpha = \beta+1$ is a successor ordinal}\\
      Let $l_i' = f_{i,\vec{l}}^\beta(\bot)$ and $(l_1', \ldots, l_{i-1}') = \sol{
          \subst{\subst{E}{\subvec{x}{i+1}{m}}{\subvec{l}{i+1}{m}}}{x_i}{l_i'}}$. Then define
      \begin{center}
        $\vec{l'} = (l_1', \ldots, l_{i-1}', l_i',\subvec{l}{i+1}{m}
        )$
      \end{center}

      Observe that $\vec{l}'$ is a $\mu$-approximant by
      Lemma~\ref{le:approx-update}. Moreover
      $\vec{l}' \in \Emoves{b,i}$. In fact
      \begin{align*}
        b & \sqsubseteq l_i =  f_{i,\vec{l}}^{\beta+1}(\bot)\\
          & = f_{i,\vec{l}}(f_{i,\vec{l}}^{\beta}(\bot))\\
          &= f_{i,\vec{l}}(l_i')\\
          & = f_i(\sol{\subst{\subst{E}{\subvec{x}{i+1}{m}}{\subvec{l}{i+1}{m}}}{x_i}{l_i'}}, l_i', \subvec{l}{i+1}{m})\\
          & = f_i(l_1', \ldots, l_{i-1}', l_i', \subvec{l}{i+1}{m})\\
          & = f_i(\vec{l}')
      \end{align*}

      Finally, note that $\ord{\vec{l}'} \prec_i \ord{\vec{l}}$ since
      vectors $\vec{l}$ and $\vec{l}'$ coincide on the components
      $i+1, \ldots, m$, and $l_i = f_{i,\vec{l}}^{\beta+1}(\bot)$
      while $l_i' = f_{i,\vec{l}}^\beta(\bot)$.

      \medskip
     
    \item \emph{$\alpha$ is a limit ordinal}\\
      Since
      $b \ll l_i = f_{h,\vec{l}}^\alpha(\bot)=
      \bigsqcup_{\beta<\alpha} f_{i,\vec{l}}^\beta(\bot)$, which is a
      directed join, by
      definition of the way-below relation, there is $\beta < \alpha$
      such that $b \sqsubseteq f_{i,\vec{l}}^\beta(\bot)$. We set
      $l_i' = f_{i,\vec{l}}^\beta(\bot)$ and
      $(l_1', \ldots, l_{i-1}') = \sol{
        \subst{\subst{E}{\subvec{x}{i+1}{m}}{\subvec{l}{i+1}{m}}}{x_i}{l_i'}}$.
      Then we define     
      \begin{center}
        $\vec{l'} = (l_1', \ldots, l_{i-1}', l_i',\subvec{l}{i+1}{m}
        )$
      \end{center}

      The vector $\vec{l}'$ is a $\mu$-approximant by
      Lemma~\ref{le:approx-update}.
      Moreover
      $\vec{l}' \in \Emoves{b,i}$ since
      \begin{align*}
        b & \sqsubseteq l_i'\\
          & \sqsubseteq f_{i,\vec{l}}(l_i')
          &  \mbox{[since $l_i'=f_{i,\vec{l}}^\beta(\bot)$ is a post-fixpoint]}\\
          & = f_i(\sol{\subst{\subst{E}{\subvec{x}{i+1}{m}}{\subvec{l}{i+1}{m}}}{x_i}{l_i'}}, l_i', \subvec{l}{i+1}{m})\\
          & = f_i(l_1', \ldots, l_{i-1}', l_i', \subvec{l}{i+1}{m})\\
          & = f_i(\vec{l}')
      \end{align*}

      Finally, note that $\ord{\vec{l}'} \prec_i \ord{\vec{l}}$ since
      vectors $\vec{l}$ and $\vec{l}'$ coincide on the components
      $i+1, \ldots, m$, and $l_i = f_{i,\vec{l}}^{\alpha}(\bot)$ while
      $l_i' = f_{i,\vec{l}}^\beta(\bot)$, with $\beta < \alpha$.

    \end{enumerate}

    \medskip
    
  \item ($\eta_i = \nu$)\\
    In this case
    $l_i = \nu (f_{i,\vec{l}})$.
    Let
    $(l_1', \ldots, l_{i-1}') = \sol{
      \subst{E}{\subvec{x}{i}{m}}{\subvec{l}{i}{m}}}$. Then define
      \begin{center}
        $\vec{l'} = (l_1', \ldots, l_{i-1}', \subvec{l}{i}{m})$
      \end{center}
      
      The vector $\vec{l}'$ is a
      $\mu$-approximant by Lemma~\ref{le:approx-update}. Moreover, observe that
      $\vec{l}' \in \Emoves{b,i}$, since      
      \begin{align*}
        b & \sqsubseteq l_i\\
          & = f_{i,\vec{l}}(l_i)
          &  \mbox{[since $l_i$ is a fixpoint]}\\
          & = f_i(\sol{\subst{E}{\subvec{x}{i}{m}}{\subvec{l}{i+1}{m}}}, \subvec{l}{i}{m})\\
          & = f_i(l_1', \ldots, l_{i-1}', \subvec{l}{i}{m})\\
          & = f_i(\vec{l}')
      \end{align*}

      Finally, note that $\ord{\vec{l}'} \preceq_i \ord{\vec{l}}$ since
      vectors $\vec{l}$ and $\vec{l}'$ coincide on the components
      $i, \ldots, m$.
    \end{enumerate}
\end{proof}

The previous result allows us to prove that player $\exists$ can
always win starting from a $\mu$-approximant. Roughly, relying on
Lemma~\ref{le:descend-mu}, we can prove that player $\exists$ can play
on $\mu$-approximants in a way that each time the $i$-th equation is
chosen, the ordinal vector associated to the approximant decreases
with respect to $\preceq_i$, and it strictly decreases when the $i$-th
equation is a $\mu$-equation. This, together with the fact that
the order on ordinals is well-founded, allows one to conclude that
either the play is finite and $\exists$ plays last or the highest
index on which one can cycle is necessarily the index of a
$\nu$-equation. In both cases player $\exists$ wins.

\begin{lemma}[$\exists$ wins on $\mu$-approximants]
  \label{le:completeness}
  Let $E$ be a system of $m$ equations over a
  lattice $L$
  of the kind $\vec{x} =_{\vec{\eta}} \vec{f} (\vec{x})$ and let
  $\vec{l} \in L^m$ be a $\mu$-approximant. Then in a game starting
  from $\vec{l}$ (which is a position of $\forall$) player $\exists$
  has a winning strategy.
\end{lemma}

\begin{proof}
  We first describe the strategy for player $\exists$ and then
  prove that it is a winning strategy.

  The key observation is that $\exists$ can always play a
  $\mu$-approximant, where she plays the solution in the first
  step. In fact, let $\vec{l}' \in L^m$ be the current
  $\mu$-approximant. For any possible move
  $(b',i') \in \Amoves{\vec{l}'}$ of $\forall$, by
  Lemma~\ref{le:descend-mu} there always exists a move $\vec{l}'' \in \Emoves{b',i'}$
  of $\exists$ which is a $\mu$-approximant such
  that $\ord{\vec{l'}} \succeq_i \ord{\vec{l}''}$. Additionally, if
  $\eta_i = \mu$ the inequality is strict.
  
  \smallskip

  Since $\exists$ player has always a move,  either the
  play finishes because $\forall$ has no moves, hence $\exists$ wins
  or the play continues forever.

  In this last case, note that, if $h$ is the largest index occurring
  infinitely often, then necessarily $\eta_h = \nu$, hence $\exists$
  wins.
  In fact, assume by contradiction that $\eta_h = \mu$. Consider the
  sequence of turns of the play starting from the point where all
  indexes repeat infinitely often. 

  Let $\vec{l}'$, $(b',j)$, $\vec{l}''$ be consecutive turns.
  By the choice of $h$, necessarily $j \leq h$. Moreover, by
  construction, if
  \begin{center}
    $\ord{\vec{l}'} \succeq_j \ord{\vec{l}''}$
  \end{center}
  Observing that for $j \leq j'$ it holds
  $\vec{\alpha} \succeq_j \vec{\alpha'}$ implies
  $\vec{\alpha} \succeq_{j'} \vec{\alpha'}$, we deduce that
  \begin{center}
    $\ord{\vec{l}'} \succeq_h \ord{\vec{l}''}$ 
  \end{center}
  i.e., the sequence is decreasing. Moreover, since $\eta_h = \mu$,
  whenever $j=h$, $\ord{\vec{l}'} \succ_h \ord{\vec{l}''}$, i.e., the sequence
  strictly decreases. This contradicts the well-foundedness of $\succ_h$.
\end{proof}

Since the solution of a system of equation is a $\mu$-approximant (the greatest one), completeness is an easy corollary of Lemma~\ref{le:completeness}.

\begin{corollary}[completeness]
  \label{co:completeness}
  Let $E$ be a system of $m$ equations over a
  lattice $L$
  of the kind $\vec{x} =_{\vec{\eta}} \vec{f} (\vec{x})$. Given any
  $\mu$-approximant $\vec{l} \in L^m$, $b \in B_L$ and
  $i \in \interval{m}$, if $b \sqsubseteq
  l_i$ then $\exists$ has a
  winning strategy from position $(b,i)$.
\end{corollary}

\begin{proof}
  Just observe that at the first turn $\exists$ can play the
  $\mu$-approximant $\vec{l}$ that is in $\Emoves{b,i}$ by
  hypotheses. Then using Lemma~\ref{le:completeness} we conclude that
  $\exists$ wins.
\end{proof}

For correctness we rely on a result, dual to
Lemma~\ref{le:descend-mu}, that allows to ascend along
$\nu$-approximants. However, in this case, the fact of working in a continuous lattice is crucial (see Proposition~\ref{pr:counter-continuity}).

\begin{lemma}[ascending on $\nu$-approximants]
  \label{le:ascend-nu}
  Let $E$ be a system of $m$ equations over a continuous lattice $L$
  of the kind $\vec{x} =_{\vec{\eta}} \vec{f} (\vec{x})$.
  Given a $\nu$-approximant $\vec{l} \in L^m$, an element $b \in B_L$
  and an index $i \in \interval{m}$ with $b \not\sqsubseteq l_i$, for
  all tuples $\vec{l}' \in \Emoves{b,i}$ there are a $\nu$-approximant
  $\vec{l}''$ and $(b'', j) \in \Amoves{\vec{l}'}$ such that
  (1)~$b'' \not\sqsubseteq l_j''$ and
  (2)~$\ord{\vec{l}} \succeq_i \ord{\vec{l}''}$.
  Moreover, if  $\eta_i = \nu$, the $i$-th component strictly
  decreases and thus the inequality in item 2 above is strict.
\end{lemma}

\begin{proof}
  Let $\vec{l} \in L^m$ be a $\nu$-approximant, let $b \in B_L$ and let
  $i \in \interval{m}$ with $b \not\sqsubseteq l_i$. Take $\vec{l}' \in \Emoves{b,i}$, i.e., such that $b \sqsubseteq f_i(\vec{l}')$.
  We prove that there are a $\nu$-approximant $\vec{l}''$ and $(b'',j)
 \in \Amoves{\vec{l}'}$ satisfying (1) and (2) above, by 
  distinguishing various cases:

  \begin{enumerate}[(i)]
    
  \item ($\eta_i = \mu$)
    Define
    $\vec{l}'' = (\sol{\subst{E}{\subvec{x}{i}{m}}{\subvec{l}{i}{m}}},
    l_i, \subvec{l}{i+1}{m})$, which is a $\nu$-approximant by
    Lemma~\ref{le:approx-update}. Note that, since
    $l_i = \mu(f_{i,\vec{l}})$,
    \begin{center}
      $l_i
      = f_{i,\vec{l}}(l_i)
      =    f_i(\sol{\subst{E}{\subvec{x}{i}{m}}{\subvec{l}{i}{m}}}, l_i, \subvec{l}{i+1}{m})      
      = f_i(\vec{l}'')
      $
    \end{center}
        
    We first prove prove (1), i.e., that there exists $(b'',j) \in \Amoves{\vec{l}'}$, i.e.,
    $j \in \interval{m}$ and $b'' \in B_L$, $b'' \ll l_j'$ with
    $b'' \not\sqsubseteq l_j''$.  In fact, otherwise, if for any $j$ and
    $b'' \ll l_j'$ we had $b'' \sqsubseteq l_j''$, then for any $j$,
    since $B_L$ is a basis and $L$ a continuous lattice:
    \begin{center}
      $l_j' = \bigsqcup \{ b'' \mid b'' \in B_L \land\ b'' \ll l_j' \} \sqsubseteq  l_j''$.
    \end{center}
    However, by monotonicity of $f_i$, this would imply
    $f_i(\vec{l}') \sqsubseteq f_i(\vec{l}'') = l_i$, that together
    with the hypothesis $b \sqsubseteq f_i(\vec{l}')$, would
    contradict $b \not\sqsubseteq l_i$.
    
    \smallskip
    
    For point (2), note that $\ord{\vec{l}'} \preceq_i \ord{\vec{l}}$ since
    vectors $\vec{l}$ and $\vec{l}'$ coincide on all components
    $i, \ldots, m$, and $l_i = f_{i,\vec{l}}^{\alpha}(\bot)$.
    
    \medskip

  \item ($\eta_i = \nu$) This means that
    $l_i = f_{i,\vec{l}}^\alpha(\top)$ for some ordinal $\alpha$,
    necessarily $\alpha \neq 0$ (since otherwise $l_i=\top$ and
    $b \not\sqsubseteq l_i$ could not hold). We distinguish two
    subcases
    
    \begin{enumerate}
    \item \emph{$\alpha = \beta+1$ is a successor ordinal}\\
      Define
      $\vec{l}'' =
      (\sol{\subst{\subst{E}{\subvec{x}{i+1}{m}}{\subvec{l}{i+1}{m}}}{x_i}{f_{i,\vec{l}}^{\beta}(\top)}},
      f_{i,\vec{l}}^{\beta}(\top), \subvec{l}{i+1}{m})$. Then we have
      \begin{align*}
        b & \not\sqsubseteq  l_i\\
          & = f_{i,\vec{l}}^{\beta+1}(\top)\\
          & = f_{i,\vec{l}}(f_{i,\vec{l}}^{\beta}(\top))\\
          & = f_i(\sol{\subst{\subst{E}{\subvec{x}{i+1}{m}}{\subvec{l}{i+1}{m}}}{x_i}{f_{i,\vec{l}}^{\beta}(\top)}}, f_{i,\vec{l}}^{\beta}(\top), \subvec{l}{i+1}{m})\\
          & = f_i(\vec{l}'')
      \end{align*}

      Recalling that $b \sqsubseteq f_i(\vec{l}')$, as in case (i) we
      deduce point (1), i.e., that there exists
      $(b'',j) \in \Amoves{\vec{l}'}$ such that
      $b'' \not\sqsubseteq l_j''$.

      \smallskip

      Concerning point (2), note that
      $\ord{\vec{l}'} \prec_i \ord{\vec{l}}$ since vectors $\vec{l}$
      and $\vec{l}'$ coincide on the components $i+1, \ldots, m$, and
      $l_i = f_{i,\vec{l}}^{\beta+1}(\bot)$ while
      $l_i' = f_{i,\vec{l}}^\beta(\bot)$.

      \medskip

    \item \emph{$\alpha$ is a limit ordinal}\\
      In this case
      \begin{center}
        $b \not\sqsubseteq l_i =
        f_{i,\vec{l}}^\alpha(\top) =
        \bigsqcap_{\beta<\alpha} f_{i,\vec{l}}^\beta(\top) = \bigsqcap_{\beta<\alpha} f_{i,\vec{l}}^{\beta+1}(\top)$
      \end{center}
      Therefore there exists $\beta < \alpha$ such that
      $b \not\sqsubseteq f_{i,\vec{l}}^{\beta+1}(\top)$. Hence, we can
      define $l_i'' = f_{i,\vec{l}}^{\beta}(\top)$ and take the
      $\nu$-approximant
      \begin{center}
        $\vec{l}'' = (\sol{\subst{\subst{E}{\subvec{x}{i+1}{m}}{\subvec{l}{i+1}{m}}}{x_i}{l_i''}},
        l_i'', \subvec{l}{i+1}{m})$
    \end{center}
    
    Then we have
      \begin{align*}
        b & \not\sqsubseteq  f_{i,\vec{l}}^{\beta+1}(\top)\\
          & = f_{i,\vec{l}}(f_{i,\vec{l}}^{\beta}(\top))\\
          & = f_i(\sol{\subst{\subst{E}{\subvec{x}{i+1}{m}}{\subvec{l}{i+1}{m}}}{x_i}{f_{i,\vec{l}}^{\beta}(\top)}}, f_{i,\vec{l}}^{\beta}(\top), \subvec{l}{i+1}{m})\\
          & = f_i(\vec{l}'')
      \end{align*}
      and thus, again, recalling that $b \sqsubseteq f_i(\vec{l}')$, as in 
      case (i) we deduce point (1), i.e., that there exists
      $(b'',j) \in \Amoves{\vec{l}'}$ such that
      $b'' \not\sqsubseteq l_j''$.

    Concerning point (2), note that
    $\ord{\vec{l}''} \prec_i \ord{\vec{l}}$ since vectors $\vec{l}$ and
    $\vec{l}''$ coincide on the components $i+1, \ldots, m$, and
    $l_i = f_{i,\vec{l}}^{\alpha}(\bot)$ while
    $l_i'' = f_{i,\vec{l}}^\beta(\bot)$, with $\beta < \alpha$.
    \end{enumerate}
  \end{enumerate}
\end{proof}

As in the dual case, correctness is an easy corollary of the above
lemma, recalling that the solution is the least
$\nu$-approximant. %

\begin{lemma}[correctness]
  \label{le:correctness}
  Let $E$ be a system of $m$ equations over a continuous lattice $L$
  of the kind $\vec{x} =_{\vec{\eta}} \vec{f} (\vec{x})$. For a
  $\nu$-approximant $\vec{l} \in L^m$, $b \in B_L$ and
  $i \in \interval{m}$, if $b \not\sqsubseteq l_i$ then $\forall$ has a
  winning strategy from position $(b,i)$.
\end{lemma}

\begin{proof}
  We first describe the strategy for the universal player and then
  prove that it is a winning strategy.

  Let $\vec{l} \in L^m$ be a $\nu$-approximant, $b \in L$ and
  $i \in \interval{m}$ such that $b \not\sqsubseteq l_i$.  Starting
  from $(b, i)$, for any possible move $\vec{l}' \in \Emoves{b,i}$ of
  $\exists$.
  Then $\forall$ can play a pair $(b',j) \in \Amoves{\vec{l}'}$, whose
  existence is ensured by Lemma~\ref{le:ascend-nu}, such that there is
  a $\nu$-approximant $\vec{l}''$ satisfying
  $b'' \not\sqsubseteq l_j''$ and
  $\ord{\vec{l}''} \prec_i \ord{\vec{l}}$. Additionally, if
  $\eta_i = \nu$ the inequality is strict.
  
  \medskip
  
  \medskip

  According to the strategy defined above $\forall$ player has always a
  move. Thus either the play finishes because $\exists$ has no moves,
  hence $\forall$ wins or the play continues forever.

  In this last case, with an argument dual with respect to that in Lemma~\ref{le:completeness}, we can show that if $h$ is the largest index occurring
  infinitely often, then necessarily $\eta_h = \mu$, hence $\forall$
  win.
\end{proof}

Combining Corollary~\ref{co:completeness} and
Lemma~\ref{le:correctness} we reach the desired result.

\begin{theorem}[correctness and completeness]
  \label{th:game-corr-comp}
  Given a system of $m$ equations $E$ over a continuous lattice $L$ of the kind
  $\vec{x} =_{\vec{\eta}} \vec{f} (\vec{x})$ with solution $\vec{u}$,
  then for all $b \in B_L$ and $i \in \interval{m}$,
  \begin{center}
    $b \sqsubseteq u_i$ \quad iff \quad $\exists$ has a winning
    strategy from position $(b,i)$.
\end{center}
\end{theorem}

\begin{proof}
  Immediate corollary of Lemma~\ref{le:solution-is-approximant},
  Corollary~\ref{co:completeness} and Lemma~\ref{le:correctness}.
\end{proof}

Note that even when the fixpoint is reached in more than $\omega$
steps, thanks to the fact that the order on the ordinals is
well-founded and players ``descend'' over the order, ordinals do
not play an explicit role in the game.  In particular plays are not
transfinite and whenever $\forall$ or $\exists$ win due to the fact
that the other player cannot make a move, this happens after a
\emph{finite} number of steps. This can be a bit surprising at first
since the game works for general continuous lattices, including, for
instance, intervals over the reals.

We close this subsection by proving two results that, in a sense, show
that the choice of continuous lattices and the design of the game
based on the way-below relation are ``the right ones''.
We first observe that the restriction to continuous
lattices is not only sufficient but also necessary for the correctness of the game. 

\begin{proposition}[correctness holds exactly in continuous lattices]
  \label{pr:counter-continuity}
  Let $L$ be a lattice and let $B_L$ be a fixed basis with
  $\bot\not\in B_L$. The game is correct for every system of equations
  over $L$ if and only if $L$ is continuous.
\end{proposition}

\begin{proof}
  We already know from Lemma~\ref{le:correctness} that when $L$ is
  continuous the game is correct.
  
  Conversely, let $L$ be a non-continuous lattice. This means that there is an
  element $l \in L$ such that $l \neq \lub \cone{l}$. Note that since
  $\cone{l} \subseteq \,\downarrow\! l$, where $\downarrow$ is the
  downward-closure with respect to $\sqsubseteq$, we have
  $\lub \cone{l} \sqsubset \lub \downarrow\! l = l$. We prove that,
  for any basis $B_L$ for $L$ such that $\bot \not\in B_L$, there are a
  monotone function $f \colon L \to L$ and an element $b \in B_L$ such that
  $b \not\sqsubseteq \mu f$, for which there is a winning strategy for
  the existential player for the corresponding fixpoint game starting
  from position $b$, while such a strategy should not exists. The
  function $f$ is defined by:
	\[ f(x) = \left\{
      \begin{array}{ll}
        \lub \cone{l} & \text{ if }x \sqsubseteq\lub \cone{l} \\
        \top & \text{ otherwise.}
      \end{array}
    \right.
  \]
  Notice that necessarily $\lub \cone{l} \neq \top$, since
  $\lub \cone{l} \sqsubset l \sqsubseteq \top$. Then, clearly $f$ is
  monotone and its least fixpoint is $\mu f = \lub \cone{l}$.
  Moreover, since $B_L$ is a basis, we know that
  $\lub (\downarrow\! l \cap B_L) = l$. Then there must be $b \in B_L$
  such that $b \sqsubseteq l$ but $b \not\sqsubseteq \lub \cone{l}$.
  Otherwise, if for all $b \in B_L$ such that $b \sqsubseteq l$,
  $b \sqsubseteq \lub \cone{l}$, then we would have
  $\lub (\downarrow\! l \cap B_L) \sqsubseteq \lub \cone{l}$,
  contradicting the hypothesis
  $\lub \cone{l} \sqsubset l = \lub (\downarrow\! l \cap B_L)$.  Now
  we show that player $\exists$ is able to win any play of the game,
  with a single equation $x =_\mu f(x)$, for checking whether such a
  $b \sqsubseteq \mu f$. The strategy is actually quite simple and can
  be described by just one family of plays. We use the fact that $f(l)
  = \top$.
  \[b \stackrel{\exists}{\leadsto} l
    \stackrel{\forall}{\leadsto} b'
    \stackrel{\exists}{\leadsto} \bot
    \stackrel{\forall}{\not\leadsto}\]
  for any $b' \ll l$, since $b' \sqsubseteq \lub \cone{l} = f(\bot)$.
  Thus player $\exists$ can always win, despite the fact that
  $b \not\sqsubseteq \mu f$.
\end{proof}

As a counterexample, consider the lattice $W$ in
Fig.~\ref{fi:non-continuous}, which is not continuous and let $B_W$ be
any basis such that $0 \not\in B_W$.
First note that necessarily $a \in B_W$, otherwise
$a \neq \bigsqcup \{ x \in B_W \mid x \sqsubseteq a \} = \bigsqcup
\emptyset = 0$. Secondly, $\cone{a} = \{0\}$ since $a \not\ll a$.
Then, consider the equation $x =_\mu f(x)$, where the function
$f \colon W \to W$ is defined by $f(0) = 0$, and $f(x) = \omega$ for
$x \neq 0$. Clearly $f$ is monotone and its least fixpoint is
$\mu f = 0$. However, the player $\exists$ can win any play of the
game from position $a$, despite the fact that
$a \not\sqsubseteq \mu f = 0$. In fact, the first move of $\exists$
can be $a$, since $a \sqsubseteq f(a) = \omega$. But then
player $\forall$ has no moves since
$\cone{a} \cap B_W = \emptyset$. And so player $\exists$ always wins
while she should not.

\smallskip

The second observation is that using the lattice order instead of the
way-below relation may break completeness.
More precisely, consider the natural variant of the game where the
way-below relation is replaced by the lattice order. Let us call it
\emph{weak game}. Since the set of possible moves of player $\forall$
is enlarged, correctness clearly continues to hold. Instead, as we
hinted before, completeness could fail. We show that it is exactly on
algebraic lattices that completeness still holds for the weak game.

\begin{proposition}[way-below is needed in non-algebraic lattices]
  \label{pr:counter-way-below}  
  Let $L$ be a lattice. The  weak game is complete on every system of
  equations over $L$ if and only if $B_L$ consists of compact elements (which
  in turn  means that $L$ is algebraic). 
\end{proposition}

\begin{proof}
  Let $K_L$ be the set of compact elements of $L$. If
  $B_L \subseteq K_L$, then for any $b \in B_L$ and $l \in L$, we have
  that $b \ll l$ if and only if $b \sqsubseteq l$. Therefore the weak
  game coincides with the original one and hence completeness clearly
  holds.

  Conversely, assume that $B_L \not\subseteq K_L$. We show that we
  can identify a system consisting of a single equation $x =_\mu f(x)$
  for which the weak game is not complete. Let
  $b \in B_L \setminus K_L$ be a non-compact element in the
  basis. Therefore $b \not\ll b$ which means that there exists a
  directed set $D$ such that $b \sqsubseteq \bigsqcup D$ and
  $b \not\sqsubseteq d$ for all $d \in D$. Without loss of generality
  we can assume that $D$ is a transfinite chain
  $D = (d_\alpha)_\alpha$ (see, e.g.,~\cite[Theorem 1]{Mar:CCP}).
    
  Consider the function $f \colon L \to L$ defined as
  \begin{quote}
    $
    f(x) =
    \left\{
      \begin{array}{ll}
        \top & \mbox{if $b \sqsubseteq x$}\\
        d_{\alpha} & \mbox{otherwise, where $\alpha = \min\{ \beta \mid d_\beta \not\sqsubseteq x \}$}
      \end{array}
    \right.
    $
  \end{quote}
  Observe that the function $f$ is well-defined. In fact, when
  $b \not\sqsubseteq x$ there exists $\beta$ such that
  $d_\beta \not\sqsubseteq x$ and hence the set
  $\{ \beta \mid d_\beta \not\sqsubseteq x \}$ is not empty. In fact,
  if we had $d_\beta \sqsubseteq x$ for all elements of the chain, we
  would deduce $\bigsqcup D \sqsubseteq x$ and thus, recalling
  $b \sqsubseteq \bigsqcup D$, we would conclude $b \sqsubseteq x$.

  Observe that $f$ is monotone. In fact, let $x, y \in L$ with
  $x \sqsubseteq y$. If $b \sqsubseteq y$ and thus $f(y) = \top$, we
  trivially conclude $f(x) \sqsubseteq \top = f(y)$. Let us then consider
  the case in which $b \not\sqsubseteq y$ and thus $f(y) = d_{\alpha}$
  where $\alpha = \min\{ \beta \mid d_\beta \not\sqsubseteq y
  \}$. Obviously $b \not\sqsubseteq x$ and thus $f(x) = d_{\alpha'}$,
  where $\alpha' = \min\{ \beta \mid d_\beta \not\sqsubseteq x \}$. Since
  $x \sqsubseteq y$, we have
  $\{ \beta \mid d_\beta \not\sqsubseteq x \} \supseteq \{ \beta \mid
  d_\beta \not\sqsubseteq y \}$, hence $\alpha' \leq \alpha$ and thus
  $f(x) = d_{\alpha'} \sqsubseteq d_{\alpha} = f(y)$, as desired.
  
  By construction $\top$ is the only fixpoint of $f$, hence
  $\top = \mu f$.
  Thus $b \sqsubseteq \mu f = \top$. Now, if we play the weak game,
  since $b \not\sqsubseteq d_\alpha$ for all $\alpha$, the possible moves for $\exists$
  are initially only those $x \in L$ such that $b \sqsubseteq x$ and
  thus $b \sqsubseteq f(x) = \top$. However, if $\exists$ play such an
  $x$, in the weak game $\forall$ can answer $b$, getting back to the
  initial situation. Hence $\forall$ wins, providing the desired
  counterexample to completeness.
\end{proof}

Note that when the elements of the basis are compact, the way-below relation with respect to elements of the basis is the lattice order. Hence the  result above essentially states that the weak game is complete exactly when it coincides with the original game, thus further supporting the appropriateness of our formulation of the game.

As a counterexample, consider the continuous lattice $[0,1]$ with the
usual order and basis $B_{[0,1]} = \mathbb{Q} \cap (0,1]$.  Recall
that $[0,1]$ is not algebraic (the only compact element is $0$) and
way-below relation is the strict order $<$.
Let $g \colon [0,1] \to [0,1]$ be the function defined by
$g(x) = \frac{x+1}{2}$. The fixpoint equation $x =_\mu g(x)$ has
solution $\mu g = 1$.
  
In the weak game, from position $l \in [0,1]$, player $\forall$ can
play any $b \leq l$ (instead, of $b < l$).
Then player $\exists$ loses any play starting from position $1$,
despite the fact that $1 \leq \mu g = 1$. In fact, the only possible
move of player $\exists$ is $1$, and $\forall$ can play any
$x \leq 1$. In particular, playing $1$ the game will continue forever
and will thus be won by $\forall$.

Notice that, instead, in the original game, from position $1$, player
$\forall$ has to play an element $1 - \epsilon$ for some
$\epsilon > 0$.  Then, it is easy to see that at each step $i$ player
$\exists$ will be able to play some $z_i \leq 1-2^i \epsilon$. This
means that after finitely many steps $\exists$ will be allowed to play
$0$, thus leaving no possible answer to $\forall$ and winning the
game.

\subsection{Relation to $\mu$-Calculus Model-Checking}

We discuss how our fixpoint game over systems
arising from $\mu$-calculus formulae relates to classical techniques for
model-checking the $\mu$-calculus, which can be presented interchangeably
in terms of parity games, tableaux, and automata (see, e.g.,~\cite{Eme:ATTL}). Specifically, we compare our game with classical tableau
systems for the $\mu$-calculus (e.g., as in~\cite{Cle:TMCPMC,SW:LMCMMC}) where the similarities can be presented more directly.

Recall that a tableau is a (finite) proof tree whose nodes are
labelled by sequents. Usually sequents are of the kind $s\models\phi$,
where $s$ is a state of the model and $\varphi$ formula. The fact that
a state $s$ satisfies a formula $\phi$ amounts to the existence of a
tableau, rooted in $s\models\phi$ and that it is \emph{successful},
according to a suitable definition.

Given a closed $\mu$-calculus formula $\phi$ and a state $s$ in a
transition system $(\mathbb{S},\to)$, let $E$ be the corresponding
system of $m$ equations as in
Definition~\ref{de:eqnsys-for-muformula}.
The model-checking problem using tableaux is solved by searching for a
successful tableau for the sequent $s\models\phi$. Instead, using the
fixpoint game, it is reduced to the existence of a winning strategy
for player $\exists$ starting from position $(\{s\},m)$, where $m$ is
the highest equation index.

We discuss the two approaches using Example~\ref{ex:running}.
Recall that the formula of interest is
$\phi = \nu x_2.((\mu x_1.(p\lor\Diamond x_1))\land\Box x_2)$. Let
$\psi$ denote the subformula $\mu x_1.(p\lor\Diamond x_1)$. Using the
tableau rules in the style of~\cite{Cle:TMCPMC} (omitting assumptions
for the sake of the presentation), we can build a successful tableau
for the sequent $a\models\phi$ as in Fig.~\ref{fi:running-tableau}.
It is not difficult to see that this tableau corresponds to the
winning strategy $\varsigma$ for $\exists$ discussed in
Example~\ref{ex:running-a}. In fact, consider the reduced tree in
Fig.~\ref{fi:running-tableau-red}, which is obtained from the tableau
by keeping only the sequents corresponding to fixpoint formulae (i.e.,
either $\phi$ or $\psi$) and replacing such formulae with the
corresponding variable ($\phi$ with $x_2$ and $\psi$ with $x_1$).

Each sequent $s \models x_i$ can be seen as a position
$(\{s\},i) \in B_{2^{\mathbb{S}}}\times \interval{2}$ of $\exists$ in
the fixpoint game.
The successor sequents correspond to the move prescribed on
$(\{s\},i)$ by the strategy $\varsigma$. More precisely, the move
should be $(y_1, y_2)$ where
$y_j = \{ s' \mid\, s' \models x_j\text{ is a successor of }s \models x_i\}$ for $j \in \interval{2}$.
For instance, the sequent $a\models x_2$ corresponds to the position
$(\{a\},2)$. The three successors $a \models x_1$, $a \models x_2$ and
$b \models x_2$ determine the move prescribed by the strategy
$\varsigma(\{a\},2) = (\{a\},\{a,b\})$. Instead, the sequent
$a\models x_1$ has only one successor $b \models x_1$ and,
correspondingly, we have $\varsigma(\{a\},1) = (\{b\},\emptyset)$,
since there are no successors containing variable $x_2$.
When a sequent appears on a leaf of the reduced tree which
was already a leaf in the original tableau,
by definition of the tableau rules it must have an ancestor labelled
by the same sequent and in this case the strategy is defined by the
ancestor. For instance, in Fig.~\ref{fi:running-tableaux}, the sequent $b\models x_2$ labels the leaf {\raisebox{.5pt}{\textcircled{\raisebox{-.9pt}{1}}} which was already a leaf in the original tableau, marked by {\raisebox{.5pt}{\textcircled{\raisebox{-.9pt}{2}}}}. The strategy is
thus defined by the ancestor {\raisebox{.5pt}{\textcircled{\raisebox{-.9pt}{3}}}, labelled by the same sequent
$b \models x_2$, as $\varsigma(\{b\},2) = (\{b\},\{b\})$.

\begin{figure}
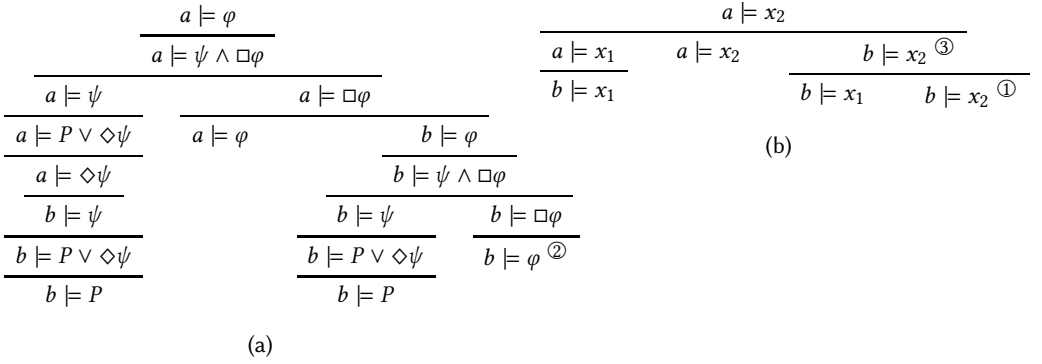

  \small
  \begin{subfigure}[t]{.49\textwidth}
    \centering
    \begin{prooftree}
      \AxiomC{$b \models P$}
      \UnaryInfC{$b \models P\lor\Diamond\psi$}
      \UnaryInfC{$b \models \psi$}
      \UnaryInfC{$a \models \Diamond\psi$}
      \UnaryInfC{$a \models P\lor\Diamond\psi$}
      \UnaryInfC{$a \models \psi$}
        \AxiomC{$a \models \phi$}
      	  \AxiomC{$b \models P$}
          \UnaryInfC{$b \models P\lor\Diamond\psi$}
      	  \UnaryInfC{$b \models \psi$}
      	    \AxiomC{$b \models \phi$ \textsuperscript{\raisebox{.5pt}{\textcircled{\raisebox{-.9pt}{2}}}}}
      	    \UnaryInfC{$b \models \Box\phi$}
          \BinaryInfC{$b \models \psi\land\Box\phi$}
      	  \UnaryInfC{$b \models \phi$}
        \BinaryInfC{$a \models \Box\phi$}
      \BinaryInfC{$a \models \psi\land\Box\phi$}
      \UnaryInfC{$a \models \phi$}
    \end{prooftree}
    \caption{}
    \label{fi:running-tableau}
  \end{subfigure}
  \begin{subfigure}[t]{.49\textwidth}
    \centering
    \begin{prooftree}
    \AxiomC{$b \models x_1$}
    \UnaryInfC{$a \models x_1$}
      \AxiomC{$a \models x_2$}
        \AxiomC{$b \models x_1$}
    	    \AxiomC{$b \models x_2$ \textsuperscript{\raisebox{.5pt}{\textcircled{\raisebox{-.9pt}{1}}}}}
        \BinaryInfC{$b \models x_2$ \textsuperscript{\raisebox{.5pt}{\textcircled{\raisebox{-.9pt}{3}}}}}
      \TrinaryInfC{$a \models x_2$}
    \end{prooftree}
    \caption{}
    \label{fi:running-tableau-red}
  \end{subfigure}
  \caption{$\mu$-calculus tableaux vs strategies in the fixpoint game}
  \label{fi:running-tableaux}
\end{figure}

Additionally, it can be seen that plays of the fixpoint game
correspond to paths in the reduced tree. For example, the first play
discussed in Example~\ref{ex:running-a} corresponds to the leftmost path
in the tree. In fact, while successors of sequents define the strategy
for player $\exists$, the moves of player $\forall$ determine the path
to follow.

For general, possibly non-successful tableaux, if we consider the
reduced tree, then for each subtree the sequents at the leaves can be
read as a set of assumptions that player $\exists$ has taken to show
that the root sequent holds.
Player $\forall$ chooses among such assumptions which one player
$\exists$ should develop next.
If there is no winning strategy for player $\exists$, the winning
strategy for player $\forall$ is such that he always chooses a path in
the tableau that cannot be successfully concluded at a leaf.

\leavevmode
\subsection{Fixpoint Games in Data-Flow Analysis}
\label{ss:data-flow-game}

We get back to constant propagation example in
\S~\ref{ss:data-flow}. Recall that the system of fixpoint equations
expressing the analysis in Fig.~\ref{fig:constant-analysis} had
solution $\rho_1 = \bot$,
$\rho_2 = \bot[\prg{y}\mapsto 6]$,
$\rho_3 = \rho_4 = \bot[\prg{x}\mapsto 7]$.
We  next
describe a game that shows that indeed
$\bot[\prg{x}\mapsto 7] \sqsubseteq \rho_4$ and hence
$\prg{x}$ has constant value~$7$ at block~$4$. The game starts as
follows: 
\[
  (\bot[\prg{x}\mapsto 7],4) \stackrel{\exists}{\leadsto}
  (\bot,\bot,\bot[\prg{x}\mapsto 7],\bot) \stackrel{\forall}{\leadsto}
  (\bot[\prg{x}\mapsto 7],3) \stackrel{\exists}{\leadsto}
  (\bot,\bot[\prg{y}\mapsto 6],\bot,\bot[\prg{x}\mapsto 7])
  \stackrel{\forall}{\leadsto}
\]
Now the universal player has two options: either choose
$(\bot[\prg{x}\mapsto 7],4)$, which brings him back to an earlier game
situation and might potentially lead to an infinite game. Since we are
considering greatest fixpoints, this means that $\exists$ wins. If he
chooses the other option, the game continues as follows, where
eventually $\forall$ has no move left and $\exists$ wins as well:
\[
  (\bot[\prg{y}\mapsto 6],2) \stackrel{\exists}{\leadsto}
  \bot \stackrel{\forall}{\not\leadsto}
\]

\section{Strategies as Progress Measures}
\label{sec:progress}

Along the lines of~\cite{j:progress-measures-parity}, influenced
by~\cite{hsc:lattice-progress-measures}, in this section we introduce
a general notion of progress measure for fixpoint games over
continuous lattices.
We will show how a complete progress measure characterises the winning positions for the two players. The existence of a so-called small progress measure will allow us to express a complete progress measure as a least fixpoint, thus providing a technique for computing the progress measure and solving the corresponding system of equations.

\subsection{General Definition}

Given an ordinal $\alpha$ we denote by
$\lift{\alpha}{m} = \{ \beta \mid \beta \leq \alpha\}^m \cup \{
\err\}$, the set of ordinal vectors with entries smaller or equal
than $\alpha$, with an added bound $\err$.

\begin{definition}[progress measure]
  \label{de:progress-measure}
  Let $L$ be a continuous lattice and let $E$ be a system of $m$
  equations over $L$ of the kind
  $\vec{x} =_{\vec{\eta}} \vec{f}(\vec{x})$. Given an ordinal
  $\lambda$, a \emph{$\lambda$-progress measure} for $E$ is a function
  $R\colon B_L \to \interval{m} \to \lift{\lambda}{m}$ such that for
  all $b \in B_L$, $i \in \interval{m}$, either $R(b)(i) = \err$ or
  there exists $\vec{l} \in \Emoves{b,i}$ such that for all
  $(b',j) \in \Amoves{\vec{l}}$ it holds
  
  \begin{itemize}
  \item if $\eta_i = \mu$ then $R(b)(i) \succ_i R(b')(j)$;
  \item if $\eta_i = \nu$ then 
      $R(b)(i) \succeq_i R(b')(j)$
  \end{itemize}
\end{definition}

A progress measure maps any basis element of the lattice and index
$i \in \interval{m}$ to an $m$-tuple of ordinals, with one component
for each equation. Components relative to $\mu$-equations roughly
measure how many unfolding steps for the equation would be needed to
reach an under-approximation $l_i$ above $b$, and thus, for $\exists$,
to win the game. Components relative to $\nu$-equations, as in the
original work of~\cite{j:progress-measures-parity}, are less relevant,
as we will see.

Intuitively, whenever $R(b)(i) \neq \err$, the progress measure $R$
provides an evidence of the existence of a winning strategy for
$\exists$ in a play starting from $(b,i)$. The tuple $\vec{l}$, whose
existence is required by the definition, is a move of player $\exists$
such that for any possible answer of $\forall$, the progress measure
will not increase with respect to $\preceq_i$, and it will strictly
decrease in the case of $\mu$-equations. Since $\prec_i$ is
well-founded, this ensures that we cannot cycle on a
$\mu$-equation. Also note that whenever the current index is $i$, all
indices lower than $i$ are irrelevant (expressed by the orders
$\succeq_i$ resp.\ $\succ_i$), which is related to the fact that the
highest index which is visited infinitely often is the only relevant
index for determining the winner of the game.
This idea is formalised in the following lemma.

\begin{lemma}[progress measures are strategies]
  \label{le:progress-strategy}
  Let $L$ be a continuous lattice and let $E$ be a system of $m$ equations over $L$
  of the kind $\vec{x} =_{\vec{\eta}} \vec{f}(\vec{x})$ with solution
  $\vec{u}$.
  For any $b \in B_L$ and $i \in \interval{m}$, if there exists some ordinal $\lambda$  and a
  $\lambda$-progress measure $R$ such that
  $R(b)(i) \preceq_i (\lambda, \ldots, \lambda)$, then $b \sqsubseteq u_i$.
\end{lemma}

\begin{proof}  
  We show that  $\exists$ has a
  winning strategy from $(b,i)$. The strategy consists in choosing a move
  $\vec{l} \in \Emoves{b,i}$ such that  for all
  $(b',j) \in \Amoves{\vec{l}}$, it holds

  \begin{itemize}    
  \item $R(b)(i) \succ_i R(b')(j)$, if $\eta_i = \mu$
  \item $R(b)(i) \succeq_i R(b')(j)$, if $\eta_i = \nu$
  \end{itemize}
  which  exists by definition of progress measure.

  Now, observe that player $\exists$ can always make its
  turn. Therefore either the play stops because $\forall$ runs out of
  moves, hence $\exists$ win. Otherwise, the play is infinite, and, if
  we denote by $h$ the largest index occurring infinitely often, then
  $\eta_h = \nu$, hence $\exists$ wins. In fact, assume by
  contradiction that $\eta_h = \mu$. Consider the sequence of turns of
  the play starting from the point where all indexes repeat infinitely
  often and take the $m$-tuples of ordinals $R(b')(h)$ corresponding
  to the positions $(b', i)$ where $\exists$ plays. For any two
  successive elements, say $(b',i)$ and $(b'',j)$, by construction
  \begin{center}
    $R(b')(i) \succeq_i R(b'')(j)$
  \end{center}
  Observing that for $i \leq j$ it holds
  $\vec{\alpha} \succeq_i \vec{\alpha'}$ implies
  $\vec{\alpha} \succeq_j \vec{\alpha'}$, we deduce that
  \begin{center}
    $R(b')(i) \succeq_h R(b'')(j)$
  \end{center}
  i.e., the sequence is decreasing. Moreover, since $\eta_h = \mu$,
  whenever $i=h$, $R(b')(i) \succ_h R(b'')(j)$, i.e., the sequence
  strictly decreases. This contradicts well-foundedness of $\prec_h$.
\end{proof}

The above lemma, in a sense, says that progress measures provide sound
characterisations of the solution. However, in general, they are not
complete, since whenever $R(b)(i) = \err$ we cannot derive any
information on $(b,i)$, i.e., if $\vec{u}$ is the solution of the
system, we cannot conclude that $b \not\sqsubseteq u_i$. This motivates
the following definition.

\begin{definition}[complete progress measures]
  \label{de:complete-progress-measure}
  Let $L$ be a continuous lattice and let $E$ be a system of equations over $L$
  of the kind $\vec{x} =_{\vec{\eta}} \vec{f}(\vec{x})$ with solution
  $\vec{u}$.
  A $\lambda$-progress measure
  $R\colon B_L \to \interval{m} \to \lift{\lambda}{m}$ is called
  \emph{complete} if for all $b \in B_L$ and $i \in \interval{m}$, if
  $b \sqsubseteq u_i$ then
  $R(b)(i) \preceq_i (\lambda, \ldots, \lambda)$.
\end{definition}

Observe that in search of a complete progress measure, in principle,
we would have to try all ordinals as a bound.
We next show that we can take as bound the height $\asc{L}$ of the
lattice $L$.
This provides a generalisation of the \emph{small progress measure}
in~\cite{j:progress-measures-parity}.

\begin{definition}[small progress measure]
  \label{de:small}
  Let $L$ be a continuous lattice and let $E$ be a system of $m$ equations
  over $L$ of the kind $\vec{x} =_{\vec{\eta}}
  \vec{f}(\vec{x})$. Given an $m$-tuple of ordinals $\vec{\alpha}$,
  let us denote by $z_E(\vec{\alpha})$ the $m$-tuple of ordinals where
  $\nu$-components are set to $0$, i.e.,
  $z_E(\vec{\alpha})= \vec{\beta}$ with $\beta_i = \alpha_i$ if
  $\eta_i=\mu$ and $\beta_i = 0$ otherwise.
  We define the \emph{small progress measure}
  $R_E\colon B_L \to \interval{m} \to \lift{\asc{L}}{m}$ 
  \begin{center}
    $
    R_E(b)(i) = \min\nolimits_{\preceq_i} \{ z_E(\ord{\vec{l}}) \mid \vec{l}\
    \text{ is a }\mu\text{-approximant}\ \land\ \vec{l} \in \Emoves{b,i} \}
    $
  \end{center}
  where $\min\nolimits_{\preceq_i}$ is the minimum on $\preceq_i$ as
  given in Definition~\ref{de:trunc-ord}, with the convention that
  $\min\nolimits_{\preceq_i} \emptyset = \err$.
\end{definition}

Observe that $R_E$ is well-defined, i.e., it actually takes values in
$\lift{\asc{L}}{m}$. In fact, the components of $z_E(\ord{\vec{l}})$
corresponding to $\mu$-indices are ordinals expressing the number of
Kleene iterations needed to reach under-approximations of the least
fixpoint. These are clearly bounded by $\asc{L}$, since for a monotone
function $f\colon L \to L$, the sequence $f^\alpha(\bot)$ is strictly
increasing until it reaches the least fixpoint of $f$. For
$\nu$-indices, $z_E(\ord{\vec{l}})$ is always $0$.

Observe that while formally $R_E(b)(i)$ takes values in
$\lift{\asc{L}}{m}$, whenever $j < i$ or $\eta_j = \nu$, due
to the effect of the $\min\nolimits_{\preceq_i}$ and of the $z_E$ operations,
the only possible value for the $j$-th component is $0$. Despite such
components are then clearly irrelevant, we keep them for notational
convenience.

The fact that $R_E$ is indeed a progress measure follows from
Lemma~\ref{le:descend-mu}. 
Moreover, we can easily show that it is complete.

\begin{lemma}[small progress measure]
  \label{le:small-progress}
  Let $L$ be a continuous lattice and let $E$ be a system of $m$
  equations over $L$ of the kind
  $\vec{x} =_{\vec{\eta}} \vec{f}(\vec{x})$.
  Then $R_E \colon B_L \to \interval{m} \to \lift{\asc{L}}{m}$ is a
  progress measure and it is complete.
\end{lemma}

\begin{proof}
  For the first part, let $R_E(b)(i) = \vec{\alpha} \neq \err$. Hence
  $R_E(b)(i) =_i z_E(\ord{\vec{l}})$ for some $\mu$-approximant $\vec{l}$
  such that $\vec{l} \in \Emoves{b,i}$.
  By Lemma~\ref{le:descend-mu}, for all
  $(b',j) \in \Amoves{\vec{l}}$ there exists
  $\vec{l}' \in \Emoves{b',j}$ such that
  $\ord{\vec{l}} \succeq_i \ord{\vec{l}'}$ and, if $\eta_i = \mu$, the
  inequality is strict since the $i$-th component strictly decreases.
  Clearly, this implies
  $z_E(\ord{\vec{l}}) \succeq_i z_E(\ord{\vec{l}'})$. Additionally, if
  $\eta_i = \mu$, the inequality remains strict since the $i$-th
  component is left unchanged by the $z_E$ operation.

  Therefore, by definition of $R_E$, for all
  $(b',j) \in \Amoves{\vec{l}}$ we have
  \begin{equation}
    \label{eq:prog-small}
    R_E(b')(j) \preceq_i z_E(\ord{\vec{l}'}) \preceq_i z_E(\ord{\vec{l}}) =_i
    R_E(b)(i).
  \end{equation}
  where, if $\eta_i = \mu$, the inequality is strict, as desired.

  \medskip

  Let us now show that $R_E$ is complete.
  Let $b \in B_L$ be such that $b \sqsubseteq u_i$. We know that 
  the solution $\vec{u}$ is a $\mu$-approximant. Moreover, since
  $b \sqsubseteq u_i = f_i(\vec{u})$, we have that
  $\vec{u} \in \Emoves{b,i}$. Hence
  $R_E(b)(i) \preceq_i z_E(\ord{\vec{u}})$ and thus
  $R_E(b)(i) \neq \err$.

\end{proof}

\subsection{Progress Measures as Fixpoints}
\label{ssec:progress-fix}
  
Here we show that a complete progress measure can be characterised as the least solution of a system of equations over tuples of ordinals, naturally induced by Definition~\ref{de:progress-measure}.

\begin{definition}[progress measure equations]
  \label{de:progress-fixpoint}
  Let $L$ be a continuous lattice and let $E$ be a system of $m$ equations over $L$
  of the kind $\vec{x} =_{\vec{\eta}} \vec{f}(\vec{x})$.
  Let $\vec{\delta}_i^\eta$, with
  $i \in \interval{m}$, be, for $\eta=\nu$, the null vector and, for
  $\eta = \mu$, the vector such that $\delta_j = 0$ if $j\neq i$ and
  $\delta_i = 1$.
  The \emph{progress measure equations} for $E$ over the lattice
  $\lift{\asc{L}}{m}$, are defined, for $b \in B_L$,
  $i \in \interval{m}$, as:

  \begin{center}
    $
    R(b)(i) = \min\nolimits_{\preceq_i}
    \{
    \sup \{ R(b')(j) + \vec{\delta}^{\eta_i}_i 
            \mid (b',j) \in \Amoves{\vec{l}} \}
            \mid \vec{l} \in \Emoves{b,i}
    \}
    $
  \end{center}

  We will denote by $\Phi_E$ the corresponding endofunction on
  $L \to \interval{m} \to \lift{\asc{L}}{m}$ which is defined, for
  $R\colon B_L \to \interval{m} \to \lift{\asc{L}}{m}$, by
  \begin{center}
    $
    \Phi_E(R)(b)(i)  = \min\nolimits_{\preceq_i}
    \{
    \sup \{ R(b')(j) + \vec{\delta}^{\eta_i}_i
    \mid (b',j) \in \Amoves{\vec{l}} \}
    \mid \vec{l}  \in \Emoves{b,i}
    \}
  $
  \end{center}
\end{definition}

Observe that, since $\lift{\asc{L}}{m}$ is a lattice, also the
corresponding set of progress measures, endowed with pointwise
$\preceq$-order, is a lattice. It is immediate to see that $\Phi_E$ is
monotone with respect to such order, i.e., if $R \preceq R'$
pointwise then $\Phi_E(R) \preceq \Phi_E(R')$ pointwise.
This allows us to obtain a complete progress measure as a (least)
fixpoint of $\Phi_E$.

\begin{lemma}[complete progress measure as a fixpoint]
  \label{le:progress-fixpoint}
  Let $L$ be a continuous lattice and let $E$ be a system of $m$ equations over $L$
  of the kind $\vec{x} =_{\vec{\eta}} \vec{f}(\vec{x})$.  
  Then the least solution $R_M$ of the progress measure equations
  (least fixpoint of $\Phi_E$ with respect to $\preceq$) is the least
  $\asc{L}$-progress measure, hence it is smaller than $R_E$ and it is
  complete.
\end{lemma}

\begin{proof}
  We first observe that $\asc{L}$-progress measures $R$ are all and only
  pre-fixpoints of $\Phi_E$.  This implies that $R_M$, which is the least
  pre-fixpoint, is the least progress measure.

  In fact, if $R$ is a pre-fixpoint, i.e., for all $b \in B_L$,
  $i \in \interval{m}$, $\Phi_E(R)(b)(i) \preceq R(b)(i)$, which
  implies $\Phi_E(R)(b)(i) \preceq_i R(b)(i)$.  Then, for $b \in B_L$
  and $i \in \interval{m}$, if $R(b)(i) \neq \err$, necessarily
  $\Phi_E(R)(b)(i) \neq \err$. Hence we can take
  $\vec{l} \in \Emoves{b,i}$ that realises the minimum in the
  definition of $\Phi_E(R)(b)(i)$, namely such that
  $\Phi_E(R)(b)(i) = \sup \{ R(b')(j) + \vec{\delta}^{\eta_i}_i \mid
  (b',j) \in \Amoves{\vec{l}} \}$ and we have that for all
  $(b',j) \in \Amoves{\vec{l}}$
  \begin{center}
    $R(b)(i) \succeq_i \Phi_E(R)(b)(i) \succeq R(b')(j) +
    \vec{\delta}^{\eta_i}_i$
  \end{center}
  which amounts to the validity of the progress measure property (it
  gives strict inequality for $\eta_i = \mu$ and general inequality
  for $\eta_i = \nu$).
  
  Conversely, let $R$ be a progress measure. We have to show that for all $b \in B_L$, $i \in \interval{m}$
  \begin{equation}
    \label{eq:pr-fix}
    R(b)(i) \succeq_i \min\nolimits_{\preceq_i}
    \{
    \sup \{ R(b')(j) + \vec{\delta}^{\eta_i}_i
    \mid (b',j) \in \Amoves{\vec{l}}\}
    \mid \vec{l} \in \Emoves{b,i} \}
  \end{equation}
  Given $b \in L$, $i \in \interval{m}$, by definition of progress
  measure, there is $\vec{l} \in \Emoves{b,i}$ such that for all
  $(b',j) \in \Amoves{\vec{l}}$, it holds $R(b)(i) \succeq_i R(b')(j)$, with
  strict inequality if $\eta_i = \mu$. This can be equivalently stated
  $R(b)(i) \succeq_i R(b')(j) + \vec{\delta}^{\eta_i}$. Hence
  $R(b)(i) \succeq_i \sup \{ R(b')(j) + \vec{\delta}^{\eta_i}_i \mid (b',j)
  \in \Amoves{\vec{l}} \}$.
  Namely, $R(b)(i)$ is larger than an element of the set of which we
  take the minimum, hence (\ref{eq:pr-fix}) immediately follows. Since
  in the right-hand side all entries with an index below $i$ are $0$,
  we even have $\succeq$ (instead of $\succeq_i$ in (\ref{eq:pr-fix}),
  which implies that $R$ is a pre-fixpoint of $\Phi_E$.

  \medskip

  For completeness, recall that by Lemma~\ref{le:small-progress},
  $R_E$ is a $\asc{L}$-progress measure and it is complete. Therefore
  for all $b \in B_L$ and $i \in \interval{m}$, we have
  $R_M(b)(i) \preceq R_E(b)(i)$, from which completeness of $R_M$
  immediately follows.
\end{proof}

Observe that, since
$R_M \preceq R_E$, in particular, for all $b \in B_L$ and
$i \in \interval{m}$, if $R_M(b)(i) \neq \err$, then all components of
$R_M(b)(i)$ corresponding to $\nu$-indices are $0$.

\begin{example}
  \label{ex:running-measure}
  If we consider the system of equations of Example~\ref{ex:running} we
  obtain as least fixpoint the progress measure $R_M(\{a\})(1) = (1,0)$
  while $R_M(\{a\})(2) = R_M(\{b\})(1) = R_M(\{b\})(2) = (0,0)$. Note that
  $R_M$ never assumes the top value $\err$, consistently with the fact
  that the solution is $(u_1, u_2) = (\mathbb{S},\mathbb{S})$.
  We will discuss how $R_M$ is obtained later when providing a more
  ``efficient'' way for computing it.
\end{example}

We next observe that the operator $\Phi_E$ creates monotone
  functions and, applied to functions that respect joins, it
  produces functions enjoying the same property. We first introduce
  the formal definitions.

\begin{definition}[monotonicity and sup-respecting]
  Let $L$ be a lattice. A function
  $R\colon B_L \to \interval{m} \to \lift{\asc{L}}{m}$ is \emph{monotone}
  if for all $b, b' \in L$, $i \in \interval{m}$, if $b \sqsubseteq b'$
  then $R(b)(i) \preceq R(b')(i)$.
  It is \emph{sup-respecting} if for
  all $b \in B_L$ and $X \subseteq B_L$, if
  $b \sqsubseteq \bigsqcup X$ then
  $R(b)(i) \preceq \sup \{ R(b')(i) \mid b' \in X \}$.
\end{definition}

Note that the notion of monotonicity for $R \colon B_L \to \interval{m} \to \lift{\asc{L}}{m}$ is the standard one,  with respect to the pointwise order on $\interval{m} \to \lift{\asc{L}}{m}$.

Observe that $R$ is defined only on the basis elements, which are
possibly (and typically) not closed under joins. The requirement of
being sup-respecting ensures that $R$ extends to a function on $L$
which preserves joins.
  Also note that a sup-respecting function $R$ is always monotone.

\begin{lemma}[$\Phi_E(R)$ is monotone]
  \label{le:phi_e-r-mon}
  Let $L$ be a lattice and let $E$ be a system of $m$ equations over $L$
  of the kind $\vec{x} =_{\vec{\eta}} \vec{f}(\vec{x})$. For every
  function and $R\colon B_L \to \interval{m} \to \lift{\asc{L}}{m}$, the
  function $\Phi_E(R)$ is monotone.
\end{lemma}

\begin{proof}
  Given $b\sqsubseteq b'$ we have to show that
  $\Phi_E(R)(b)(i) \preceq \Phi_E(R)(b')(i)$. Note that
  $\Emoves{b,i} = \{\vec{l}\mid b\sqsubseteq f_i(\vec{l})\} \supseteq
  \{\vec{l}\mid b'\sqsubseteq f_i(\vec{l})\} = \Emoves{b',i}$, hence
  in order to determine $R(b)(i)$ we take the $\min\nolimits_{\preceq_i}$ over
  a larger set, resulting in a smaller vector of ordinals than for
  $R(b')(i)$.
\end{proof}

The fact that $\Phi_E$ preserves sup-respecting functions is proved
  in Lemma~\ref{le:sup-respecting-preserve} in
  Appendix~\ref{ssec:sup-respecting}.

\subsection{Computing Progress Measures}

\subsubsection{Selections}
\label{ssec:selections}
  
In principle, at least on finite lattices, the previous results allow
one to compute the progress measure and thus to prove properties of the
solutions of the system of equations.
However, the computation can be quite inefficient due to the fact that
the existential player has a (uselessly) large number of possible
moves. In fact, given a system
$\vec{x} =_{\vec{\eta}} \vec{f}(\vec{x})$ on a lattice $L$, from a
position $(b,i)$, given any move $\vec{l} \in \Emoves{b,i}$ for player
$\exists$, i.e., any tuple such that $b \sqsubseteq f_i(\vec{l})$, it
is immediate to see that all $\vec{l}'$ such that
$\vec{l} \sqsubseteq \vec{l}'$ are valid moves for $\exists$, since by
monotonicity of $f_i$ we have
$b \sqsubseteq f_i(\vec{l}) \sqsubseteq f_i(\vec{l}')$. In other
words, $\Emoves{b,i}$ is upward-closed. However, player $\exists$, in
order to win, has to try to descend as much as possible, hence playing
large elements is inconvenient.

We next introduce some machinery that formalises the above intuition
and allows us to make the calculation more efficient. The idea is
discussed for a single function first, and then for a system of
equations. For this we need some additional notation. Given a monotone
function $f \colon  L^m \to L$ and $b \in B_L$, we write
$\Emoves{b,f} = \{ \vec{l} \mid \vec{l} \in L^m\ \land\ b \sqsubseteq
f(\vec{l})\}$.

\begin{definition}[selection]
  \label{de:selection}
  Let $L$ be a lattice. Given a monotone function $f\colon L^m \to L$, a
  \emph{selection for $f$} is a function
  $\sigma\colon B_L \to \Pow{L^m}$ such that for all $b \in B_L$ it
  holds $\Emoves{b,f} = \filter{\sigma(b)}$.
  Given a system $E$ of $m$ equations on $L$ of the kind
  $\vec{x} =_{\vec{\eta}} \vec{f}(\vec{x})$, a \emph{selection for
    $E$} is an $m$-tuple of functions
  $\vec{\sigma}$ such that, for each $i \in \interval{m}$, the function $\sigma_i$ is
  a selection for $f_i$.
\end{definition}

Intuitively, a selection provides for each element of the
basis and function $f_i$, a subset of the moves $\Emoves{b,i}$ that are sufficient to ``cover'' $b$ in all possible ways.
Indeed, we can show that when computing the complete progress measure
$R_M$ according to the equations in Lemma~\ref{le:progress-fixpoint},
we can restrict the moves of the existential player to a selection.
Dually, since the moves of the universal player $\Amoves{\vec{l}}$ are
downward-closed and the progress measures of interest are monotone
(see Lemma~\ref{le:phi_e-r-mon}), we can restrict also such
moves to a subset whose downward-closure is $\Amoves{\vec{l}}$.

\begin{lemma}[progress measure on a selection]
  \label{le:progress-fixpoint-selection}
  Let $L$ be a continuous lattice, let $E$ be a system of equations
  over $L$ of the kind $\vec{x} =_{\vec{\eta}} \vec{f}(\vec{x})$ and
  let $\vec{\sigma}$ be a selection for $E$. Moreover, for all
  $\vec{l} \in L^m$ let
  $\Amovesred{\vec{l}} \subseteq B_L \times \interval{m}$ be such that
  $\Amoves{\vec{l}} = \{ (b', i) \mid (b, i) \in \Amovesred{\vec{l}}\
  \land\ b' \sqsubseteq b\}$.
  The system of equations over
  the lattice $\lift{\asc{L}}{m}$, defined, for $b \in L$, $i \in \interval{m}$,
  as:

  \begin{center}
    $
    R(b)(i) = \min\nolimits_{\preceq_i}
    \{
    \sup \{ R(b')(j) + \vec{\delta}^{\eta_i}_i 
            \mid (b',j) \in \Amovesred{\vec{l}} \}
    \mid \vec{l} \in \sigma_i(b) \}
    $
  \end{center}
  has the same least solution as
  that in Lemma~\ref{le:progress-fixpoint}.
\end{lemma}

\begin{proof}
  Let $\Phi_E'$ be the operator associated with the equations in the
  statement of the lemma. We prove that $\Phi_E$ and $\Phi_E'$ have the
  same fixpoint by showing that they coincide on monotone $R$'s.

  Let $b \in B_L$ and $i \in \interval{m}$.
  Let us write
  \begin{eqnarray*}
    \vec{\beta}_b = \Phi_E'(R)(b)(i) & = & \min\nolimits_{\preceq_i}\{ \sup \{
    R(b')(j) + \vec{\delta}^{\eta_i}_i \mid (b',j) \in
    \Amovesred{\vec{l}} \} \mid \vec{l} \in \sigma_i(b) \} \\
    \vec{\gamma}_b = \Phi_E(R)(b)(i) & = & \min\nolimits_{\preceq_i} \{ \sup \{
    R(b')(j) + \vec{\delta}^{\eta_i}_i \mid (b',j) \in
    \Amoves{\vec{l}}\} \mid \vec{l} \in \Emoves{b,i} \}
  \end{eqnarray*}
  and we show $\vec{\beta}_b =_i \vec{\gamma}_b$. First observe that
  $\sup \{ R(b')(j) + \vec{\delta}^{\eta_i}_i \mid (b',j) \in
  \Amovesred{\vec{l}} \} = \sup \{ R(b')(j) + \vec{\delta}^{\eta_i}_i
  \mid (b',j) \in \Amoves{\vec{l}}\}$ since $\Amoves{\vec{l}}$ is the
  downward-closure of $\Amovesred{\vec{l}}$ and $R$ is monotone. Then,
  the fact that $\vec{\gamma}_b \preceq_i \vec{\beta}_b$ follows from
  the observation that, by Definition~\ref{de:selection},
  $\sigma_i(b) \subseteq \Emoves{b,i}$, i.e., the first is a minimum
  over a smaller set. The converse inequality follows from the fact
  that, by Definition~\ref{de:selection}, for each $b \in B_L$,
  $i \in \interval{m}$ if $\vec{l} \in \Emoves{b,i}$ then there exists
  $\vec{l'} \in \sigma_i(b)$ such that $\vec{l'} \sqsubseteq \vec{l}$,
  hence $\Amoves{\vec{l}'} \subseteq \Amoves{\vec{l}}$ and thus
  $\sup \{ R(b')(j) + \vec{\delta}^{\eta_i}_i \mid (b',j) \in
  \Amoves{\vec{l}'}\} \preceq_i \sup \{ R(b')(j) +
  \vec{\delta}^{\eta_i}_i \mid (b',j) \in \Amoves{\vec{l}}\}$
\end{proof}

Since the complete progress measure $R_M$ witnesses the existence of a
winning strategy for $\exists$, the above result implies that whenever
$\exists$ has a winning strategy, it has one also in the game where
the moves of $\exists$ are restricted to be in the selection. A
similar property holds for $\forall$ and $\Amovesred{\vec{l}}$.

Clearly, for computational purposes, we are interested in having the
selections as small as possible. 
Given a monotone function
$f\colon L^m \to L$, and two selections $\sigma, \sigma'\colon B_L \to \Pow{L^m}$ for $f$, we
write $\sigma \subseteq \sigma'$ if for all $b \in B_L$ it holds
$\sigma(b) \subseteq \sigma'(b)$. We will use the same notation for the
pointwise order on selections for systems of equations.

\begin{example}[selections for $\mu$-calculus operators]
  \label{ex:selection}
  Given a transition system $(\mathbb{S}, \to)$, consider the
  powerset lattice $\Pow{\mathbb{S}}$ ordered by subset inclusion, with basis
  $B_{\Pow{\mathbb{S}}} = \{ \{s\} \mid s \in \mathbb{S} \}$.
  Then standard $\mu$-calculus operators admit a least selection, as
  detailed below.

  \begin{itemize}
    
  \item Given $f \colon (\Pow{\mathbb{S}})^2 \to \Pow{\mathbb{S}}$ defined by
    $f(X_1, X_2) = X_1 \cup X_2$, then
    $\sigma\colon B_{\Pow{\mathbb{S}}} \to \Pow{(\Pow{\mathbb{S}})^2}$ is
    $\sigma(\{s\}) = \{ (\emptyset, \{s\}), (\{s\}, \emptyset)\}$

  \item Given $f \colon (\Pow{\mathbb{S}})^2 \to \Pow{\mathbb{S}}$ defined by
    $f(X_1, X_2) = X_1 \cap X_2$, then
    $\sigma \colon B_{\Pow{\mathbb{S}}} \to \Pow{(\Pow{\mathbb{S}})^2}$ is
    $\sigma(\{s\}) = \{ (\{s\}, \{s\})\}$
  
  \item Given $f \colon \Pow{\mathbb{S}} \to \Pow{\mathbb{S}}$ defined by $f(X) = \Diamond X$,
    then $\sigma \colon B_{\Pow{\mathbb{S}}} \to \Pow{\Pow{\mathbb{S}}}$ is
    $\sigma(\{s\}) = \{ \{s'\} \mid s \to s' \}$

  \item Given $f \colon \Pow{\mathbb{S}} \to \Pow{\mathbb{S}}$ defined by $f(X) = \Box X$,
    then $\sigma \colon B_{\Pow{\mathbb{S}}} \to \Pow{\Pow{\mathbb{S}}}$ is
    $\sigma(\{s\}) = \{ \{s' \mid s \to s' \} \}$

  \end{itemize}
\end{example}

We next provide sufficient conditions for a function to admit a
least selection.

\begin{lemma}[existence of least selections]
  \label{le:minimising-selections}
  Let $L$ be a lattice with a basis $B_L$ and let $f \colon L^m \to L$ be a
  monotone function. If $f$ preserves the meet of descending
  chains, then it admits a least selection $\sigma_m$ that maps each $b \in B_L$ to the set of minimal elements of $\Emoves{b,f}$.
\end{lemma}

\begin{proof}
  Assume that $f$ preserves the meet of descending chains.
  First observe that given a descending chain
  $(\vec{l}_{\alpha})_\alpha$ in $\Emoves{b,f}$ we have that
  $\bigsqcap_\alpha \vec{l}_\alpha \in \Emoves{b,f}$.
  In fact, for each $\alpha$ we have
  $b \sqsubseteq f(\vec{l}_\alpha)$ and thus
  $b \sqsubseteq \bigsqcap_\alpha f(\vec{l}_\alpha) =
  f(\bigsqcap_\alpha \vec{l}_\alpha)$.

  The above implies that for each $\vec{l} \in \Emoves{b,f}$ there exists
  $\vec{l}' \in \Emoves{b,f}$, minimal, such that
  $\vec{l}' \sqsubseteq \vec{l}$. In fact, consider the (possibly
  transfinite) chain of tuples $\vec{l}_\alpha$ in $\Emoves{b,f}$ defined
  as follows. Start from $\vec{l}_0 = \vec{l}$. For any ordinal
  $\alpha$, if there is $\vec{l}' \in \Emoves{b,f}$, such that
  $\vec{l}' \neq \vec{l}_\alpha$ and
  $\vec{l}' \sqsubseteq \vec{l}_\alpha$, let
  $\vec{l}_{\alpha+1} = \vec{l}'$. If $\alpha$ is a limit ordinal 
  $\vec{l}_\alpha = \bigsqcap_{\beta<\alpha} \vec{l}_\beta$.
  
  This is a strictly descending chain, that thus necessarily stops at
  some ordinal $\lambda$ bounded by the length of the longest
  descending chain in $L$. By construction
  $\vec{l}_\lambda \sqsubseteq \vec{l}$,
  $\vec{l}_\lambda \in \Emoves{b,f}$ and it is minimal in $\Emoves{b,f}$.

  Define $\sigma_m(b)$ as the set of minimal elements of $\Emoves{b,f}$
  for each $b \in B_L$. It is immediate to see that this is a
  selection. Moreover, it is the least selection. In fact, let
  $\sigma'$ be another selection for $f$. Let
  $\vec{l} \in \sigma_m(b)$. Since $b \sqsubseteq f(\vec{l})$ and
  $\sigma'$ is a selection, there is $\vec{l}' \in \sigma'(b)$ such
  that $\vec{l}' \sqsubseteq \vec{l}$. Now,
  $b \sqsubseteq f(\vec{l}')$ and thus there must be
  $\vec{l}'' \in \sigma_m(b)$ such that
  $\vec{l}'' \sqsubseteq \vec{l}'$. Therefore by transitivity
  $\vec{l}'' \sqsubseteq \vec{l}$, but $\vec{l}$ is minimal and thus
  $\vec{l} = \vec{l}' = \vec{l}'' \in \sigma'(b)$. This shows
  $\sigma_m(b) \subseteq \sigma'(b)$ for all $b \in B_L$. Thus
  $\sigma_m \subseteq \sigma'$, as desired.
\end{proof}

\begin{example}
  \label{ex:running-selection}
  Consider our running example in Example~\ref{ex:running}. Minimal
  selections for the functions $f_1$ and $f_2$ associated with the
  first and second equation are given by
  \begin{itemize}
  \item $\sigma_1(\{a\}) = \{ (\{a\},\emptyset), (\{b\}, \emptyset)\}$
    and $\sigma_1(\{b\}) = \{ (\emptyset, \emptyset)\}$;
  \item $\sigma_2(\{a\}) = \{ (\{a\}, \{a,b\})\}$ and
    $\sigma_2(\{b\}) = \{ (\{b\}, \{b\})\}$.
  \end{itemize}
  Observe that the winning strategy for $\exists$ discussed in
  Example~\ref{ex:running-a} is a subset of the selection. We already
  noticed that this is a general fact: if a winning strategy exists,
  we can find one that is a subset of any given selection.
\end{example}

Selections can be constructed ``compositionally'', i.e., if a function
$f$ arises as the composition of some component functions then we can
derive a selection for $f$ from selections of the components.
The details are presented in Appendix \ref{ssec:compositional-selection}.

\subsubsection{A Logic for Characterising the Moves of the Existential Player}
\label{sec:logic-selection}

The set of possible moves of the existential player is an
upward-closed set in the lattice. Such sets can be conveniently
represented and manipulated in logical form (see,
e.g.,~\cite{DR:SRUCS}). Intuitively, (minimal) selections describe a
disjunctive normal form, but more compact representations can be
obtained using arbitrary nesting of conjunction and disjunction.
For instance, the minimal selection for the monotone function
$f(X_1,\dots,X_{2n}) = (X_1\cup X_2)\cap (X_3\cup X_4)\cap \dots \cap
(X_{2n-1}\cup X_{2n})$ would be of exponential size (think of the
corresponding disjunctive normal form), but we can easily give a
formula of linear size.

This motivates the introduction of a propositional logic for
expressing the set of moves of the existential player along with a
technique for deriving the fixpoint equations for computing the
progress measure, avoiding the potential exponential explosion.

\begin{definition}[logic for upward-closed sets]
  Let $L$ be a continuous lattice and let $B_L$ be a basis for $L$.
  Given $m \in \mathbb{N}$, the logic $\mathcal{L}_m(B_L)$ has
  formulae defined as follows, where $b\in B_L$ and
  $j \in\interval{m}$:
  \[
    \phi ::= \atom{b}{j} \mid \bigvee_{k\in K} \phi_k \mid
    \bigwedge_{k\in K} \phi_k
  \]
  We will write $\mathit{true}$ for the empty conjunction.
  The semantics of a formula $\phi$ is an upward-closed set
  $\semlog{\phi} \subseteq L^m$, defined as follows:
    \begin{eqnarray*}
      \semlog{\atom{b}{j}} & = & \{\vec{l}\in L^m \mid b \sqsubseteq
      l_j\} \\
      \semlog{\bigvee_{k\in K} \phi_k} & = & \bigcup_{k\in
        K} \semlog{\phi_k} \\
      \semlog{\bigwedge_{k\in K} \phi_k} & = & \bigcap_{k\in K}
      \semlog{\phi_k} = \big\{\bigsqcup_{k\in K} m_k \mid
      m_k\in \semlog{\phi_k}, k\in K\big\} \\
      & = & \big\{\bigsqcup_{k\in K} f(k) \mid f \colon K \to L\ \land\
      \forall k \in K.\, f(k) \in  \semlog{\phi_k} \big\}
  \end{eqnarray*}
  The last equality in the definition above holds since every formula
  represents an upward-closed set.  
\end{definition}

It is easy to see that indeed each upward-closed set is denoted by a
formula, showing that the logic is sufficiently expressive.

\begin{lemma}[formulae for upward-closed sets]
  Let $L$ be a continuous lattice with basis $B_L$ and let
  $X \subseteq L^m$ be upward-closed. Then $X = \semlog{\phi}$
  where $\phi$ is the formula in $\mathcal{L}_m(B_L)$ defined as
  follows:
  \[\phi = \bigvee_{\vec{l}\in X} \bigwedge \big\{ \atom{b}{j} \mid
    j\in\interval{m}\ \land\ b \sqsubseteq l_j\big\} . \]
\end{lemma}

\begin{proof}
  We have to show that $\semlog{\phi} = X$:
  \begin{itemize}
  \item ($\subseteq$) Let $\vec{l}'\in \semlog{\phi}$,
    hence
    \[
      \vec{l}' \in \bigcup_{\vec{l}\in X} \bigcap \{\{\vec{l}'' \in
      L^m\mid b\sqsubseteq l''_k\} \mid k \in\interval{m} \land
      b\sqsubseteq l_k\}.
    \]

    Hence there exists $\vec{l} \in X$ such that for all
    $j \in \interval{m}$ and $b \sqsubseteq l_j$ it holds that
    $b \sqsubseteq l'_j$. Then
    \[
      l_j = \bigsqcup \{b \mid b\sqsubseteq l_j\} \sqsubseteq
      \bigsqcup \{b \mid b\sqsubseteq l'_j\} = l'_j.
    \]
    Hence $\vec{l}\sqsubseteq \vec{l}'$ and since $X$ is upward-closed
    $\vec{l}'\in X$.

  \item ($\supseteq$) Let $\vec{l}\in X$. We show that
    $\vec{l}\in \semlog{\psi_{\vec{l}}}$ where
    $\psi_{\vec{l}} = \bigwedge \big\{ \atom{b}{j} \mid j\in\interval{m}\  \land\ b \sqsubseteq l_j\big\}$. In fact
    \[ 
      \semlog{\psi_{\vec{l}} } = \bigcap \{\{\vec{l}'\in
      L^m\mid b \sqsubseteq l'_j\} \mid j \in \interval{m} \land
      b \sqsubseteq l_j\}. 
    \]
    Now, if $j\in\interval{m}$ and $b \sqsubseteq l_j$ then clearly
    $\vec{l} \in \{ \vec{l}' \mid b \sqsubseteq l'_j\}$ and hence
    $\vec{l}$ is contained in the intersection.
  \end{itemize}
\end{proof}

For
practical purposes we should restrict to
finite formulae. This can surely be done in the case of finite lattices, but also for well-quasi orders (see, e.g.,~\cite{DR:SRUCS}).

\begin{definition}[symbolic $\exists$-moves]
  Let $L$ be a continuous lattice and let $f\colon L^m \to L$ be a
  monotone function. A \emph{symbolic $\exists$-move}
  for $f$ is a family $(\phi_b)_{b\in B_L}$ of formulae in
  $\mathcal{L}_m(B_L)$ such that $\semlog{\phi_b} = \Emoves{b, f}$ for
  all $b \in B_L$.

  If $E$ is a system of $m$ equations of the kind
  $\vec{x} =_{\vec{\eta}} \vec{f}(\vec{x})$ over a continuous lattice
  $L$, a \emph{symbolic $\exists$-move} for $E$ is a family of formulae
  $(\phi_b^i)_{b\in B_L, i \in \interval{m}}$ such that for all
  $i \in \interval{m}$, the family $(\phi_b^i)_{b\in B_L}$ is a
  symbolic $\exists$-move for $f_i$.
\end{definition} 

Interestingly, symbolic $\exists$-moves can be obtained compositionally, namely, the formulae corresponding to a functions arising as a composition can be obtained from those of the components.

\begin{lemma}[symbolic $\exists$-moves, compositionally]
  Let $L$ be a continuous lattice with a basis $B_L$, and let
  $f\colon L^n \to L$, $f_j\colon L^m \to L$ for $j \in \interval{n}$
  be monotone functions and let $(\phi_b)_{b \in B_L}$,
  $(\phi^j_b)_{b \in B_L}$, $j \in \interval{n}$ be symbolic
  $\exists$-moves for $f, f_1, \ldots, f_n$.  Consider the function
  $h\colon L^m \to L$ obtained as the composition
  $h(\vec{l}) = f(f_1(\vec{l}), \ldots, f_n(\vec{l}))$. Define
  $(\phi'_b)_{b \in B_L}$ as follows. For all $b \in B_L$, the formula
  $\phi'_b$ is obtained from $\phi_b$ by replacing each occurrence of
  $[b',j]$ by $\phi^j_{b'}$. Then $(\phi'_b)_{b \in B_L}$ is a
  symbolic $\exists$-move for $h$.
\end{lemma}

\begin{proof}
  We first show that given a formula $\phi \in \mathcal{L}_n(B_L)$,
  if $\phi'$ is the formula in $\mathcal{L}_m(B_L)$ obtained from
  $\phi$ by replacing each occurrence of an atom $\atom{b}{j}$ by
  $\phi^j_b$, then

  \[
    \semlog{\phi'} = \{ \vec{l} \mid \vec{l} \in L^m\ \land\
    (f_1(\vec{l}), \ldots, f_n(\vec{l})) \in \semlog{\phi} \}
  \]

  We proceed by induction on $\phi_b$.

  \begin{itemize}
    
  \item ($\phi = \atom{b}{j}$): In this case $\phi' =
    \phi_b^j$. Therefore we have
    \begin{align*}
      \semlog{\phi'}
      & = 
        \semlog{\phi_b^j}\\
      & =  \{ \vec{l} \mid \vec{l} \in L^m\ \land\ b \sqsubseteq f_j(\vec{l}) \} \\
      & =  \{ \vec{l} \mid \vec{l} \in L^m\ \land\  (f_1(\vec{l}), \ldots, f_n(\vec{l})) \in \semlog{\atom{b}{j}} \} \\
      & =  \{ \vec{l} \mid \vec{l} \in L^m\ \land\  (f_1(\vec{l}), \ldots, f_n(\vec{l})) \in \semlog{\phi} \}      
    \end{align*}
    
  \item ($\phi = \bigvee_{k\in K} \phi_k$): In this case
    $\phi' = \bigvee_{k\in K} \phi_k'$, where each $\phi_k'$ is
    obtained from $\phi_k$ by by replacing each occurrence of an atom
    $\atom{b}{j}$ by $\phi^j_b$. Then
    \begin{align*}
      \semlog{\phi'}
      & = 
        \semlog{\bigvee_{k\in K} \phi_k'}\\
      & = \bigcup_{k\in K} \semlog{\phi_k'}\\
      & = \bigcup_{k\in K} \{ \vec{l} \mid \vec{l} \in L^m\ \land\
        (f_1(\vec{l}), \ldots, f_n(\vec{l})) \in \semlog{\phi_k} \}
      & \mbox{[by inductive hyp.]}\\
      & = \{ \vec{l} \mid \vec{l} \in L^m\ \land\
    (f_1(\vec{l}), \ldots, f_n(\vec{l})) \in  \bigcup_{k\in K} \semlog{\phi_k} \}\\
      & = \{ \vec{l} \mid \vec{l} \in L^m\ \land\
    (f_1(\vec{l}), \ldots, f_n(\vec{l})) \in  \semlog{\bigvee_{k\in K} \phi_k} \}\\
      & =  \{ \vec{l} \mid \vec{l} \in L^m\ \land\  (f_1(\vec{l}), \ldots, f_n(\vec{l})) \in \semlog{\phi} \}      
    \end{align*}
    as desired.

  \item ($\phi = \bigwedge_{k\in K} \phi_k$): Analogous.
  \end{itemize}

  \smallskip

  Now, given $b \in B_L$, we have to show that
  \begin{center}
    $\semlog{\phi'_b } = \Emoves{b, h} = \{ \vec{l} \mid
    \vec{l} \in L^m\ \land\ b \sqsubseteq h(\vec{l})\} = \{ \vec{l}
    \mid \vec{l} \in L^m\ \land\ b \sqsubseteq f(f_1(\vec{l}), \ldots,
    f_n(\vec{l}))\}$.
  \end{center}

  This is almost immediate. In fact
  \begin{align*}
    & \semlog{\phi'_b } =\\
    & \quad = \{ \vec{l} \mid \vec{l} \in L^m\ \land\ (f_1(\vec{l}), \ldots,
    f_n(\vec{l})) \in \semlog{\phi_b} \} & \mbox{[by the property proved above]}\\
    & \quad  = \{ \vec{l}
    \mid \vec{l} \in L^m\ \land\ b \sqsubseteq f(f_1(\vec{l}), \ldots,
    f_n(\vec{l}))\}  & \mbox{[by def. of symbolic $\exists$-move]}
  \end{align*}

\end{proof}

\begin{example}
  \label{ex:selection-formulae}
  Consider again our running example in Example~\ref{ex:running}. The
  selections specified in Example~\ref{ex:running-selection} can be
  expressed in the logic as follows:
  \begin{eqnarray*}
    \phi_{\{a\}}^1 = \atom{\{a\}}{1}\lor \atom{\{b\}}{1} &&
    \phi_{\{b\}}^1 = \mathit{true}\\
    \phi_{\{a\}}^2 = \atom{\{a\}}{1}\land \atom{\{a\}}{2} \land
    \atom{\{b\}}{2}  &&
    \phi_{\{b\}}^2 = \atom{\{b\}}{1}\land \atom{\{b\}}{2}
  \end{eqnarray*}
  These formulae can be obtained compositionally. For instance the
  formula $\phi_{\{a\}}^2$ for the equation
  $x_2 =_\nu x_1\cap\Box x_2$ is obtained by combining a logical
  formula for $x_1$ (namely $\atom{\{a\}}{1}$) via conjunction with a
  logical formula for $\Box x_2$
  (namely $\atom{\{a\}}{2} \land \atom{\{b\}}{2}$).
\end{example}

A symbolic $\exists$-move for a system can be directly converted into
a recipe for evaluating the fixpoint expressions for progress
measures. Essentially, every disjunction simply has to be replaced by
a minimum and every conjunction by a supremum (although the proof,
which relies on complete distributivity of the lattice
  $\lift{\asc{L}}{m}$ is not trivial). Furthermore, in the case of an
algebraic lattice, where we can ensure that the elements of the basis
are compact, an atom translates to a straightforward lookup of the
progress measure without additional computation.

\begin{proposition}[progress measure from symbolic $\exists$-moves]  
  \label{prop:evaluation-logic-selections}
  Let $E$ be a system of $m$ equations over a continuous lattice $L$
  and let $B_L$ be a basis for $L$.  Let 
  $(\phi_b^i)_{b\in B_L,i\in\interval{m}}$ be a symbolic $\exists$-move for $E$.

  Then the system of fixpoint equations for computing the progress measure 
  can be written, for all $b \in B_L$ and $i \in \interval{m}$, as
  $R(b)(i) = R^i_{\phi_b^i}$ where 
  $R^i_\psi$ is defined inductively as follows:
  \begin{align*}
    R^i_{\atom{b}{j}} = \min\nolimits_{\preceq_i} \{ \sup \{ R(b')(j) +
      \vec{\delta}^{\eta_i}_i\mid b'\ll b\}\} & \quad &
    R^i_{\bigvee_{k\in K} \phi_k} = \min_{k\in K} R^i_{\phi_k}
    & \quad &
    R^i_{\bigwedge_{k\in K} \phi_k} = \sup_{k\in K} R^i_{\phi_k} 
  \end{align*}
  Whenever the basis element $b$ is compact it holds that
  $R^i_{\atom{b}{j}} = \min\nolimits_{\preceq_i} \{ R(b)(j) +
  \vec{\delta}^{\eta_i}_i \}$.
\end{proposition}

\begin{proof}
  First observe that due to Lemma~\ref{le:sup-respecting-preserve}
  $\Phi_E$ preserves sup-respecting progress measures. Furthermore the
  supremum of sup-respecting progress measures is again
  sup-respecting. This means that the fixpoint iteration generates
  only sup-respecting functions and we can in the following assume
  that $R$ is sup-respecting.
  
  Since $(\phi_b^i)_{b\in B_L,i\in\interval{m}}$ is a symbolic
  $\exists$-move for $E$, the equations of
  Definition~\ref{de:progress-fixpoint} can be written as
  \[
    R(b)(i) = \min\nolimits_{\preceq_i} \{ \sup \{ R(b')(j) +
    \vec{\delta}^{\eta_i}_i \mid (b',j)\in \Amoves{\vec{l}} \}
    \mid \vec{l} \in \semlog{\phi_b^i}
    \}.
  \]
  
  We conclude by proving that, when $R$ is monotonic
  \[
    R^i_{\psi} = \min\nolimits_{\preceq_i} \{ \sup \{ R(b')(j) +
    \vec{\delta}^{\eta_i}_i \mid (b',j)\in \Amoves{\vec{l}} \}
    \mid \vec{l} \in \semlog{\psi}
    \}.
  \]
  We proceed by induction on the structure of $\psi$.
  
  \begin{itemize}
  \item ($\psi = \atom{b}{k}$): By definition
    \begin{align*}
      \min\nolimits_{\preceq_i} \{ \sup \{ R(b')(j) +
      \vec{\delta}^{\eta_i}_i \mid (b',j) \in \Amoves{\vec{l}} \}
      \mid \vec{l} \in \semlog{\atom{b}{k}} \} = \\
      \min\nolimits_{\preceq_i} \{ \sup \{ R(b')(j) + \vec{\delta}^{\eta_i}_i
      \mid j \in \interval{m}\ \land\ b' \ll l_j \} \mid \vec{l} \in
      L^m\ \land\ b \sqsubseteq l_k \}
    \end{align*}
    A vector $\vec{l}\in L^m$ satisfying $b\ll l_k$ has the form
    $(l_1,\dots,l_m)$ where $l_j$ is arbitrary 
    if $j\neq k$ and $b \ll l_k$.
    Since we can assume that $R$ is monotonic and hence the inner
    supremum is monotone in $\vec{l}$, we can conclude that the
    minimum is reached for a vector $\vec{\ell}$ where $l_j=\bot$ if
    $j\neq k$ and $b\sqsubseteq l_k$. Hence we obtain
    \[ \min\nolimits_{\preceq_i} \{ \sup \{ R(b')(j) + \vec{\delta}^{\eta_i}_i
      \mid j \in \interval{m}\  \land\ b' \ll l_j \}
      \mid \vec{l}\in L^m, b\sqsubseteq l_k, l_j=\bot \mbox{ if }j\neq k\}. \]
    Since there is no basis element $b'$ with $b'\ll \bot$, it is
    sufficient if one takes the inner suprema only for elements with
    $j=k$ and $b'\sqsubseteq l_k$. And so we obtain
    \[ \min\nolimits_{\preceq_i} \{ \sup \{ R(b')(k) +
      \vec{\delta}^{\eta_i}_i\mid b'\ll l\}\mid l \in L\ \land\
      b\sqsubseteq l \}
    \]
    We can now infer that $b$ is the least value $l\in L$ such that
    $b\sqsubseteq l$ and hence -- again by monotonicity -- the above
    can be rewritten as
    \[ \min\nolimits_{\preceq_i} \{ \sup \{ R(b')(k) +
      \vec{\delta}^{\eta_i}_i\mid b'\ll b\} \}
    \]
    which is exactly
    $R^i_{\atom{b}{k}}$, as desired.
    
    If $b$ is compact, we know that $b$ itself is the least element of
    all $l$ such that $b\ll l$ and we can write the above as
    \[ R^i_{\atom{b}{k}} = \min\nolimits_{\preceq_i} \{ R(b)(k) +
      \vec{\delta}^{\eta_i}_i\}. \]
  \item Disjunction:
    \begin{eqnarray*}
      R^i_{\bigvee_{k\in K} \phi_k} & = & \min\nolimits_{\preceq_i} \{ \sup \{
      R(b')(j) + \vec{\delta}^{\eta_i}_i \mid j \in \interval{m}\
      \land\ b' \ll l_j\} \mid \vec{l} \in \bigcup_{k\in K}
      \semlog{\phi_k} \} \\
      & = & \min\nolimits_{\preceq_i} \{ \sup \{
      R(b')(j) + \vec{\delta}^{\eta_i}_i \mid j \in \interval{m}\
      \land\ b' \ll l_j\} \mid \vec{l} \in 
      \semlog{\phi_k}, k\in K \} \\
      & = & \min_{k\in K} \min\nolimits_{\preceq_i} \{ \sup \{
      R(b')(j) + \vec{\delta}^{\eta_i}_i \mid j \in \interval{m}\
      \land\ b' \ll l_j\} \mid \vec{l} \in 
      \semlog{\phi_k}\} \\
      & = & \min_{k\in K} R^i_{\phi_k}
    \end{eqnarray*}
  \item Conjunction: since every set $\semlog{\phi_k}$ is
    upward-closed we can immediately apply
    Lemma~\ref{lem:exchange-min-sup} and obtain
    \begin{eqnarray*}
      R^i_{\bigwedge_{k\in K} \phi_k} & = & \min\nolimits_{\preceq_i} \{ \sup \{
      R(b')(j) + \vec{\delta}^{\eta_i}_i \mid j \in \interval{m}\
      \land\ b' \ll l_j\} \mid \vec{l} \in \bigcap_{k\in K}
      \semlog{\phi_k} \} \\
      & = & \sup_{k\in K} \min\nolimits_{\preceq_i} \{ \sup \{
      R(b')(j) + \vec{\delta}^{\eta_i}_i \mid j \in \interval{m}\
      \land\ b' \ll l_j\} \mid \vec{l} \in 
      \semlog{\phi_k}\} \\      
      & = & \sup_{k\in K} R^i_{\phi_k}
    \end{eqnarray*}
  \end{itemize}
\end{proof}

Note that the operator $\min\nolimits_{\preceq_i}$ in the definition of
$R^i_{\atom{b}{j}}$ above is just there to ensure that all entries in
positions smaller than $i$ are set to $0$.

\begin{example}
  Using the logical formulae from
  Example~\ref{ex:selection-formulae}, we obtain the following
  equations for the progress measure (where $\max_{\preceq_i}$ works
  analogously to $\min\nolimits_{\preceq_i}$: it sets all vector entries in
  positions smaller than $i$ to $0$):
  \begin{align*}
    R(\{a\})(1) &=
    \min\nolimits_{\preceq_1}\{R(\{a\})(1)+(1,0),R(\{b\})(1)+(1,0)\} &
    R(\{b\})(1) &= (0,0) \\
    R(\{a\})(2) &=
    \max\nolimits_{\preceq_2}\{R(\{a\})(1),R(\{a\})(2),R(\{b\})(2)\} &
    R(\{b\})(2) &= 
    \max\nolimits_{\preceq_2}\{R(\{b\})(1),R(\{b\})(2)\}
  \end{align*}
  The solution for the progress measure equations has already been
  given in Example~\ref{ex:running-measure}.
\end{example}

\subsubsection{Complexity Analysis}
\label{sec:complexity-analysis}

The benefit of the progress measures introduced in~\cite{j:progress-measures-parity} is to ensure that model-checking is
polynomial in the number of states and exponential in (half of) the
alternation depth.
We will now perform a corresponding complexity
analysis for our setting, based on symbolic $\exists$-moves and
by assuming that we are working on a finite lattice.

Let $E$ be a fixed system of $m$ equations over a finite lattice $L$,
let $k$ be the number of $\mu$-equations and let $B_L$ be a basis for
$L$.
Let $(\phi_b^i)_{b\in B_L, i \in \interval{m}}$ be a symbolic
$\exists$-move for $E$ and assume that the size of every such formula
is bounded by $s$.
Note that the formulae are typically of moderate size. For instance,
$\mu$-calculus model-checking, the branching of a transition system
(i.e., the number of successors of a single state) is a determining
factor.  In fact, as it can be grasped from our running example (see
Example~\ref{ex:selection-formulae}), the size of the symbolic
$\exists$-move $\phi_b^i$ will be linear in the number of
propositional operators and, in the presence of modal operators,
linear in the branching degree of the transition system. For arbitrary
monotone functions it is more difficult to give a general rule.

The shape of the formulae in the symbolic $\exists$-move determine how the values of the progress measure at various positions $(b,i)$ of the games are interrelated. These dependencies clearly play a role in the computation and thus are made explicit by following definition.

\begin{definition}[dependency graph]
  Given two game positions $(b,i),(b',j)\in B_L\times\interval{m}$ of
  $\exists$ we say that $(b,i)$ is a \emph{predecessor} of $(b',j)$ if
  $\atom{b'}{j}$ occurs in $\phi_b^i$. We will write
  $\mathit{pred}(b',j)$ for the set of predecessors of $(b',j)$. In
  this situation we will also call the pair $((b,i),(b',j))$ an
  \emph{edge} and refer to corresponding graph as the \emph{dependency
    graph} for $E$.
\end{definition}

As a first step we provide a bound to the number of edges in the dependency graph.

\begin{proposition}[edges in the dependency graph]
  The number $e$ of edges in the dependency graph for system $E$ is
  such that $e\le \min\{|B_L|\cdot m\cdot s,(|B_L|\cdot m)^2\}$, where
  $m$ is the number of equations and $s$ is the bound on the size of
  symbolic $\exists$-moves.
\end{proposition}
}

\begin{proof}
  There are at most $|B_L|\cdot m$ game positions and hence
  the number of edges is obviously bounded by $(|B_L|\cdot
  m)^2$. Moreover, each game position, the number of outgoing edges is
  bounded by the size of the formula (symbolic $\exists$-move)
  associate to the position. Hence the thesis.
\end{proof}

In order to bound the complexity of the overall computation of the
progress measure, first note that the lattice $\lift{\asc{L}}{m}$
contains $(\asc{L}+1)^m + 1$ elements. However only
$h = (\asc{L}+1)^k + 1$ are relevant, since the entries of
$\nu$-indices are always set to $0$. As an example, when
model-checking a $\mu$-calculus formula over a finite state system,
$\asc{L}$ is the size of the state space of the Kripke structure. In
fact, the lattice is $(\Pow{\mathbb{S}}, \subseteq)$ where
$\mathbb{S} = \{ s_0, \ldots, s_n\}$ is the state space, then the
longest ascending chain is
$\emptyset \subseteq \{s_0\} \subseteq \{s_0, s_1\} \subseteq \ldots
\subseteq \mathbb{S}$.

This fact and the observation that we can perform the
fixpoint iteration for the progress measure using a worklist
algorithm on the dependency graph, lead to the following
result.

\begin{theorem}[computing progress measures]
  The time complexity for computing the least fixpoint progress measure
  for system $E$ is $O(s\cdot k\cdot e\cdot h)$, where $s$ is the bound
  on the size of symbolic $\exists$-moves, $k$ is the number of
  $\mu$-equations, $e$ the number of edges in the dependecy graph, and
  $h = (\asc{L}+1)^k + 1$.
\end{theorem}

\begin{proof}
  We use a worklist algorithm, and the worklist initially contains all
  edges.

  Processing an edge $((b,i),(b',j))$ means to update the value
  $R(b)(i)$ by evaluating the formula $\phi_b^i$. Afterwards all edges
  originating from $(b,i)$ can be removed from the worklist. Whenever
  a value $R(b')(j)$ increases, all edges $((b,i),(b',j))$ with
  $(b,i)\in\mathit{pred}(b',j)$ will be again inserted into the
  worklist.
  Hence, at most $\sum_{(b',j)\in B_L\times\interval{m}} h\cdot
  \mathit{pred}(b',j) = e\cdot h$ edges will be inserted into the
  worklist and processed later.
  
  In turn, processing an edge has complexity at most $O(s\cdot k)$,
  since we inductively evaluate a formula of size $s$ on ordinal
  vectors of length $k$. (Since the lattice is finite, it is
  automatically algebraic and the simpler case for compact elements of
  Theorem~\ref{prop:evaluation-logic-selections} applies.)  Everything
  combined, we obtain a runtime of $O(s\cdot k\cdot e\cdot h)$.
\end{proof}

We compare the above with the runtime
in~\cite{j:progress-measures-parity}, which is
$O(dg\big( \frac{n}{d}\big)^{\lceil \frac{d}{2}\rceil})$, where $d$ is
the alternation depth of the formula, $g$ the number of edges and $n$
the number of nodes of the parity game.

The correspondence is as follows: $g$ corresponds to our number $e$
and $n$ to $\asc{L}$ (where we cannot exploit the optimisation by
Jurdzi\'nski which uses the fact that every node in the parity game is
associated with a single parity, leading to the division by
$d$). Furthermore $s$ is a new factor, which is due to the fact that
we are working with arbitrary functions. But this is mitigated by the
fact that we often obtain smaller parity games than in the standard
$\mu$-calculus case (see for instance Example~\ref{ex:running-a},
Figure~\ref{fi:game}). The number $\frac{d}{2}$ corresponds to our
$k$.
However $\frac{d}{2}$ could potentially be strictly lower than $k$,
since we did not take into account the fact that some equations might
not be dependent on other equations.

To incorporate this and possibly further optimisations into the
complexity analysis we need a notion of alternation depth for equation
systems. This can be easily obtained by extending the one introduced
in~\cite{cks:faster-modcheck-mu,Sch:VRS}.  A system of equations can
be split into closed subsystems corresponding to the strongly
connected components of the dependency graph for the system. Then the
alternation depth of the system is defined as the length of the
longest chain of mutually dependent $\mu$ and $\nu$-equations within a
closed subsystem. By solving every component separately we obtain a
more efficient algorithm.

In particular, systems of fixpoint equations that consist only of
$\mu$-equations or $\nu$-equations can be solved by a single fixpoint
iteration on $L^m$, where $m$ is the number of
equations~\cite{v:lectures-mu-calculus}. Similarly, equations with
indices $i,i+1$ where $\eta_i = \eta_{i+1}$ can be merged. This
results in an equation system where subsequent equations alternate
between $\mu$ and $\nu$. (Notice that this transformation means that
the equations are over $L^j$ instead of $L$, but this can be easily
adapted in our setting.)

Note also that the runtime might be substantially improved by finding
a good strategy for computing the progress measure, as spelled out
in~\cite{j:progress-measures-parity}, in the same way as efficient
ways can be found for implementing the worklist algorithm in program
analysis.

\section{Model-Checking Latticed $\mu$-Calculi}
\label{sec:applications-latticed}

As explained earlier, model-checking for $\mu$-calculus
formulae can be reduced to solving fixpoint equations over the
powerset lattice $2^{\mathbb{S}}$ where $\mathbb{S}$ is the state space of the system
under consideration.
A state $x\in \mathbb{S}$ can either satisfy or not satisfy a formula,
meaning it either belongs to the solution or not. However, there are
also multi-valued logics for modelling uncertainty, disagreement or
relative importance in which it is natural to have ``non-binary''
truth values (see,
e.g.,~\cite{kl:latticed-simulation,ekn:bisim-heyting,GLLS:DNMC,Fitting91}). Such
a setting, as detailed later, can also be used to model and verify
conditional (or featured) transition systems with upgrades.
Here we discuss latticed $\mu$-calculi, inspired by the work cited
above, and discuss a corresponding model checking procedure.

A lattice of truth values $L$ is
fixed, which is typically finite. and then formulae are evaluated over
the lattice $L^{\mathbb{S}}$, endowed with the pointwise order.
Also transitions are associated with an element in the
lattice of truth values.

\begin{definition}[multi-valued transition system]
  A \emph{multi-valued transition system} over $L$ is a function
  $R \colon \mathbb{S} \times \mathbb{S}\to L$, where $\mathbb{S}$ is
  the set of states.
\end{definition}

Since $L$ can be non-boolean, multi-valued modal logics express forms
of negation or implication by relying on \emph{residuation} or
\emph{relative pseudo-complement} operation which is well defined for
all complete lattices $L$.

\begin{definition}[residuation]
  Let $L$ be a lattice. Given $l, m \in L$, we define
  $(l\Rightarrow m) = \bigsqcup\{ l'\in L \mid l\sqcap l' \sqsubseteq
  m\}$.
\end{definition}

Latticed $\mu$-calculi use atoms, conjunction, disjunction and
residuation. 
The modal operators $\Diamond$ and $\Box$ are interpreted as
follows. Given $u \in L^{\mathbb{S}}$ we define
$\Diamond u, \Box u \in L^{\mathbb{S}}$ as
\[
  (\Diamond u)(x) = \bigsqcup_{y\in \mathbb{S}} (R(x,y) \sqcap u(y))
  \qquad\qquad (\Box u)(x) = \bigsqcap_{y\in \mathbb{S}} (R(x,y) \Rightarrow
  u(y))
\]

The approach discussed in \S~\ref{ssec:mu-calc} for
model-checking the $\mu$-calculus can be easily adapted to this
setting. Instead of the powerset lattice we now have $L^{\mathbb{S}}$
and, as a basis $B_{L^{\mathbb{S}}}$ we can take the functions 
$b_x \in B_{L^{\mathbb{S}}}$, with $x \in \mathbb{S}$, $b\in B_L$,
defined by $b_x(x) = b$ and $b_x(y) = \bot$ for all $y\neq x$.

In order to perform the calculation of the progress measure
efficiently, we use symbolic $\exists$-moves as defined in
\S~\ref{sec:logic-selection}. Here we assume that $L$ is a finite
distributive lattice. In this case $\ll$ and $\sqsubseteq$ coincide.
Moreover, for finite distributive lattice it is is well-known from the
Birkhoff duality (see also~\cite{dp:lattices-order}) that every
element can be uniquely represented as the join of a downward-closed
set of join-irreducibles. Note that if $B_L$ is the set of
join-irreducibles in $L$, then the basis
$B_{L^\mathbb{S}} = \{b_x\mid x\in \mathbb{S},b\in B_L\}$, given above
is the set of join-irreducibles of ${L^\mathbb{S}}$.

\begin{proposition}[symbolic $\exists$-moves in latticed $\mu$-calculi]
  Let $L$ be a finite distributive lattice, let $B_L$ be the set of
  its join-irreducibles. The
  following are symbolic $\exists$-moves for the semantic functions:

  \begin{itemize}
  \item For
    $\sqcup \colon {L^\mathbb{S}}\times {L^\mathbb{S}}\to {L^\mathbb{S}}$, we let
    $\psi_{b_x}^\sqcup = [b_x,1]\lor [b_x,2]$.

  \item For
    $\sqcap \colon {L^\mathbb{S}}\times {L^\mathbb{S}}\to {L^\mathbb{S}}$, we let
    $\psi_{b_x}^\sqcap = [b_x,1]\land [b_x,2]$.
  \item For $l \Rightarrow \_ \colon {L^\mathbb{S}}\to {L^\mathbb{S}}$ (where $l\in L$ is fixed and seen as
    a constant function $\mathbb{S}\to L$), we let
    $\psi_{b_x}^\Rightarrow = \bigwedge \{ [b'_x,1] \mid b'\sqsubseteq l\ \land
    \ b'\sqsubseteq b \}$.

  \item For $\Diamond  \colon {L^\mathbb{S}}\to {L^\mathbb{S}}$ %
    we let $\psi_{b_x}^\Diamond = \bigvee \{ [b_y,1] \mid y\in Y\ \land
        b\sqsubseteq R(x,y)\}$
    
  \item For $\Box \colon {L^\mathbb{S}}\to {L^\mathbb{S}}$ we let
    $\psi_{b_x}^\Box = \bigwedge \{ [b'_y,1] \mid y\in Y\ \land\ b'\sqsubseteq R(x,y)\ \land  b'\sqsubseteq b\}$.
    
  \end{itemize}
\end{proposition}

\begin{proof}
  We will only consider two cases, since the remaining ones
  cases are analogous.

  \begin{itemize}
  \item \emph{$\sqcup$:} Let $u_1,u_2\in {L^\mathbb{S}}$. Since $b$ and
    hence $b_x$ are join-irreducibles, it holds that
    $b_x \sqsubseteq u_1\sqcup u_2$ iff $b_x \sqsubseteq u_1$
    or $b_x\sqsubseteq u_2$. Hence we can define
    \[ \Emoves{b_x,\sqcup} = \{(u_1,u)\mid b_x\sqsubseteq
    u_1,u\in{L^\mathbb{S}}\}\cup \{(u,u_2)\mid b_x\sqsubseteq
    u_2,u\in{L^\mathbb{S}}\} = \semlog{\psi_{b_x}}. \]
  \item \emph{$\Box$-operator:} Let $u\in {L^\mathbb{S}}$. It holds that
    \begin{eqnarray*}
      && b_x \sqsubseteq \Box u \\
      & \iff & b \sqsubseteq (\Box u)(x) = \bigsqcap_{y\in \mathbb{S}} (R(x,y)
      \Rightarrow u(y)) \\
      & \iff & \mbox{for all $y\in \mathbb{S}$: }
      b\sqsubseteq (R(x,y) \Rightarrow u(y)) \\
      & \iff & \mbox{ for all $y\in \mathbb{S}$: }
      b\sqcap R(x,y) \sqsubseteq u(y) \\
      & \iff & \mbox{ for all $y\in \mathbb{S}$, $b'\in B_L$ with
        $b'\sqsubseteq R(x,y)$,
        $b'\sqsubseteq b$: } b'\sqsubseteq u(y) \\
      & \iff & \mbox{ for all $y\in \mathbb{S}$, $b'\in B_L$ with
        $b'\sqsubseteq R(x,y)$, $b'\sqsubseteq b$: } b'_y\sqsubseteq u
    \end{eqnarray*}
    Note that we used that $(l\Rightarrow m)$ is the maximal element
    in the downward-closed set $\{ l'\in L \mid l\sqcap l' \sqsubseteq
    m\}$, which holds for distributive lattices.
    
    Hence can define
    $\Emoves{b_x,\Box} = \{u\in {L^\mathbb{S}}\mid b_x\sqsubseteq \Box u\} =
    \semlog{\psi_{b_x}}$.
  \end{itemize}
\end{proof}

Note that residuation is only monotone in the second argument and that
distributivity is essential for this definition of symbolic
$\exists$-moves. For instance, if $b$ is not a join-irreducible then
$b\sqsubseteq l_1\sqcup l_2$ is not equivalent to
$b\sqsubseteq l_1 \lor b\sqsubseteq l_2$.

\begin{example}[conditional transition systems with upgrades]
  An interesting special case are conditional transition systems with
  upgrades~\cite{bkks:cts-upgrades} for which a logic satisfying the
  Hennessy-Milner property has been studied in~\cite{p:modal-logic-cts}.
  This logic uses the operators given above, enriched with
  constants. This kind of systems extend the well-known featured
  transition systems for modelling software product
  lines~\cite{ccpshl:simulation-product-line-mc} by upgrades.

  Let $(P,\le)$ be a given partial order where $P$ is the set of
  products and $\le$ is the upgrade relation. If $p\le q$, it is
  possible to make an upgrade from $q$ to $p$ during the runtime of
  the system, i.e., $p$ is the more advanced product compared to
  $q$. We consider the lattice $L = (\mathcal{O}(P),\sqsubseteq)$,
  where $\mathcal{O}(P)$ is the set of all downward-closed subsets of
  $P$. (In fact the sets $\downarrow p$, for $p\in P$, where
  $\downarrow$ denotes downward-closure, are the join-irreducibles of
  $L$.) A transition system that compactly specifies the system
  behaviour for all possible products is given by
  $R\colon \mathbb{S}\times \mathbb{S}\to \mathcal{O}(P)$ where
  $p\in R(x,y)$ means that there exists a transition from $x$ to $y$
  if one is in possession of product $p$. More advanced products lead
  to additional transitions, due to the downward-closure. It is
  possible to spontaneously perform upgrades during runtime.

  Now one can study the latticed modal logic or latticed
  $\mu$-calculus arising in such a setting. Evaluating a formula
  $\phi$ yields a function
  $\sem{\phi} \colon \mathbb{S}\to \mathcal{O}(P)$ which intuitively
  gives us for every state those products on which $\phi$ holds
  (taking upgrades into account).

  The approach outlined in the first part of the section can be
  directly used for model checking the Hennessy-Milner logic on
  product lines. Note that, as it happens in this case, the lattice
  $L$ of truth values can have a considerable size and thus the
  availability of general approaches for handling latticed
  $\mu$-calculi can be of great help.
\end{example}

\section{Solving Fixpoint Equations over Infinite Lattices}
\label{sec:solving-fp-equations-smt}

We present some initial but promising results concerning the solution
of fixpoint equations in infinite lattices. We will mainly concentrate
on equations over the real interval $[0,1]$, as considered also
in~\cite{MS:MS} as a precursor to model-checking PCTL or probabilistic
$\mu$-calculi. We adapt our fixpoint game in a way that it can be
encoded into a finite first-order formula. If this formula is in a
decidable fragment -- such as linear arithmetic -- we can use an SMT
solver to determine its satisfiability. In this way one can either
check that a value is smaller or equal than the solution or even let
the SMT solver calculate the solution.

The starting observation is that the existence of a winning strategy
for  player $\exists$ in the game can be expressed as a
first-order formula with nested quantifiers (existential quantifiers
for the $\exists$ player, universal quantifiers for the $\forall$
player). However, the formula is in general of infinite size, since
plays are unbounded or even infinite. Starting from $(b,i)$ the
formula would be of the kind
$\exists \vec{l}_0 \in \Emoves{b,i}.\, \forall (b_0,i_0) \in
\Amoves{\vec{l}_0}.\ \exists \vec{l}_1 \in \Emoves{b_0,i_0}.\, \forall
(b_1,i_1) \in \Amoves{\vec{l}_0}.\, \ldots$ In order to make the
formula finite we need a stopping condition: if an equation index is
visited for the second time (without any higher index in between), we
make sure that the game can be continued if we jump back to the
previous occurrence of the index. For $\nu$-indices this simply
amounts to checking that the tuple's values seen at the two
occurrences are in the $\sqsubseteq$-relation. For $\mu$-indices the
situation is more complicated: ensuring that we can cycle on that
equation is not sufficient (since by continuing forever $\exists$
would lose) and thus we have to provide a proof that the values truly
decrease and that we will eventually reach $\bot$.
We will see that for lattices based on a well-founded order, it is
sufficient to require a strict inequality, for lattices which do not
enjoy this property (such as the real interval $[0,1]$), we have to
find a different condition. In fact, for the reals we will present a
condition which is correct, i.e., if the formula is satisfiable, we
know that the considered value is bounded by the solution. However,
this method is not always complete.
We will discuss the limitations of the approach and provide a
preliminary characterisation of functions for which we obtain
completeness.

We will first adapt our game (Definition~\ref{def:fp-game}) to a
modified version, which incorporates the stopping condition mentioned
above. We define a game parametrised over a predicate $\decr(v,b,l)$
which takes three lattice elements as parameters.

\begin{definition}[modified fixpoint game]
  \label{def:fp-game-modified}
  Let $L$ be a lattice and let
  $\vec{x} =_{\vec{\eta} } \vec{f}(\vec{x})$ be a system of equations
  over $L$ and let $\decr \subseteq L^3$ be a fixed predicate.
  The game has a state consisting of a vector $\vec{v} \in L^m$, whose
  entries are defined while playing the game, and the current index
  $j$. 

  The game starts on some $(v_i,i) \in L \times \interval{m}$, namely,
  initially $v_i$ is the only component of $\vec{v}$ which is set and
  the current index is set $j:=i$.
  Player $\forall$ chooses $b_i\ll v_i$. At a generic step:

  \begin{itemize}
  \item $\exists$ chooses $\vec{l}\in \Emoves{b_j,j}$
  \item $\forall$ chooses an index $k\in\interval{m}$ and set $j:=k$. Then:
    \begin{itemize}
    \item if the current index $j$ was already set to $k$ earlier in
      the play and no higher index has occurred in between, then there
      are two possibilities:
      \begin{itemize}
      \item if $\eta_k = \nu$ check whether $l_k\sqsubseteq v_k$. If
        yes, $\exists$ wins, otherwise $\forall$ wins.
      \item if $\eta_k = \mu$ check whether $l_k\sqsubseteq b_k$ and
        $\decr(v_k,b_k,l_k)$. If yes, $\exists$ wins, otherwise
        $\forall$ wins.
      \end{itemize}
    \item otherwise set $v_k := l_k$ and $\forall$ chooses
      $b_k\in B_L$ with $b_k\ll v_k$ (hence
      $(b_k,k)\in\Amoves{\vec{l}}$), and continue.
    \end{itemize}
  \end{itemize}
\end{definition}

We can imagine the game as being played on game trees as depicted
below for the case $m=2$.  In the first tree we start with index
$i = 1$ and in the second with $i = 2$ and -- depending on the choice
of $\forall$ -- we descend in the tree. Once we reach a leaf (i.e., a
node with a repeated index with no larger index in between) we can
stop and determine the winner of the game. Note that in the left-hand
tree we need one extra level of nodes, since we cannot yet stop at
the $1$-node on level~three, since there is a higher index ($2$) on
the path between this node and the root.

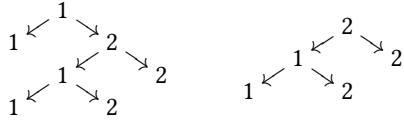
\begin{center}
  \parbox{3cm}{
    \xymatrix@C=8pt@R=0pt{
      & 1 \ar[ld] \ar[rd] & & \\
      1 & & 2 \ar[ld] \ar[rd]  & \\
      & 1 \ar[ld] \ar[rd] & & 2 \\
      1 & & 2 }}
  \qquad
  \parbox{3cm}{
    \xymatrix@C=8pt@R=0pt{
      & & 2 \ar[ld] \ar[rd] & \\
      & 1 \ar[ld] \ar[rd] & & 2 \\
      1 & & 2 } }
\end{center}

It is possible to show that every such tree is finite, which follows
from the fact that there are no infinite paths and that it is finitely
branching.

\begin{lemma}[modified plays are finite]
  The game of Definition~\ref{def:fp-game-modified} does not admit
  infinite plays.
\end{lemma}

\begin{proof}
  Assume that there exists an infinite play, i.e., whenever an index
  is reached for the second time, there is always a higher index in
  between. Consider the suffix of the play which contains only indices
  that occur infinitely often. Assume that $k$ is the highest among
  those indices. Then the play will stop when $k$ is reached for the
  second time in the suffix, which is a contradiction.
\end{proof}

We will next show that a winning strategy in the modified game implies
a winning strategy in the original game. The basic idea is that we
follow the winning strategy in the modified game and once we have
reached a leaf, we ``jump'' back to the predecessor node with the same
index and continue to follow the strategy. A crucial point is to show
that the $\decr$ predicate ensures that there cannot exist an
infinite path where the highest index occurring infinitely often is a
$\mu$-index.

\begin{definition}[well-foundedness]
  \label{def:pred-decr}
  Let $\decr$ be a ternary predicate on $L$. We say that $\decr$ is
  well-founded if there exist no $v \in L$ and $b^m,l^m\in L$ for $m\in\mathbb{N}$
  such that
  \[ b^{m+1}\ll l^m\sqsubseteq b^m \ll v \] and
  $\decr(v,b^m,l^m)$ for all $m$.
\end{definition}

Intuitively we want to ensure that there is no infinite play
underneath a fixed starting value $v$.  Obviously, if the lattice
order $\sqsubset$ is well-founded one can define
$\decr(v,b,l) = (l\sqsubset b)$. (Or even $\mathit{true}$ if the
way-below relation $\ll$ should be well-founded.) For the real
interval $[0,1]$ we need a more sophisticated predicate,
whose
shape will be explained in more detail in
Lemma~\ref{lem:decrease-above-piecewise-linear}.

\begin{lemma}
  \label{lem:decr-wellfounded-reals}
  Let $a_i\in [0,1]$, $i\in \{0, \ldots, \ell\}$ be a finite set of
  real constants with $a_0=0 < a_1 < \dots < a_\ell$ and let
  $c\in [0,1]$.
  Given $v,b,l\in L$ we define that $\decr(v,b,l)$ holds if
  \begin{itemize}
  \item $l = a_i$ for some $i\in \{0, \ldots, \ell\}$ 
  \item \textit{or} $a_\ell \le b$ and $b-l\ge c\cdot (v-b)$ 
  \item \textit{or} $a_i \le b < a_{i+1}$ and
    $b-l\ge c\cdot (a_{i+1}-b)$ for some $i\in \{0, \ldots, \ell\}$
  \end{itemize}
  
  Then $\decr$ is a well-founded predicate.
\end{lemma}

\begin{proof}
  Assume, by contradiction, that there exist $v,b^m,l^m\in L$,
  $m\in\mathbb{N}_0$ such that $b^{m+1} < l^m \le b^m< v$ and
  $\decr(v,b^m,l^m)$ for all $m$.

  Since the sequence is infinite and strictly decreasing, it must have
  a suffix for which $b^m,l^m \neq a_i$ for all $i\in \{0, \ldots, \ell\}$.
  Hence there exists an index $n$ such that for all $m\ge n$ we have
  $b^m-l^m \ge c\cdot (a-b^m)$ where either $a=v$ or $a=a_i$ for
  some $i\in\interval{m}$. Furthermore $b^n < a$.

  We show by induction on $m$ that for these indices
  \[ a - l^m \ge (a-b^n)\cdot (1+(m-n+1)\cdot c) \]
  \begin{itemize}
  \item $m=n$: we know that
    $b^n-l^n \ge c\cdot (a-b^n)$. By
    rearranging we obtain $l^n \le b^n - c\cdot (a-b^n)$. We
    subtract both sides of the inequality from $a$ and get
    $a- l^n \ge a - b^n + c\cdot (a-b^n) = (a-b^n)\cdot
    (1+c)$.
  \item $m-1\to m$: We have
    \begin{eqnarray*}
      a-l^m & = & a - l^{m-1} + l^{m-1} - l^m \\
      & \ge & (a-b^n)\cdot (1+(m-n)\cdot c) + l^{m-1} - l^m \\
      & > & (a-b^n)\cdot (1+(m-n)\cdot c) + b^m - l^m \\
      & \ge & (a-b^n)\cdot (1+(m-n)\cdot c) + c\cdot (a-b^m) \\
      & \ge & (a-b^n)\cdot (1+(m-n)\cdot c) + c\cdot (a-b^n) \\
      & = & (a-b^n)\cdot (1+(m-n+1)\cdot c)
    \end{eqnarray*}
    where the first inequality ($\ge$) is due to the induction
    hypothesis, the second inequality ($>$) holds since
    $l^{m-1} > b^m$, the third inequality ($\ge$) holds because of
    the $\decr$-constraint and the fourth inequality ($\ge$) is
    satisfied since $b^m \le b^n$.
  \end{itemize}
  This implies that
  $l^m \le a - (a-b^n)\cdot (1+(m-n+1)\cdot c)$ for all
  $m$. Since $a-b^n > 0$, this is a contradiction, since the
  right-hand side of the inequality will eventually be negative.
\end{proof}

Naturally, one has to determine suitable constants $a_i,c$. These can
either be derived in some way from the given functions or one can
existentially quantify over the constants.
Note that it is sound to let $\decr$ hold for $l=a_0=0$, since
$\exists$ automatically wins in this case.

We can now show the correctness of the modified game, provided that
$\decr$ is well-founded.

\begin{proposition}[correctness of the modified game]
  \label{prop:modified-game-sound}
  Let $E$ be a system of equations over $L$ of the kind
  $\vec{x} =_{\vec{\eta}} \vec{f}(\vec{x})$ with solution
  $\vec{u} \in L^m$ and let $\decr$ be a well-founded
  predicate.
  Then the modified game is correct: for all
  $(v_i,i) \in L \times \interval{m}$, if $\exists$ has a winning
  strategy in a play starting from $(v_i,i)$ then $v_i \sqsubseteq u_i$.
\end{proposition}

\begin{proof}
  We show that whenever $\exists$ has a winning strategy in the
  modified game for $(v_i,i)$, then $\exists$ has a winning strategy
  in the original game for all $(b_i,i)$ with $b_i\ll v_i$. This
  implies $b_i\sqsubseteq u_i$ and finally
  $v_i = \bigsqcup_{b_i\ll v_i} b_i \sqsubseteq u_i$.

  Let $b_i\ll v_i$. We start the game with $(v_i,i)$ and assume
  that $\forall$ chooses $b_i$ in the first step. Then $\exists$
  follows her winning strategy in the modified game until she reaches
  a leaf, i.e., an index $k$ which already appeared earlier in the
  game and no higher index has occurred in between.  At this point
  $\forall$ will choose $b\ll l_k$ in the original game. Note that
  $l_k\sqsubseteq v_k$ since $\exists$ wins the game: in the case of a
  $\nu$-index this follows directly and in the case of a $\mu$-index
  we have $l_k\sqsubseteq b_k\ll v_k$.

  We will now restart the game after the prefix of the play which ends
  at the first occurrence of the index $k$. We keep the value $v_k$
  but set $b_k := b$. This is a valid choice since $b\ll
  l_k\sqsubseteq v_k$ and hence $b\ll v_k$. 

  Note that $\exists$ always has an available move, since she can move
  in the modified game. It is left to show that she can win infinite
  games: assume that we have an infinite run where the highest index
  that occurs infinitely often is $k$ with $\eta_k = \mu$. Consider
  the run from the point onwards where we do not visit any indices
  $\ell>k$ any more. Then we will eventually find a $k$-index (either
  seeing it for the first time or via a restart). The next occurrence
  of $k$ will be in a restart situation (since the condition that
  there is no higher index in between is automatically satisfied). In
  this case we will verify $\decr(v_k,b_k,l_k)$ for the current values
  $b_k,l_k$, which will be denoted $b^0,l^0$ ($v_k$ is always left
  unchanged). This continues and we obtain lattice elements $b^m,l^m$,
  $m\in\mathbb{N}_0$ where $b^{m+1}\ll l^m$ (since the choice of
  $\forall$ is restricted accordingly) and $l^m\sqsubseteq b^m$ (since
  we check this condition at the restart for each $\mu$-index). All
  are way-below $v_k$. Furthermore $\decr(v_k,b^m,l^m)$ holds for all
  these values. But this is a contradiction to the fact that the
  predicate $\decr$ is well-founded.
\end{proof}

\begin{example}
  \label{ex:smt-solver-discontinuous}
  As an example, before discussing completeness,  consider the
  monotone, but discontinuous function 
  $f\colon [0,1]\to [0,1]$ defined by:
  \[ f(x) = \left\{
      \begin{array}{ll}
        \frac{1}{4} + \frac{1}{2}x & \mbox{if $0\le x<\frac{1}{2}$} \\
        \frac{3}{8} + \frac{1}{2}x & \mbox{if $\frac{1}{2}\le x\le 1$}
      \end{array}
    \right.
  \]
  The graph of the function looks as shown in
  Figure~\ref{fig:function-graph}. The dashed diagonal intersects the
  graph at the position of the only fixpoint.

  We are interested in computing the least fixpoint, i.e., we consider
  the equation $x =_\mu f(x)$. We set $c=1$, $a_1 = 0$,
  $a_2 = \frac{1}{2}$ (the discontinuity point of $f$) and consider
  the corresponding $\decr$-constraint. The basis contains all
  elements of $[0,1]$, apart from $0$. Then we can easily encode the
  modified game in the SMT-LIB format, see
  Figure~\ref{fig:smt-formula}, which shows the relevant part (the rest
  is the definition of the functions \texttt{f}, \texttt{decrease} and
  the declaration of the constant \texttt{v}). Note that SMT-LIB uses
  a prefix notation. We define a predicate \texttt{win-game} which
  encodes the fact that $\exists$ win the modified game for a value \texttt{v}
  by simply spelling out the definition. Then we require that the game
  can be won for \texttt{v} and that \texttt{v} is the largest such
  value.

  \begin{figure}
    \begin{subfigure}[c]{0.28\textwidth}
      \centering
      \scalebox{0.9}{
      \begin{tikzpicture}[scale=3]
        \draw[->, semithick] (0,0) -- (1.1,0) node[right] {$x$};
        
        \foreach \x in {1,...,3} \draw (0.25*\x,0.05) -- +(0,-0.1)
        node[below] {\x/4};

        \draw (0,0.05) -- +(0,-0.1) node[below] {0}; \draw (1,0.05) --
        +(0,-0.1) node[below] {1};

        \draw[->, semithick] (0,0) -- (0,1.1) node[above] {$y$};
        \foreach \x in {1,...,3} \draw (-0.05,0.25*\x) -- +(0.1,0)
        node[left=0.2cm] {\x/4};

        \draw (0.05,0) -- +(-0.1,0) node[left] {0}; \draw (0.05,1) --
        +(-0.1,0) node[left] {1};

        \draw[domain=0:0.5, thick] plot[id=discon0, samples=50]
        function{0.25+0.5*x};

        \draw[domain=0.5:1, thick] plot[id=discon1, samples=50]
        function{0.375+0.5*x};

        \draw [fill=white, semithick] (0.5,0.5) circle (0.03);
        
        \filldraw (0.5,0.625) circle (0.03);

        \draw[domain=0:1, thick, dashed] plot[id=diagonal, samples=50]
        function{x};
      \end{tikzpicture}
    }
      \caption{}
      \label{fig:function-graph}
    \end{subfigure}
    \qquad\quad
    \begin{subfigure}[c]{0.62\textwidth}
      \centering
      {\tt\scriptsize
        \begin{mytab}
          ; Predicate encoding the game \\
          (define-fun win-game ((v Real)) Bool \\
          (forall ((b Real)) \> \> \> \> \> \> \> \> \> \> \> \>
          ; forall chooses b \\
          \> (=> (and (< 0.0 b) (< b v)) \> \> \> \> \> \> \> \> \> \>
          \> ; with 0 < b < v \\
          \> \> (exists ((l Real)) \> \> \> \> \> \> \> \> \> \>
          ; exists chooses l \\
          \> \> \> (and (<= 0.0 l) (<= l 1.0) \> \> \> \> \> \> \> \>
          \> ; with 0 <= l <=
          1 \\
          \> \> \> \> (>= (f l) b) \> \> \> \> \> \> \> \> ; with f(l) >=
          b \\
          \> \> \> \> (<= l b) (decrease v b l)))))) \> \> \> \> \> \>
          \> \> ; and we decrease \\
          \\
          ; Specify that we can win the game for v \\
          (assert (win-game v)) \\
          \\
          ; v is the greatest value for which one can win the game \\
          (assert (forall ((w Real)) \\
          \> (=> (and (<= 0.0 w) (<= w 1.0) (win-game w)) \\
          \> \> (<= w v)))) \\
          (check-sat) \\
          (get-model) 
        \end{mytab}
      }
      \caption{SMT formula encoding the modified game}
      \label{fig:smt-formula}
    \end{subfigure}
    \caption{}
  \end{figure}

  By running the SMT solver \texttt{cvc4}, we obtain
  $\frac{3}{4}$. Since we only showed correctness we can only
  guarantee that the value found is smaller or equal than the true
  solution. However, in this case $\frac{3}{4}$ is the true
  solution,
  and this is not by chance. In fact, we will discuss
    sufficient conditions that ensure completeness that cover also
    this specific example.

  We have also run successful experiments with the SMT solver
  \texttt{z3} involving equation systems and non-linear (quadratic)
  equations, where it is less obvious to compute fixpoints.
\end{example}

It is also possible to encode the solution of a fixpoint equation
systems into SMT solvers in a more direct way (see~\cite{MS:MS} for a
more detailed explanation). In the above example one would simply
search for a fixpoint (which can be determined by solving linear
equations) such that all other fixpoints are larger or equal. While
the direct encoding is reasonably straightforward, it has been shown
in~\cite{MS:MS} that due to the nesting of equations the encoding will
be of a size exponential in the number of equations. This can be also
the case in our setting (due to the growth of the trees depicted
above), however if every function $f_i$ depends only on few
parameters
(preferably the $i$-th parameter and one other),
then the game trees can be of linear size and we obtain also formulae
of linear size. To our knowledge, such an efficiency gain cannot be
achieved in the direct encoding.

\smallskip

We will now discuss the issue of completeness. We will first prove
that it holds when the lattice order is well-ordered,
hence well-founded and total. This is for instance the case for the
lattice of integers (enriched with a top element).

\begin{proposition}[completeness on well-orders]
  \label{prop:modified-game-complete}
  Let $(L,\sqsubseteq)$ be a lattice where $\sqsubset$ is
  a well-order and define the $\decr$-predicate
  $\decr(v,b,l) = (l\sqsubset b)$. Furthermore assume that the
  solution is reached in at most $\omega$ steps, i.e., all entries in
  $\ord{\vec{u}}$ for the solution $\vec{u}$ are at most $\omega$.
  Then the modified game is complete in the following sense: $\exists$
  has a winning strategy for $(u_i,i)$ for every $i\in\interval{m}$.
\end{proposition}

\begin{proof}
  Since the modified game starts with $v_i = u_i$ we can always assume
  that $v_i$ is the component of a $\mu$-approximant $\vec{v}$. This
  means that we can follow the winning strategy of $\exists$ in the
  original game, described in Lemmas~\ref{le:descend-mu}
  and~\ref{le:completeness} by descending along the
  $\mu$-approximants.

  We only make a slight modification in the choice of the new
  $\mu$-approximant whenever the current game index $j$ satisfies
  $\eta_j = \mu$. Assume that $v_j$ is fixed and that $v_j$ is a
  component of a $\mu$-approximant $\vec{v}$. Furthermore
  $b_j \ll v_j$. Take the least ordinal such that
  $b_j\sqsubseteq f_{j,\vec{v}}^\gamma(\bot)$. Since
  $v_j = f_{j,\vec{v}}^\alpha(\bot)$ for some $\alpha \le \omega$ it
  holds that $\gamma < \omega$. Hence either $\gamma = 0$ and
  $b_j = \bot$ (which cannot occur since $\bot\not\in B_L$) or
  $\gamma = \delta+1$ is a successor ordinal. We define
  $l_j = f_{j,\vec{v}}^\delta(\bot)$ and we can show with the same
  arguments as in Lemma~\ref{le:descend-mu} that one can define a new
  $\mu$-approximant $\vec{l} = (l_1,\dots,l_j,v_{j+1},\dots,v_m)$
  where
  $(l_1, \ldots, l_{j-1}) =
  \sol{\subst{\subst{E}{\subvec{x}{j+1}{m}}{\subvec{v}{j+1}{m}}}{x_j}{l_j}}$
  which satisfies $\vec{l}\in\Emoves{b_j,j}$ and
  $\ord{\vec{v}} \succ_j \ord{\vec{l}}$. Furthermore
  $l_j\sqsubset b_j$, due to the fact that $\gamma$ is minimal and the
  order is total.

  Now assume that we have reached a leaf in the game tree with index
  $k$, i.e., there is an earlier occurrence of $k$ with no larger
  index in between, and let $v_k,b_k,l_k\in L$ the lattice elements
  which are recorded in the game. Furthermore $v_k$ is a component of
  a $\mu$-approximant $\vec{v}$ and $l_k$ is a component of another
  $\mu$-approximant $\vec{l}$. Since all indices $j$ that occur
  between the two occurrences of $k$ satisfy $j < k$ and the
  subsequent $\mu$-approximants are ordered by $\succeq_j$, we have
  $\vec{v} \succeq_k \vec{l}$ (if $\eta_k = \mu$ the equality is
  strict).

  In particular the construction of new $\mu$-approximants is such
  that all components with $i > k$ are unchanged: $v_i = l_i$. If
  $\eta_k=\nu$ also $v_k = l_k$ and hence $l_k\sqsubseteq v_k$.  If
  $\eta_k=\mu$ we only modify the $k$-the component of the
  $\mu$-approximant in the first step and it is left unchanged
  afterwards.  The adapted construction of the $\mu$-approximant
  explained above ensures $l_k\sqsubset b_k$.  Hence $\exists$ can
  make sure that the $\decr$-predicate is satisfied and wins the
  modified game.
\end{proof}

Note that it does not hold that $\exists$ has a winning strategy for
$(v_i,i)$ for every $v_i\sqsubseteq u_i$. The requirement that $v_i$ is a component of a $\mu$-approximant is important. Consider
for instance the three-element lattice $L = \{\bot,a,\top\}$ and a
monotone function $f\colon L\to L$ with $f(\bot)=f(a)=\bot$,
$f(\top)=\top$. This function is monotone and has $\mu f=\bot$,
$\nu f=\top$. We consider the equation $x =_\nu f(x)$ and start
playing with $v_1 = a$, which is below the greatest fixpoint. If
$\forall$ chooses $b_1=a\ll a$, there does not exist a value $l_1$
with $f(l_1) \sqsupseteq b_1$ and $l_1\sqsubseteq b_1$. This is
connected to the fact that every post-fixpoint is below the greatest
fixpoint, but not every lattice element below the greatest fixpoint is
a post-fixpoint.

It would of course be desirable to extend
Proposition~\ref{prop:modified-game-complete} in such a way that it
also covers fixpoint iteration beyond $\omega$ and non-well-founded
lattice orders. Going beyond $\omega$ for well-founded orders requires
the introduction of ``special'' lattice elements $a_i$ as in the
$\decr$-predicate in order to cover the discontinuity
points.
However, this could be more complex for arbitrary lattices than for
the reals, since there we know that $b\sqsubseteq v$ ($b \le v$) and
$b\not\ll v$ ($b\not< v$) imply $b=v$.

On the other hand, handling non-well-founded lattices orders would
require an adaptation where, given $b_j\ll v_j$ and $v_j$ is below the
least fixpoint, we are always able to choose $\vec{l}$ such that
$l_j \sqsubseteq b_j$ and $\decr(b_j,v_j,l_j)$. We will characterise a class of
functions on the real interval $[0,1]$ for which this is the case and
which also explains the shape of the $\decr$-condition of
Lemma~\ref{lem:decr-wellfounded-reals}.

\begin{lemma}[completeness for piecewise linear dominated functions]
  \label{lem:decrease-above-piecewise-linear}
  Fix real numbers $a_0=0<a_1<\dots<a_\ell < a_{\ell+1}\le 1$,
  $0<p_i<1$, $0\le q_i\le 1$, $i\in\interval{\ell+1}$ be given. Let
  $g_i(x)=p_ix+q_i$ and assume that $g_i(a_i)=a_i$ for all
  $i\in\interval{\ell+1}$. Let $g\colon [0,1]\to [0,1]$ be a function
  with
  \[ g(x) = \left\{
      \begin{array}{ll}
        g_i(x) & \mbox{if $a_{i-1}\le x< a_i$} \\
        g_{\ell+1}(x) & \mbox{if $a_\ell\le x$} 
      \end{array}\right. . \]
  Define the $\decr$ predicate as in
  Lemma~\ref{lem:decr-wellfounded-reals} with
  $c = \min_i \frac{1-p_i}{p_i}$ and with the values
  $a_0,\dots,a_\ell$.

  Then the modified game for the  single equation
  $x =_\mu g(x)$ is complete, i.e., $\exists$ has a winning strategy
  for all $v\le \mu f$.
  The same holds for any monotone function $f\colon [0,1]\to[0,1]$
  with $f\ge g$ and $\mu f = a_{\ell+1}$.
\end{lemma}

\begin{proof}
  In this case the game is over after one iteration and we only have
  to show that $\exists$ can always find an answering move that
  ensures her win. In particular, we have to prove that for all
  $v,b\in[0,1]$ with $b < v \le \mu f$ there exists $l\in [0,1]$ such
  that $0\le l\le b$, $b\le f(l)$ and $\decr(v,b,l)$.
  
  First note that
  $g_{\ell+1}(\mu f) = g_{\ell+1}(a_{\ell+1}) = a_{\ell+1} = \mu f =
  f(\mu f)$ means that $g,f$ agree on the least fixpoint $\mu
  f$.
  
  Now assume that $a_{i-1} \le b < a_i$, where
  $i\in\interval{\ell+1}$. We have $f(b) \ge g_i(b)$. Define
  $l = \frac{b-q_i}{p_i}$ if $\frac{b-q_i}{p_i}\ge a_{i-1}$ otherwise
  $l = a_{i-1}$. We consider both cases:
  \begin{itemize}
  \item $l = \frac{b-q_i}{p_i}$: 
    \begin{itemize}
    \item It obviously holds that
      $f(l) \ge g_i(l) = p_i\cdot \frac{b-q_i}{p_i} + q_i = b$, since
      $a_{i-1}\le l\le a_i$ ($l\le b < a_i$ is shown below).
    \item Now we show that
      $b-l = \frac{1-p_i}{p_i}\cdot (a_i-b) \ge c\cdot (a_i-b)$.

      Since $g(a_i) = a_i$, we have $p_ia_i + q_i = a_i$, which
      implies $a_i = \frac{a_i-q_i}{p_i}$. Hence we get
      \begin{eqnarray*}
        b & = & a_i + (b-a_i) = \frac{a_i-q_i}{p_i} + (b-a_i) = 
        \frac{a_i-q_i}{p_i} + \frac{b-a_i}{p_i} - \frac{b-a_i}{p_i} +
        (b-a_i) \\
        & = & \frac{b-q_i}{p_i} + \frac{-(b-a_i)+p_i(b-a_i)}{p_i}
        = l+\frac{1-p_i}{p_i}\cdot (a_i-b)
      \end{eqnarray*}
      This immediately implies
      $b-l = \frac{1-p_i}{p_i}\cdot (a_i-b) \ge c\cdot (a_i-b)$.

      If $a_\ell < b < a_{\ell+1} = \mu f$ we can infer 
      $b-l\ge c\cdot (a_{\ell+1}-b) \ge c\cdot (v-b)$.
    \item Note that $b-l\ge c\cdot (a_i-b)$ implies $l\le b$ (since
      $b\le a_i$) and hence $0\le l\le b$.
    \end{itemize}
  \item $l = a_{i-1}$:
    \begin{itemize}
    \item It holds that $\frac{b-q_i}{p_i} < 0$ and hence $b <
      q_i$. This means that $f(l) = f(a_{i-1}) \ge g_i(a_{i-1}) =
      p_ia_{i-1} + q_i \ge q_i > b$, hence $f(l) \ge b$.
    \item The $\decr$-predicate is automatically satisfied since $l =
      a_{i-1}$. 
    \end{itemize}
  \end{itemize}
\end{proof}

Lemma~\ref{lem:decrease-above-piecewise-linear} states the following
condition: the function $f$ must be larger or equal than a piecewise
linear function where the pieces always end on the diagonal and the
least fixpoints of $f,g$ are equal. (See Figure~\ref{fig:fct-g-f}
where $g$ is drawn with a solid and $f$ with a dotted line.) The slope
of these piecewise linear functions can be arbitrarily close to $1$,
that is close to the diagonal. Since monotone functions on $[0,1]$ are
always above the diagonal (in a post-fixpoint) before they reach the
diagonal (the fixpoint), this is not a very strong restriction.

\begin{wrapfigure}{r}{0.3\textwidth}
  \centering
  \scalebox{0.8}{
  \begin{tikzpicture}[scale=3]
    \draw[->, semithick] (0,0) -- (1.1,0) node[right] {$x$};
    
    \draw (0,0.05) -- +(0,-0.1) node[below] {$a_0=0$};

    \draw (0.25,0.05) -- +(0,-0.1) node[below] {$a_1$};

    \draw (0.5,0.05) -- +(0,-0.1) node[below] {$a_2$};

    \draw (0.75,0.05) -- +(0,-0.1) node[below] {$a_3$};

    \draw (1,0.05) -- +(0,-0.1) node[below] {$a_4=1$};

    \draw[->, semithick] (0,0) -- (0,1.1) node[above] {$y$};
    
    \draw (0.05,0) -- +(-0.1,0) node[left] {0};

    \draw (0.05,1) -- +(-0.1,0) node[left] {1};

    \draw[domain=0:0.25, thick] plot[id=gfct0, samples=50]
    function{0.1875+0.25*x};

    \draw[domain=0.25:0.5, thick] plot[id=gfct1, samples=50]
    function{0.25+0.5*x};

    \draw[domain=0.5:1, thick] plot[id=gfct2, samples=50]
    function{0.25+0.6666*x};

    \draw [fill=white, semithick] (0.25,0.25) circle (0.03);
    
    \draw [fill=white, semithick] (0.5,0.5) circle (0.03);
    
    \filldraw (0.25,0.375) circle (0.03);

    \filldraw (0.5,0.5833) circle (0.03);

    \draw[domain=0:1, thick, dashed] plot[id=diagonal, samples=50]
    function{x};

    \draw [thick, dotted](0,0.22) .. controls (0.4,0.55) and
    (0.5,0.72) .. (0.75,0.75);

    \draw [thick, dotted](0.75,0.75) -- (1,0.77);
  \end{tikzpicture}
  }
  \caption{}
  \label{fig:fct-g-f}
\end{wrapfigure}
For instance, the function $f$ in
Example~\ref{ex:smt-solver-discontinuous} satisfies the requirements
of Lemma~\ref{lem:decrease-above-piecewise-linear} since it is itself
a piecewise linear function with slopes $p_1 = p_2 =
\frac{1}{2}$. This is why in the example the constant $c =
\frac{1-p_1}{p_1} = 1$ is sufficient.

Considering the case of multiple equations, note that if we could
guarantee that every function $f_{i,\vec{l}}$, which we use to
determine $\mu$-approximants (see Definition~\ref{de:approximants}) is
of this kind, then we could generalise the proof. In~\cite{MS:MS} it
is shown that the functions that arise from evaluating Lukasiewicz
$\mu$-terms are piecewise linear, hence they would in principle fit
our characterisation. On the other hand~\cite{MS:MS} also gives
examples (based on the strong Lukasiewicz operators) that solutions of
fixpoint expressions of continuous function need not be continuous
themselves. Hence, the functions $f_{i,\vec{l}}$ can be discontinuous
even if all $f_i$ are continuous. It is an open question to find the
discontinuity points $a_i$ that are required to define the
$\decr$-predicates and it is unclear whether there are always only
finitely many of them.  We believe that this task might be easier for
non-expansive functions, but we leave this as future work.

\section{Conclusion}
\label{sec:conclusions}

\paragraph{Related work} Our work is based on lattice theory and in
particular on continuous lattices. The use of lattices in program
analysis and verification has been pioneered by the
work~\cite{cc:ai-unified-lattice-model}. Continuous lattices, which
received this name due to their intimate connection with continuous
functions, have originally been studied by Scott as a semantic domain
for the $\lambda$-calculus~\cite{Scott:CL} and have since found many
further applications in the semantics of programming
languages~\cite{ghklms:continuous-lattices-domains,AJ:DT}.

The modal $\mu$-calculus is an expressive temporal logics, which
originated in an unpublished manuscript by Scott and de Bakker and was
further developed by Kozen~\cite{k:prop-mu-calculus}. For a good
overview see~\cite{bw:mu-calculus-modcheck}.

Its introduction posed the problem of efficient model-checking, which
involves the solution of nested fixpoint equations, see,
e.g.,~\cite{bcjlm:evaluation-fixpoint,s:fast-simple-nested-fixpoints,cks:faster-modcheck-mu}. The paper~\cite{cks:faster-modcheck-mu}
introduced the notion of a hierarchical system of fixpoint equations,
on which our paper is based as well. One way to tackle the
model-checking problem is to translate it into the question of finding
winning strategies for parity games. The latter were first described
in~\cite{ej:tree-automata-mu-determinacy}.

A very satisfying technique for solving parity games was proposed
in~\cite{j:progress-measures-parity} resulting in an algorithm which
is exponential only in half of the alternation depth. The approach
crucially relies on the notion of progress measure, that can be seen
as generalising both invariants and ranking functions. The complexity
of computing progress measures has recently been improved to
quasi-polynomial~\cite{CJKLS:DPGQPT}.

An extension to general lattices has been given
in~\cite{hsc:lattice-progress-measures}, which was very inspiring for
our development. Compared to~\cite{hsc:lattice-progress-measures} we
brought games back into the picture by introducing a game that
generalises both parity games and the unfolding games
in~\cite{v:lectures-mu-calculus}. This allowed us to define a notion
of progress measures which is closer to the original definition of
Jurdzi\'nski and, as such, admits a constructive characterisation as a
least fixed point. This works in the general context of continuous
lattices, providing a way of solving systems of fixpoint equations in
settings that are beyond powerset lattices and were not covered by
previous work. A related game for complete-prime algebraic lattices
was introduced in the appendix of
\cite{ktw:higher-order-hfl-modcheck-arxiv}, the full version of
\cite{ktw:higher-order-hfl-modcheck}.

We devised the notion of selection and a logics for specifying the moves
of the existential player, with the aim of making the computation of progress measures more efficient.
We view as a valuable contribution the identification of
continuous lattices as the right setting where these general
results can be stated.

Usually, $\mu$-calculus formulae are evaluated over the state space of
a transition system, i.e., over a powerset lattice. This changes if
the $\mu$-calculus is not a classical logic, but lattice-valued as in~\cite{kl:latticed-simulation} or real-valued as in~\cite{hk:quantitative-analysis-mc}, which presents an algorithm based
on the simplex method for the non-nested case. Solving equation
systems over the reals was considered in~\cite{gs:rational-equations-strategy-iteration} and in~\cite{MS:MS,ms:lukasiewicz-mu-arxiv}. In particular,~\cite{MS:MS} presents an
algorithm for solving nested fixpoint equation systems over the
interval $[0,1]$ by a direct algorithm which represents and
manipulates piecewise linear functions as conditioned linear
expressions. As far as we know this algorithm has not been
implemented. Our results can offer
an alternative way to solve such equation systems.

Games for quantitative or probabilistic $\mu$-calculi have been
studied in~\cite{mm:quantitative-mu,m:game-semantics-prob-mu}. As
opposed to our game, such games closely follow the structure of the
$\mu$-calculus formula on which the game is based (e.g., $\exists$
makes a choice at an $\lor$-node, $\forall$ at an $\land$-node). It is
an interesting question whether the conceptual simplicity of our game
can lead to a new perspective on existing games.

\smallskip

\paragraph{Future work} A parity game over a finite graph can be easily
converted into a system of boolean equations whose solution
characterises the winning positions for the players. Since our game is
a standard parity game, possibly played on an infinite graph, the
standard conversion would lead to infinitely many equations.  Systems
of equations of this kind are considered, e.g.,
in~\cite{m:modal-boolean-equation}.  An interesting question, still to
be investigated, is under which conditions an infinite parity game can
be converted into finitely many equations on an (infinite) powerset
lattice.

\smallskip

The generality of continuous lattices suggests the possibility of instantiating our framework in various other application
scenarios.

The use of our results for solving fixpoint equations over the
  reals via SMT solvers in \S~\ref{sec:solving-fp-equations-smt}
  appears to be promising, but it requires further investigation. In
  particular, we plan to deepen the issue of completeness for which we
  currently only have partial results (see
  \S~\ref{sec:solving-fp-equations-smt}).
 
\smallskip

We also plan to study fixpoint equations on the (non-distributive, but
continuous) lattices of equivalence relations and pseudo-metrics. As
explained in the introduction, the computation of fixpoints for
equivalence relations is essential for behavioural equivalences, and
the same holds for pseudo-metrics and behavioural distances~\cite{bw:behavioural-pseudometric}.

We would also like to determine whether we can handle quantitative
logics whose modalities interact with (lattice) truth values in a
non-trivial way, such as logics with discounted modalities as studied
in~\cite{abk:discouting-ltl}. Expressing such logics as systems of fixed point equations over suitable continuous lattices and thus obtaining a game theoretical characterisation of the model checking problem seems reasonably easy. However, turning such characterisation into an effective technique requires some non-trivial symbolic approach due to the fact that the lattice is
infinite.

\smallskip

Furthermore we would like to study situations in which local (or
on-the-fly) algorithms rather than global fixpoint iteration can be
used to check whether a lattice element is below the
solution. Examples of such local algorithms are backtracking methods
studied in \cite{ss:practical-modcheck-games,h:proving-up-to}. In
particular we are interested in the integration of local methods with
up-to techniques for general lattices, see for instance
\cite{bggp:sound-up-to-complete-abstract,ps:enhancements-coinductive,p:complete-lattices-up-to}.

\begin{acks}
  Research is partially supported by DFG project BEMEGA and University
  of Padova project ASTA.
  
  We want to thank the reviewers for their insightful and valuable
  comments. Furthermore we are grateful to Clemens Kupke for bringing
  fixpoint games to our attention.
\end{acks}

\bibliography{references}

\appendix

\section{Comparing Fixpoint Equation Systems with $\mu$-Calculus
  Formulae}
\label{sec:fp-systems-mu}

We will show how $\mu$-calculus formulae can be translated into
equation systems and vice versa.

Hereafter we will assume that in every formula different bound
variables have different names, a requirement that can always be
fulfilled by alpha-renaming. In this way for every variable $x$
appearing in a closed formula $\phi$, we can refer to ``the'' fixpoint
subformula quantifying $x$, that will be denoted $\phi_x$ (hence
$\phi_x$ is of the kind $\eta x. \psi$).

\begin{definition}[equation system for a formula]
  \label{de:eqnsys-for-muformula}
  Given a closed fixpoint formula $\phi$ of the $\mu$-calculus, let
  $(x_1,\ldots,x_m)$ be the tuple of variables in
  $\phi$, in the order in which their quantification appears from right to left.
  The \emph{equational form} of $\phi$ is
  $\vec{x} =_{\vec{\eta}} \vec{\theta}$, where, for all
  $i \in \interval{m}$, if $\varphi_{x_i} = \eta_i x_i. \psi_i$ then
  $\theta_i$ is the (open) formula obtained from $\psi_i$ by replacing
  every fixpoint subformula with the corresponding propositional
  variable.
\end{definition}

Observe that the restriction to fixpoint formulae is not limiting since any formula $\varphi$ is equivalent to a fixpoint formula $\mu x.\varphi$, where $x$ is a variable not occurring in $\varphi$.

Once a transition system $(\mathbb{S},\to)$ is fixed, the formula in equational form can be interpreted as a system of equations over the powerset lattice $(\Pow{\mathbb{S}}, \subseteq)$, by replacing formulae with their semantics, i.e., an equation $x_i =_{\eta_i} \theta_i$ becomes
\begin{eqnarray*}
  x_i & =_{\eta_i} & \sem{\theta_i}_\rho
\end{eqnarray*}
where in the right-hand side $\sem{x_j}_\rho$ is replaced by $x_j$ and
$\rho$ is some fixed environments providing a meaning only for
propositions.

It is not difficult to see that also a converse transformation is
possible, i.e., a system of fixpoint equations of the kind
$\vec{x} =_{\vec{\eta}} \vec{\psi}$ where each $\psi_i$ is an open
formula with propositional variables in $\vec{x}$ and without
fixpoints, can be translated into a tuple of $\mu$-calculus formulae,
equivalent to the system in a sense formalised later.

\begin{definition}[formulae for an equation system]
  \label{de:muformulae-for-eqnsys}
  Let $E$ be a system of $m$ equations
  $\vec{x} =_{\vec{\eta}} \vec{\psi}$ where each $\psi_i$ is an open
  formula with variables in $\vec{x}$ and without fixpoints.  The
  corresponding $m$-tuple of $\mu$-calculus formulae, denoted by
  $\vec{\phi}^E$, is defined inductively as follows, where $E_i$
  denotes the system consisting of the first $i$ equations of $E$, for
  all $i \in \interval{m}$.
  \begin{center}
    $
    \begin{array}{lll}
      \vec{\phi}^{\emptyset}
      & =
      & () \\
      \vec{\phi}^{E_i} & = &
      (\propsubst{\vec{\phi}^{E_{i-1}}}{x_i}{\phi^{E_i}_{i}},\phi^{E_i}_{i}) 
      \qquad
      \mbox{where
        $\phi^{E_i}_{i} = \eta_i
        x_i. \propsubst{\psi_i}{x_j}{\phi^{E_{i-1}}_{j}}_{\forall
          j\in\interval{i-1}}$}
    \end{array}
    $
  \end{center}
  Then
  $\vec{\phi}^E = \vec{\phi}^{E_m}$.
\end{definition}

Note that $\vec{\phi}^E$ is a tuple of closed formulae.

A similar procedure is given in~\cite{cks:faster-modcheck-mu} for the
characterisation of $\mu$-calculus formulae in terms of equation
systems, to allow an efficient model-checking algorithm. However such
equation systems differ from ours. In particular, the solution of a
system is defined just by means of the semantics of the formulae into
which the system can be translated.

We finally prove that the proposed translations preserve the
semantics. %
We will need the substitution lemma for the $\mu$-calculus as stated below.

\begin{lemma}[substitution in $\mu$-calculus]
  \label{le:subst-mu}
  For all $\mu$-calculus formulae $\phi$ and $\psi$, variable $x$, and
  environment $\rho$, it holds
  $\sem{\propsubst{\phi}{x}{\psi}}_\rho =
  \sem{\phi}_{\rho[x\mapsto\sem{\psi}_\rho]}$.
\end{lemma}
\begin{proof}
It can be easily proved by routine induction on $\phi$.
\end{proof}

\begin{proposition}[correspondence between formulae and equation systems]
  Let $\phi$ be a closed fixpoint formula of the $\mu$-calculus and
  let $E$ be the system arising as its equational form. For any 
  environment $\rho$, it holds
  $\sem{\phi}_\rho = \sol[m]{E}$, where $m$ is number of equations in $E$.
  Conversely, given a system $E$ of $m$ equations
  $\vec{x} =_{\vec{\eta}} \vec{\psi}$, for all $i \in \interval{m}$,
  it holds that $\sol[i]{E} = \sem{\phi^{E}_i}_\rho$, where $\rho$ is any
  environment.
\end{proposition}

\begin{proof}
  We prove the two statements separately.

  For the first part, let $\phi$ be a closed fixpoint formula of the
  $\mu$-calculus and let $E$ the system arising as its equational
  form. Recall that for every $i \in \interval{m}$, the $i$-th equation
  of the system is $x_i =_{\eta_i} \theta_i$, where $\theta_i$ is
  obtained from the subformula $\psi_i$ of the fixpoint formula
  $\phi_{x_i} = \eta_i x_i.\psi_i$, as described in
  Definition~\ref{de:eqnsys-for-muformula}. In particular, 
  $\varphi = \phi_{x_m} = \eta_m x_m.\psi_m$, corresponds
  to the last equation
  of the system $x_m =_{\eta_m} \theta_m$. Then, we prove that for all
  $i \in \interval{m}$, $\sem{\phi_{x_i}}_{\rho'} = \sol[i]{E'}$ where
  $\rho' = \rho[x_{i+1}\mapsto S_{i+1},\ldots,x_{m}\mapsto S_{m}]$ and
  $E' = \subst{\subst{E}{x_m}{S_m}\ldots}{x_{i+1}}{S_{i+1}}$, and
  every $S_j \subseteq \mathbb{S}$. Clearly this implies the desired
  result. The proof proceeds by induction on the index $i$.

  ($i = 1$) By definition of substitution we know that the system $E'$
  consists of a single equation, i.e., $x_1 =_{\eta_1} \theta_1'$
  where
  $\theta_1' =
  \propsubst{\propsubst{\theta_1}{x_m}{S_m}\ldots}{x_{2}}{S_{2}}$. By
  definition of solution of a system, we have that
  $\sol[1]{E'} = \sol{x_1 =_{\eta_1} \theta_1'} = \eta_1 (\lambda
  S.\sem{\theta_1}_{\rho'[x_1\mapsto S]})$. Similarly, by definition
  of the semantics of $\mu$-calculus we know that
  $\sem{\phi_{x_1}}_{\rho'} = \sem{\eta_1 x_1.\psi_1}_{\rho'} = \eta_1
  (\lambda S.\sem{\psi_1}_{\rho'[x_1\mapsto S]})$. By the definition
  of the ordering of the variables given in
  Definition~\ref{de:eqnsys-for-muformula}, we must have that $\psi_1$
  does not contain any fixpoint subformula, otherwise its index could
  not be $1$. Hence $\theta_1 = \psi_1$, and so
  $\sem{\phi_{x_1}}_{\rho'} = \sol[1]{E'}$.

  ($i > 1$) In this case, by definition of solution, we have
  $\sol[i]{E'} = \eta_i (\lambda S.\sem{\theta_i}_{\rho''})$ where
  $\rho'' = \rho'[x_i\mapsto
  S][\subvec{x}{1}{i-1}\mapsto\sol{\subst{E'}{x_i}{S}}]$. While, by
  definition of the semantics we have
  $\sem{\phi_{x_i}}_{\rho'} = \sem{\eta_i x_i.\psi_i}_{\rho'} = \eta_i
  (\lambda S.\sem{\psi_i}_{\rho'[x_i\mapsto S]})$. By an inspection of
  the definition of $\vec{\theta}$ and the ordering of the variables,
  one can notice that for all $j \in \interval{m}$,
  $\psi_j = \propsubst
  {\propsubst{\theta_j}{x_{j-1}}{\phi_{x_{j-1}}}\ldots}{x_1}
  {\phi_{x_1}}$. Then, we have that
  $\sem{\psi_i}_{\rho'[x_i\mapsto S]} =
  \sem{\propsubst{\propsubst{\theta_i}{x_{i-1}}{\phi_{x_{i-1}}}\ldots}{x_1}{\phi_{x_1}}}_{\rho'[x_i\mapsto
    S]}$. Furthermore, by repeatedly applying Lemma~\ref{le:subst-mu}
  we obtain that
  $\sem{\propsubst{\propsubst{\theta_i}{x_{i-1}}{\phi_{x_{i-1}}}\ldots}{x_1}{\phi_{x_1}}}_{\rho'[x_i\mapsto
    S]} = \sem{\theta_i}_{\rho_i}$ where
  $\rho_1 = \rho'[x_i\mapsto S]$ and for all $j \in \interval{i-1}$,
  $\rho_{j+1} =
  \rho_{j}[x_{i-j}\mapsto\sem{\phi_{x_{i-j}}}_{\rho_{j}}]$. Note that
  actually
  $\rho_i = \rho'[x_i\mapsto
  S][x_{i-1}\mapsto\sem{\phi_{x_{i-1}}}_{\rho_{1}}]\ldots[x_{1}\mapsto\sem{\phi_{x_{1}}}_{\rho_{i-1}}]$. Now
  we just need to prove that $\rho_i = \rho''$. To show this we can
  use the inductive hypothesis $i-1$ times, recalling the recursive
  structure of the solutions of systems. Therefore, we can conclude
  that
  $\sem{\psi_i}_{\rho'[x_i\mapsto S]} = \sem{\theta_i}_{\rho_i} =
  \sem{\theta_i}_{\rho''}$, and so
  $\sem{\phi_{x_i}}_{\rho'} = \sol[i]{E'}$.

  \medskip

  Let us now focus on the second part. Let $E$ be a system of $m$
  equations of the kind $\vec{x} =_{\vec{\eta}} \vec{\psi}$. We have
  to prove that, $\sol[i]{E} = \sem{\phi^E_i}_\rho$ for all
  $i \in \interval{m}$. The proof proceeds by induction on the number
  of equations $m$.

  ($m = 1$) Clearly there is only one possible index
  $i \in \interval{1}$, that is $i = 1$. Then, by definition of
  solution we know that
  $\sol[1]{E} = \sol{x_1 =_{\eta_1} \psi_1} = \eta_1 (\lambda
  S.\sem{\psi_1}_{\rho[x_1\mapsto S]})$. Moreover, by definition of
  the semantics of $\mu$-calculus we have that also
  $\sem{\phi^E_1}_\rho = \sem{\eta_1 x_1.\psi_1}_\rho = \eta_1
  (\lambda S.\sem{\psi_1}_{\rho[x_1\mapsto S]})$. Thus we can
  immediately conclude.

  ($m > 1$) First we show the property for $i = m$. By definition of
  solution we have
  $\sol[m]{E} = \eta_m (\lambda S.\sem{\psi_m}_{\rho'})$ where
  $\rho' = \rho[x_m\mapsto S][\subvec{x}{1}{m-1}\mapsto\sol{E'}]$ and
  $E' = \subst{E}{x_m}{S}$. By
  Definition~\ref{de:muformulae-for-eqnsys} we know that
  $\phi^E_m = \eta_m
  x_m.\propsubst{\psi_m}{x_j}{\phi^{E_{m-1}}_j}_{\forall j \in
    \interval{m-1}}$. Then, the semantics is \linebreak
  $\sem{\phi^E_m}_{\rho} = \eta_m (\lambda
  S.\sem{\propsubst{\psi_m}{x_j}{\phi^{E_{m-1}}_j}_{\forall j \in
      \interval{m-1}}}_{\rho[x_m\mapsto S]})$. By repeatedly applying
  Lemma~\ref{le:subst-mu} we obtain that 
  $\sem{\propsubst{\psi_m}{x_j}{\phi^{E_{m-1}}_j}_{\forall j \in
      \interval{m-1}}}_{\rho[x_m\mapsto S]} = \sem{\psi_m}_{\rho''}$
  where
  $\rho'' = \rho[x_m\mapsto
  S][x_j\mapsto\sem{\phi^{E_{m-1}}_j}_{\rho[x_m\mapsto S]}]_{\forall j
    \in \interval{m-1}}$. Now we just need to show that
  $\rho'' = \rho'$, that is, for all $j \in \interval{m-1}$,
  $\sem{\phi^{E_{m-1}}_j}_{\rho[x_m\mapsto S]} = \sol[j]{E'}$. Since
  $E'$ has $m-1$ equations, by inductive hypothesis we know that
  $\sol[j]{E'} = \sem{\phi^{E'}_j}_{\rho}$. Recalling that
  $E' = \subst{E}{x_m}{S}$, we can immediately conclude that
  $\sem{\phi^{E'}_j}_{\rho} = \sem{\phi^{E_{m-1}}_j}_{\rho[x_m\mapsto
    S]}$.

  Instead, for all $i \in \interval{m-1}$, by definition of solution
  and what we just proved above we have
  $\sol[i]{E} = \sol[i]{\subst{E}{x_m}{\sol[m]{E}}} =
  \sol[i]{\subst{E}{x_m}{\sem{\phi^E_m}_\rho}}$. Moreover, let
  $E' = \subst{E}{x_m}{\sem{\phi^E_m}_\rho}$, since $E'$ has $m-1$
  equations, by inductive hypothesis we know that
  $\sol[i]{E'} = \sem{\phi^{E'}_i}_\rho$. Observe that
  $\sem{\phi^{E'}_i}_\rho = \sem{\phi^{E_{m-1}}_i}_{\rho[x_m\mapsto
    \sem{\phi^E_m}_\rho]} =
  \sem{\propsubst{\phi^{E_{m-1}}_i}{x_m}{\phi^E_m}}_\rho$ by
  Lemma~\ref{le:subst-mu} and since
  $E' = \subst{E}{x_m}{\sem{\phi^E_m}_\rho}$. Then, since by
  Definition~\ref{de:muformulae-for-eqnsys} we know that
  $\phi^E_i = \propsubst{\phi^{E_{m-1}}_i}{x_m}{\phi^E_m}$, we can
  conclude that
  $\sem{\phi^E_i}_\rho = \sem{\phi^{E'}_i}_\rho = \sol[i]{E'} =
  \sol[i]{E}$.
\end{proof}

\section{Results Concerning the Comparison with~\cite{hsc:lattice-progress-measures}}
\label{sec:hasuo}

\subsection{Comparing the Definitions of Solutions}
\label{ssec:comparison-solution-hsc}

Here we show that the definition of the solution of an equational
system in~\cite{hsc:lattice-progress-measures} is equivalent to our
Definition~\ref{de:solution}. In both definitions the solution
$\vec{u}=(u_1,\ldots,u_m)$ is solved recursively based on interim
solutions by calculating fixpoints.

\begin{definition}[Solution of an equational system~\cite{hsc:lattice-progress-measures}]
  \label{solution-Ichiro}
  Let $L$ be a lattice and let $E$ be a system of $m \ge 1$ equations
  on $L$ of the kind $\vec{x} =_{\vec{\eta}} \vec{f}(\vec{x})$. For
  each $i \in \interval{m}$ and $j \in \interval{i}$ we define monotone functions
  $f^{\ddag} \colon L^{m-i+1} \to L$ and $l_j^{(i)}\colon L^{m-i}\to L$ as follows,
  inductively on $i$:
  \begin{enumerate}
  \item  $i=1:$ 
    \begin{align*}
      &f^{\ddag}_1(l_1,\ldots,l_m):= f_1(l_1,\ldots,l_m)&\\
      &l_1^{(1)}(l_2,\ldots,l_m):= \eta_1[f^{\ddag}_1(\_,l_2,\ldots,l_m)\colon L \to L]&
    \end{align*}
    with $\eta_1 \in \{\mu,\nu \}$.
  \item  $i=i+1:$
    \begin{align*}
      &f^{\ddag}_{i+1}(l_{i+1},\ldots,l_m):= f_{i+1}(l_{1}^{(i)}(l_{i+1},\ldots,l_m),\ldots,l_{i}^{(i)}(l_{i+1},\ldots,l_m),l_{i+1},l_m)&\\
      &l_{i+1}^{(i+1)}(l_{i+2},\ldots,l_m):= \eta_{i+1}[f^{\ddag}_{i+1}(\_,l_{i+2},\ldots,l_m\colon L \to L]&
    \end{align*}
    with $\eta_{i+1} \in \{\mu,\nu \}$. The $l_{i+1}^{(i+1)}$ solution
    is then used to obtain the $(i+1)$-th interim solutions for each
    $j \in \interval{i}$:
    \begin{displaymath}
      l_j^{(i+1)}(l_{i+2},\ldots,l_m) := l_j^{(i)}(l_{i+1}^{(i+1)}(l_{i+2},\ldots,l_m),l_{i+2},\ldots,l_m)
    \end{displaymath}
  \end{enumerate}
\end{definition}

\begin{proposition}
  Let $L$ be a lattice and let $E$ be a system of $m \ge 1$ equations
  on $L$ of the kind $\vec{x} =_{\vec{\eta}} \vec{f}(\vec{x})$. Then
  the solution from Definition~\ref{solution-Ichiro} coincides with
  the solution from Definition~\ref{de:solution}.
\end{proposition}

\begin{proof}
  Let $l_{i+1},\ldots,l_m\in L$ be given. We show
  $l_j^{(i)}(l_{i+1},\ldots,l_m) =
  \sol[j]{\subst{E}{\subvec{x}{i+1}{m}}{\subvec{l}{i+1}{m}}}$ for
  $j \in \interval{i}$ by induction on $i$:
  \begin{enumerate}
  \item $i=1$: We define
    $E'=\subst{E}{\subvec{x}{2}{m}}{\subvec{l}{2}{m}}$ and according
    to Definition~\ref{de:solution} we have
    \[ \sol{\subst{E}{\subvec{x}{2}{m}}{\subvec{l}{2}{m}}}=\sol{E'}=
      \sol{\subst{E'}{x_1}{u_1}},u_1) = (u_1) \] where
    $u_1=\eta_1 (\lambda x.\, f_1(x))$.  In
    Definition~\ref{solution-Ichiro} for $i=1$ we only have to
    consider
    $l_1^{(1)}(l_2,\ldots,l_m) =
    \eta_1[f^{\ddag}_1(\_,l_2,\cdots,l_m)]=
    \eta_1[f_1(\_,l_2,\ldots,l_m)]$ which corresponds to
    $\sol[1]{\subst{E}{\subvec{x}{2}{m}}{\subvec{l}{2}{m}}} = u_1 =
    \eta_1 (\lambda x.\, f_1(x)) $.
  \item $i \to i+1:$ We define
    $E' = \subst{E}{\subvec{x}{i+2}{m}}{\subvec{l}{i+2}{m}}$. Here we
    need to distinguish two cases to prove that
    $l_j^{(i+1)}(l_{i+2},\ldots,l_m) = \sol[j]{E'}$ for all
    $j\in\interval{i}$.
    \begin{enumerate}
      \label{hsc2a}
    \item $j=i+1$: From Definition~\ref{de:solution} we have
      \[ \sol{\subst{E}{\subvec{x}{i+2}{m}}{\subvec{l}{i+2}{m}}} =
        \sol{E'}= \sol{\subst{E'}{x_{i+1}}{u_{i+1}}},u_{i+1}) \] where
      $u_{i+1}=\eta_{i+1} (\lambda x.\,
      f_{i+1}(\sol{\subst{E'}{x_{i+1}}{x}},x,l_{i+1},\dots,l_m))$. Hence
      \[ \sol[i+1]{\subst{E}{\subvec{x}{i+2}{m}}{\subvec{l}{i+2}{m}}} =
      u_{i+1}. \]

      From Definition~\ref{solution-Ichiro} we obtain
      $l_{i+1}^{(i+1)}(l_{i+2},\ldots,l_m) = \eta_{i+1}[
      f^{\ddag}_{i+1}(x,l_{i+2},\ldots,l_m)]$ where
      \[ f^{\ddag}_{i+1}(x,l_{i+2},\ldots,l_m) =
        f_{i+1}(l_1^{(i)}(x,l_{i+2},\ldots,l_m), \ldots,
        l_i^{(i)}(x,l_{i+2},\cdots,l_m),x,l_{i+2}\ldots,l_m).
      \] From the induction hypothesis it follows that
      \[ l_j^{(i)}(x,l_{i+2},\ldots,l_m) =
        \sol[j]{\subst{E}{\subvec{x}{i+1}{m}}{x,\subvec{l}{i+2}{m}}} =
        \sol[j]{\subst{E'}{x_{i+1}}{x}} =l_j \] for
      $j \in \interval{i}$.  We define
      $(l_1,\dots, l_i) = \sol{\subst{E'}{x_{i+1}}{x}}$ and observe
      that $l_{i+1}^{(i+1)}(l_{i+2},\ldots,l_m)$ is the
      $\eta_i$-fixpoint of
      $\lambda x.f_{i+1}(l_1,\ldots l_i,x,l_{i+2},\ldots,l_m)$. The
      same is true for $u_{i+1}$ and hence we conclude
      $l_{i+1}^{(i+1)}(l_{i+2},\ldots,l_m) = u_{i+1} = \sol[i+1]{E'}$.
      \label{hsc2a}
    \item $j \leq i:$ First, from Definition~\ref{de:solution} we
      obtain
      \[ \sol[j]{E'} =
        \sol[j]{\subst{E}{\subvec{x}{i+2}{m}}{\subvec{l}{i+2}{m}}} =
        \sol[j]{\subst{E}{\subvec{x}{i+1}{m}}{\subvec{l}{i+1}{m}}} \]
      where $l_{i+1} = \sol[i+1]{E'}$. From the induction hypothesis
      we know that
      $\sol[j]{\subst{E}{\subvec{x}{i+1}{m}}{\subvec{l}{i+1}{m}}} =
      l_{j}^{(i)}(l_{i+1},l_{i+2}, \ldots ,l_m) $.

      On the other hand we have from Definition~\ref{solution-Ichiro}
      that
      \[ l_{j}^{(i+1)}(l_{i+2} \ldots
        l_m)=l_{j}^{(i)}(l_{i+1}^{(i+1)}(l_{i+2}, \ldots l_m),l_{i+2},
        \ldots ,l_m) \] and from (\ref{hsc2a}) we finally get
      $l_{i+1}=l_{i+1}^{(i+1)}(l_{i+2}, \ldots l_m)$, which concludes
      the proof.
    \end{enumerate}
  \end{enumerate} 
\end{proof}

\subsection{Comparing $\mu$-Approximants and Lattice Progress
  Measures~\cite{hsc:lattice-progress-measures}}
\label{ssec:hsc-measures}
 
As hinted in the main body of the paper, $\mu$-approximants can be
seen as special lattice progress measures in the sense
of~\cite{hsc:lattice-progress-measures}, that we will refer here as
hsc-measures. More precisely, as discussed below, the function that,
for any $\mu$-approximant $\vec{l}$, maps the subvector of
$\ord{\vec{l}}$ obtained by keeping only the components corresponding
to $\mu$-indices to $\vec{l}$ is a hsc-measure. This is indeed the
hsc-measure used in~\cite[Theorem 2.13]{hsc:lattice-progress-measures}
(completeness part).

\begin{definition}[hsc-measure~\cite{hsc:lattice-progress-measures}]
  \label{de:hsc-measure}
  Let $L$ be a lattice and let $E$ be a system of equations over $L$
  of the kind $\vec{x} =_{\vec{\eta}} \vec{f}(\vec{x})$.  We assume
  that $i_1, \ldots, i_k$ are the indexes of $\mu$-equations and
  $j_1, \ldots, j_{m-k}$ are the indexes of $\nu$-equations, i.e.,
  $\eta_{i_h} = \mu$ for $h \in \{ 1, \ldots, k\}$ and
  $\eta_{j_h} = \nu$ for $h \in \{1, \ldots, m-k\}$.
  Given an $k$-tuple of ordinals $\vec{\gamma}$, a
  $\vec{\gamma}$-bounded \emph{hsc-measure} is a tuple of
  functions $\vec{p}\colon \down{\vec{\gamma}} \to L^m$ satisfying
  \begin{enumerate}
  \item (\emph{Monotonicity})
    For $\vec{\alpha}, \vec{\alpha}' \in \down{\vec{\gamma}}$ and
    $a \in \interval{k}$, if $\vec{\alpha} \preceq_a \vec{\alpha}'$ then
    for all $i \geq i_a$ it holds
    $p_i(\vec{\alpha}) \sqsubseteq p_i(\vec{\alpha}')$.
  
  \item (\emph{$\mu$-case})
    For $i \in \interval{m}$, $\eta_i = \mu$ and
    $i = i_a$ for some $a \in \interval{k}$ and
    $\vec{\alpha} = \vec{\alpha'} \alpha_{a} \vec{\alpha''} \in
    \down{\vec{\gamma}}$, we have
    (i)~$p_{i}(\vec{\alpha}',0,\vec{\alpha}'')=\bot$;
    (ii)~$p_{i}(\vec{\alpha}',\alpha+1,\vec{\alpha}'') \sqsubseteq
    f_{i}(\vec{p}(\vec{\beta}',\alpha,\vec{\alpha}''))$ for some
    $\vec{\beta}'$ and
    (iii)~$p_{i}(\vec{\alpha}',\alpha,\vec{\alpha}'') \sqsubseteq
    \bigsqcup_{\beta<\alpha}
    f_{i}(\vec{p}^E(\vec{\alpha}',\beta,\vec{\alpha}''))$ for $\alpha$ a
    limit ordinal.
    
  \item (\emph{$\nu$-case}) For $i \in \interval{m}$, $\eta_i = \nu$,
    $i_{a-1} < i < i_a$ for some $a \in \interval{k}$ and
    $\vec{\alpha} = \vec{\alpha}' \alpha_{a} \vec{\alpha''} \in
    \down{\vec{\gamma}}$, we have
    $p_{i}(\vec{\beta}',\alpha_{a},\vec{\alpha}'') \sqsubseteq
    f_{i}(\vec{p}(\vec{\beta}',\alpha_{a},\vec{\alpha}''))$ for some
    $\vec{\beta}'$.
  \end{enumerate}
\end{definition}

Note that by point (1), for $a \in \interval{k}$ and $i \geq i_a$, the
value of $p_i(\vec{\alpha})$ depends only the components of $\alpha$
of index greater or equal $a$. In fact for all $(m-a)$-tuples of
ordinals $\vec{\alpha}', \vec{\beta}'$ and $a$-tuples of ordinals
$\vec{\alpha}''$ we have 
$\vec{\alpha}' \vec{\alpha}'' \preceq_a \vec{\beta}' \vec{\alpha}''
\preceq_a \vec{\alpha}' \vec{\alpha}''$, hence
$p_i(\vec{\alpha}',\vec{\alpha}'') \sqsubseteq
p_i(\vec{\beta}',\vec{\alpha}'') \sqsubseteq
p_i(\vec{\alpha}',\vec{\alpha}'')$. Thus
$p_i(\vec{\alpha}',\vec{\alpha}'') =
p_i(\vec{\beta}',\vec{\alpha}'')$.

As mentioned above, we can easily adapt the definition of
$\mu$-approximant (Definition~\ref{de:approximants}) to get a
hsc-measure which can be shown to be the greatest one. Intuitively, the fact that $\mu$-approximants are closely related to the greatest hsc-meaure explains why our interest is mainly concentrated on $\mu$-approximants and their dual ($\nu$-approximants): the greatest hsc-measure surely provides a sound and complete approximation of the solution.

\begin{theorem}[$\mu$-approximants as hsc-measures]
  \label{thm:approximants}
  Let $E$ be a system of $m$ equations over the lattice $L$, of the
  kind $\vec{x} =_{\vec{\eta}} \vec{f}(\vec{x})$. Given
  $\vec{\alpha} \in \down{\asc{L}}^m$, define
  $\vec{p}^E(\alpha) = \vec{l}$ where for all $i \in \interval{m}$
  \begin{itemize}
  \item if $\eta_i = \nu$, then $l_i =
    \nu (f_{i,\vec{l}})$
  \item if $\eta_i = \mu$ with $i = i_a$ then
    $l_i = f_{i,{\vec{l}}}^{\alpha_a}(\bot)$
  \end{itemize}
  Then $\vec{p}$ is a $\down{\asc{L}}^m$-bounded hsc-measure and it is
  the greatest one.
\end{theorem}

\begin{proof}
  Let us start by proving that $p^E$ is a hsc-measure. Concerning
  monotonicity, let $a \in \interval{k}$, let
  $\vec{\alpha} \preceq_a \vec{\alpha}'$ and let
  $\vec{p}^E(\vec{\alpha}) = \vec{l}$ and
  $\vec{p}^E(\vec{\alpha}) = \vec{l}'$.
  We can show that for any $i \geq i_a$ it holds
  $l_i \sqsubseteq l_i'$, by means of an inductive argument (on
  $m-i$). 
  The base case is $i = m$. Observe that
  $f_{m,\vec{l}} = f_{m,\vec{l}'}$ (as $f_{m,\vec{l}}$ is independent
  of $\vec{l}$). There are two possibilities.
  \begin{itemize}
  \item If $\eta_m = \mu$ then $i_k = m$. Since
    $\alpha_k \leq \alpha_k'$ we have
    $l_m = f_{m,\vec{l}}^{\alpha_k}(\bot) \sqsubseteq
    f_{m,\vec{l}}^{\alpha_k'}(\bot) = f_{m,\vec{l}'}^{\alpha_k'}(\bot)
    = l_m'$.
  \item If $\eta_m = \nu$ then we have
    $l_m = \nu(f_{m,\vec{l}}) = \nu(f_{m,\vec{l}'}) = l_m'$.
  \end{itemize}

  For $i_a \leq i < m$, observe that by inductive hypothesis
  $\subvec{l}{i+1}{m} \sqsubseteq \subvec{l'}{i+1}{m}$. Hence by monotonicity
  of the solution (Lemma~\ref{le:solution-monotone})
  $\sol{\subst{\subst{E}{\subvec{x}{i+1}{m}}{\subvec{l}{i+1}{m}}}{x_i}{x}}
  \sqsubseteq
  \sol{\subst{\subst{E}{\subvec{x}{i+1}{m}}{\subvec{l'}{i+1}{m}}}{x_i}{x}}$.
  Thus, using the fact that $f_i$ is monotonic, we have that
  \begin{align*}
    f_{i,{\vec{l}}}(x)   & =
    f_i(\sol{\subst{\subst{E}{\subvec{x}{i+1}{m}}{\subvec{l}{i+1}{m}}}{x_i}{x}},
    x, \subvec{l}{i+1}{m})\\
                          &  \sqsubseteq f_i(\sol{\subst{\subst{E}{\subvec{x}{i+1}{m}}{\subvec{l'}{i+1}{m}}}{x_i}{x}},
    x, \subvec{l'}{i+1}{m}) = f_{i,{\vec{l'}}}(x).
  \end{align*}
  Given the above, reasoning as in
  the base case, we conclude $l_i \sqsubseteq l_i'$.

  \smallskip

  As a direct consequence of the definition of $\vec{p}^E$ we can show
  that the properties of the $\mu$-case and $\nu$-case in
  Definition~\ref{de:hsc-measure} hold with equality replacing
  $\sqsubseteq$ and the tuple of ordinals
  $\vec{\beta}' = (\asc{L}, \ldots, \asc{L}) = \vec{\asc{L}}$. More
  precisely
  \begin{itemize}
  \item (\emph{$\mu$-case})
    For $i \in \interval{m}$, $\eta_i = \mu$ and
    $i = i_a$ for some $a \in \interval{k}$ and
    $\vec{\alpha} = \vec{\alpha} \alpha_{a} \vec{\alpha''} \in
    \down{\vec{\gamma}}$, we have
    (i)~$p_{i}(\vec{\alpha}',0,\vec{\alpha}'')=\bot$;
    (ii)~$p_{i}(\vec{\alpha}',\alpha+1,\vec{\alpha}'') =
    f_{i}(\vec{p}^E(\vec{\asc{L}},\alpha_a,\vec{\alpha}''))$, and
    (iii)~$p_{i}(\vec{\alpha}',\alpha_a,\vec{\alpha}'') =
    \bigsqcup_{\beta<\alpha_a}
    f_{i}(\vec{p}^E(\vec{\alpha}',\beta,\vec{\alpha}''))$ for $\alpha$
    a limit ordinal.
    
  \item (\emph{$\nu$-case}) For $i \in \interval{m}$, $\eta_i = \nu$,
    $i_{a-1} < i < i_a$ for some $a \in \interval{k}$ and
    $\vec{\alpha} = \vec{\alpha}' \alpha_{a} \vec{\alpha''} \in
    \down{\vec{\gamma}}$, we have
    $p_{i}(\vec{\alpha}',\alpha_{a},\vec{\alpha}'') =
    f_{i}(\vec{p}^E(\vec{\asc{L}},\alpha_{a},\vec{\alpha}''))$.
  \end{itemize}

  In fact

  \begin{itemize}
  \item (\emph{$\mu$-case})
    For $i \in \interval{m}$, $\eta_i = \mu$ and
    $i = i_a$ for some $a \in \interval{k}$ and
    $\vec{\alpha} = \vec{\alpha} \alpha_{a} \vec{\alpha''} \in
    \down{\vec{\gamma}}$, we have
    \begin{enumerate}[(i)]
      
    \item (\emph{Base}) If $\vec{p}^E(\vec{\alpha}',0,\vec{\alpha}'') =\vec{l}$
      then, by definition, we have
      \begin{center}
        $p^E_{i}(\vec{\alpha}',0,\vec{\alpha}'') = l_i =
        f_{i,\vec{l}}^0(\bot) = \bot$
      \end{center}

    \item (\emph{Successor}) For the case of a successor ordinal, if we let
      $\vec{p}^E(\vec{\alpha}',\alpha_a+1,\vec{\alpha}'') = \vec{l}$, we
      have:
      \begin{align*}
        p^E_{i}(\vec{\alpha}',\alpha_a+1,\vec{\alpha}'')
        & =  f_{i,{\vec{l}}}^{\alpha_a+1}(\bot)\\
        & = f_{i,{\vec{l}}}(f_{i,{\vec{l}}}^{\alpha_a}(\bot))\\
        & = f_i(\sol{\subst{\subst{E}{\subvec{x}{i+1}{m}}{\subvec{l}{i+1}{m}}}{x_i}{f_{i,{\vec{l}}}^{\alpha_a}(\bot)}},
    f_{i,{\vec{l}}}^{\alpha_a}(\bot), \subvec{l}{i+1}{m})
      \end{align*}

      Let
      $\vec{l}' = (\sol{\subst{\subst{E}{\subvec{x}{i+1}{m}}
          {\subvec{l}{i+1}{m}}}{x_i}
        {f_{i,{\vec{l}}}^{\alpha_a}(\bot)}},
      f_{i,{\vec{l}}}^{\alpha_a}(\bot), \subvec{l}{i+1}{m})$.

      Observe that for $j>i$ we have
      $l_j' = l_j = p_{j}(\vec{\asc{L}},\alpha_a,\vec{\alpha}'')$
      since $p_j$ only depends on $\subvec{\alpha}{j}{m}$. Using
      this fact, we also get
      $l_i' = p_{i}(\vec{\asc{L}},\alpha_a,\vec{\alpha}'')$.
      Finally, for $j<i$, it holds
      $l_j' = \sol[j]{\subst{E}{\subvec{x}{i}{m}}
        {\subvec{l'}{i}{m}}}$.
      Hence, if $\eta_j = \nu$, we have
      $l_j' = \nu(f_{i,\vec{l}'}) =
      p^E_j(\vec{\asc{L}},\alpha_a,\vec{\alpha}'')$ and, if
      $\eta_j = \mu$ we have
      $l_j' = \mu(f_{i,\vec{l}'}) = f_{i,\vec{l}'}^{\asc{l}}(\bot) =
      p^E_j(\vec{\asc{L}},\alpha_a,\vec{\alpha}'')$.

      The above shows that
      $\vec{l}' = \vec{p}^E(\vec{\asc{L}},\alpha_a, \vec{\alpha}'')$
      and therefore
      $p^E_{i}(\vec{\alpha}',\alpha_a+1,\vec{\alpha}'') =
      f_i(\vec{p}^E(\vec{\asc{L}},\alpha_a,\vec{\alpha}''))$, as desired.

    \item (\emph{Limit}) If $\alpha_a$ is a limit ordinal, let
      $\vec{p}^E(\vec{\alpha}',\alpha_a,\vec{\alpha}'') =
      \vec{l}$. Then we have:
      \begin{align}
        \label{eq:hsc-limit}
        p^E_{i}(\vec{\alpha}',\alpha_a,\vec{\alpha}'')
        = f_{i,{\vec{l}}}^{\alpha_a}(\bot)
        = \bigsqcup_{\beta < \alpha_a} f_{i,{\vec{l}}}^{\beta}(\bot)
      \end{align}
      Now, observe that when taking the join above we can restrict to
      successor ordinals $\beta = \beta' + 1$, and thus, reasoning as
      in the previous case, we get that
      $f_{i,{\vec{l}}}^{\beta}(\bot) = f_i(\vec{p}^E(\vec{\asc{L}},
      \beta', \vec{\alpha}'')) = p^E_i(\vec{\asc{L}}, \beta'+1,
      \vec{\alpha}'') = p^E_i(\vec{\asc{L}}, \beta, \vec{\alpha}'') =
      p^E_i(\vec{\alpha}', \beta, \vec{\alpha}'')$, since $p^E_i$ only
      depends on the components $i, \ldots, m$ of its argument.
      Therefore, replacing in (\ref{eq:hsc-limit}) we obtain
      \begin{align}
        p^E_{i}(\vec{\alpha}',\alpha_a,\vec{\alpha}'')
        = \bigsqcup_{\beta < \alpha_a} p^E_i(\vec{\alpha}', \beta, \vec{\alpha}'')
      \end{align}
      as desired.
    \end{enumerate}
    
  \item (\emph{$\nu$-case}) For $i \in \interval{m}$, $\eta_i = \nu$,
    $i_{a-1} < i < i_a$ for some $a \in \interval{k}$ and
    $\vec{\alpha} = \vec{\alpha}' \alpha_{a} \vec{\alpha''} \in
    \down{\vec{\gamma}}$, if we let $\vec{p}^E(\vec{\alpha}',
    \beta', \vec{\alpha}'') = \vec{l}$ we have
    $p_{i}(\vec{\alpha}',\alpha_{a},\vec{\alpha}'') = \nu (f_{i,\vec{l}})$.
    Therefore
    \begin{align*}
      p^E_{i}(\vec{\alpha}',\alpha_a,\vec{\alpha}'')
      & =  l_i \\
      & =  f_{i,\vec{l}}(l_i) =\\
      & = f_i(\sol{\subst{E}{\subvec{x}{i+1}{m}}{\subvec{l}{i}{m}}},
        \subvec{l}{i}{m})
    \end{align*}

    Let
    $\vec{l}' = (\sol{\subst{E}{\subvec{x}{i}{m}}
        {\subvec{l}{i}{m}}},\subvec{l}{i}{m})$.
    Observe that for $j \geq i$ we have
    $l_j' = l_j = p_{j}(\vec{\alpha}',\alpha_a,\vec{\alpha}'') =
    p_{j}(\vec{\asc{L}},\alpha_a,\vec{\alpha}'')$ since $p_j$ only
    depends on $\subvec{\alpha}{j}{m}$. 
    Instead, for $j<i$, it holds
    $l_j' = \sol[j]{\subst{E}{\subvec{x}{i}{m}}
      {\subvec{l'}{i}{m}}}$.
    Hence, if $\eta_j = \nu$, we have
    $l_j' = \nu(f_{i,\vec{l}'}) =
    p^E_j(\vec{\asc{L}},\alpha_a,\vec{\alpha}'')$ and, if
    $\eta_j = \mu$ we have
    $l_j' = \mu(f_{i,\vec{l}'}) = f_{i,\vec{l}'}^{\asc{l}}(\bot) =
    p^E_j(\vec{\asc{L}},\alpha_a,\vec{\alpha}'')$.
    
    The above shows that
    $\vec{l}' = \vec{p}^E(\vec{\asc{L}},\alpha_a, \vec{\alpha}'')$
    and therefore
    $p^E_{i}(\vec{\alpha}',\alpha_a,\vec{\alpha}'') =
    f_i(\vec{p}^E(\vec{\asc{L}},\alpha_a,\vec{\alpha}''))$, as desired.

  \end{itemize}

  This proves that $\vec{p}^E$ is a hsc-measure.

  \smallskip

  In addition, for any other progress measure $\vec{p}$ it holds that
  for any $\vec{\alpha} \in \down{\asc{L}}^m$ and
  $i \in \interval{m}$, we have
  $p_i(\vec{\alpha}) \sqsubseteq p^E_i(\vec{\alpha})$. The proof
  proceeds by induction on the ordinal vector $\vec{\alpha}$ with
  respect to the well-founded order $\preceq$. In order to show that
  for all $i \in \interval{m}$,
  $p_i(\vec{\alpha}) \sqsubseteq p^E_i(\vec{\alpha})$ we proceed by
  induction on $m-i$. If $\eta_i = \mu$, consider the index $a$ such
  that $i_a = i$.
  If $\alpha_a =0$ then
  $p_i(\vec{\alpha}) = \bot = p^E_i(\vec{\alpha})$. For a successor
  ordinal $\alpha_a = \alpha + 1$,
  \begin{align*}
    p_{i}(\vec{\alpha}',\alpha + 1,\vec{\alpha}'')
    & \sqsubseteq  f_i(\vec{p}(\vec{\beta}',\alpha,\vec{\alpha}''))
    & \mbox{[for some $\vec{\beta}'$, by Def.~\ref{de:hsc-measure}(2)]}\\
    & \sqsubseteq  f_i(\vec{p}(\vec{\asc{L}},\alpha,\vec{\alpha}''))
    & \mbox{[by Def.~\ref{de:hsc-measure}(1) and monotonicity of $f_i$]}\\
    & \sqsubseteq  f_i(\vec{p}^E(\vec{\asc{L}},\alpha_a,\vec{\alpha}''))
    & \mbox{[by ind. hyp. and monotonicity of $f_i$]}\\
    & = p^E_{i}(\vec{\asc{L}},\alpha+1,\vec{\alpha}'')
    & \mbox{[by $\mu$-case, property (ii) of $\vec{p}^E$]}\\
    & = p^E_{i}(\vec{\alpha}',\alpha+1,\vec{\alpha}'')
    & \mbox{[since $p^E_i$ only depends on components  $i, \ldots, m$]}
  \end{align*}
  When $\alpha_a$ is a limit ordinal, by inductive hypothesis we know that for all $\beta < \alpha_a$ we have
  $p_i(\vec{\alpha}', \beta, \vec{\alpha}'') \sqsubseteq
  p^E_i(\vec{\alpha}', \beta, \vec{\alpha}'')$ and thus
  \begin{align*}
    p_{i}(\vec{\alpha}',\alpha_a,\vec{\alpha}'')
    & \sqsubseteq  \bigsqcup_{\beta<\alpha_a}
      f_i(\vec{p}(\vec{\alpha}',\beta,\vec{\alpha}''))
    & \mbox{[since $\vec{p}$ is a progress measure]}\\
    & \sqsubseteq  \bigsqcup_{\beta<\alpha_a}
      f_i(\vec{p}^E(\vec{\alpha}',\beta,\vec{\alpha}''))
    & \mbox{[by ind. hyp. and monotonicity of $f_i$]}\\
    & = \vec{p}^E(\vec{\alpha}',\alpha_a,\vec{\alpha}'')  
    & \mbox{[by $\mu$-case, property (iii) of $\vec{p}^E$]}
  \end{align*}

  For the case $\eta_i = \nu$, let $a \in \interval{k}$ be the index
  such that $i_{a-1} < i < i_a$ and let
  $\vec{\alpha} = \vec{\alpha}' \alpha_{a} \vec{\alpha''} \in
  \down{\vec{\gamma}}$. By Definition~\ref{de:hsc-measure}(3), there
  must exist $\vec{\beta}'$ such that
  \begin{align}
    \label{eq:nu-hsc0}
    p_{i}(\vec{\alpha}',\alpha_{a},\vec{\alpha}'') \sqsubseteq
    f_{i}(\vec{p}(\vec{\beta}',\alpha_{a},\vec{\alpha}'')).
  \end{align}
  
  Let $\vec{\beta} = (\vec{\beta}',\alpha_a,\vec{\alpha}'')$. 
  By inner inductive hypothesis, for $j>i$
  \begin{align}
    \label{eq:nu-hsc1}
    p_j(\vec{\beta}) \sqsubseteq 
    p^E_j(\vec{\beta})
  \end{align}

  Moreover, as shown
  in~\cite[Thm. 2.13]{hsc:lattice-progress-measures} (property denoted
  by $(*)$), we have that for $j<i$ it holds that
  \begin{center}
    $p_j(\vec{\beta}) \sqsubseteq
    \sol[j]{
      \subst{
        \subst{E}
        {\subvec{x}{i+1}{m}}
        {\subvec{p}{i+1}{m}(\vec{\beta})}
        }
        {x_i}
        {p_i(\vec{\beta})}
    }$
  \end{center}
  In turn, by monotonicity of the solution
  (Lemma~\ref{le:solution-monotone}) and (\ref{eq:nu-hsc1}), we get
  
  \begin{align}
    \label{eq:nu-hsc2}
    p_j(\vec{\beta}) \sqsubseteq
    \sol[j]{
      \subst{
        \subst{E}
        {\subvec{x}{i+1}{m}}
        {\subvec{p^E}{i+1}{m}(\vec{\beta})}
        }
        {x_i}
        {p_i(\vec{\beta})}
    }
  \end{align}

  Therefore, putting things together, from (\ref{eq:nu-hsc0}), we get
  \begin{align*}    
    & p_{i}(\vec{\alpha}',\alpha_a,\vec{\alpha}'')\\
    & \sqsubseteq  f_i(\vec{p}(\vec{\beta}))\\
    & \sqsubseteq  f_i(
      \sol{
      \subst{
        \subst{E}
        {\subvec{x}{i+1}{m}}
        {\subvec{p^E}{i+1}{m}(\vec{\beta})}
        }
        {x_i}
        {p_i(\vec{\beta})}
      },
      p_i(\vec{\beta}),
      \subvec{p^E}{i+1}{m}(\vec{\beta}
      ))\\
    & \quad \mbox{[by (\ref{eq:nu-hsc1}) and (\ref{eq:nu-hsc2}) and monotonicity of $f_i$]}\\
    & \sqsubseteq  f_{i,\vec{p}^E(\vec{\beta})}(p_i(\vec{\beta}))
    & 
  \end{align*}  
    
  Recalling that $p_i$ only depends on components $i+1, \ldots, m$, we
  have that
  $p_{i}(\vec{\alpha}',\alpha_a,\vec{\alpha}'') =
  p_{i}(\vec{\beta}',\alpha_a,\vec{\alpha}'') = p_{i}(\vec{\beta})$,
  i.e., the inequality above can be rewritten as
  \begin{center}
    $p_{i}(\vec{\beta}) \sqsubseteq
    f_{i,\vec{p}^E(\vec{\beta})}(p_i(\vec{\beta}))$.
  \end{center}

  This means that $p_i(\vec{\beta})$ is a post-fixpoint of
  $f_{i,\vec{p}^E(\vec{\beta})}$, and by definition of $\vec{p}^E$, we
  have that $p_i^E(\vec{\beta})$ is the greatest fixpoint of
  $f_{i,\vec{p^E}(\beta)}$. Therefore
  $p_i(\vec{\beta}) \sqsubseteq p^E_i(\vec{\beta})$. Recalling that
  $\vec{\beta} = (\vec{\beta}',\alpha_a,\vec{\alpha}'')$ and, again,
  that $p_i$ only depends on components $i+1, \ldots, m$, we conclude
  the desired inequality
  \begin{center}
    $p_i(\vec{\alpha}',\alpha_a,\vec{\alpha}'') =
    p_i(\vec{\beta})  \sqsubseteq p^E_i(\vec{\beta})
    = p^E_i(\vec{\alpha}',\alpha_a,\vec{\alpha}'')$.
  \end{center}
\end{proof}

Note that the characterisation of $\vec{p}^E$ used in the proof offers
an alternative method for computing approximations of the fixpoint and
could be interesting in its own right, for instance for cases where
the basis is too large -- for instance for the reals -- and it is
infeasible to determine the progress measure for every element of the
basis.

The notion of \emph{matrix progress measure} (MPM)
in~\cite{hsc:lattice-progress-measures}, which is introduced for
powerset lattices, is closely related to the game-theoretical progress
measure that we proposed for equations over continuous lattices: it
can be seen as an instance of our notion for systems of equations
arising from formulae in the coalgebraic $\mu$-calculus.

\section{Technical results}
\label{sec:technical}

\subsection{Sup-Respecting Progress Measures
  (\S~\ref{ssec:progress-fix})}
\label{ssec:sup-respecting}

  In order to show that $\Phi_E$ preserves sup-respecting functions $R$
we first need a technical lemma that will also prove useful for the
logic characterising symbolic $\exists$-moves.

\begin{lemma}
  \label{lem:exchange-min-sup}
  Let $L$ be a continuous lattice and let $(U_k)_{k\in K}$ with
  $U_k\subseteq L^m$ be a collection of upward-closed sets. Assume
  that $R\colon L \to \interval{m} \to \lift{\asc{L}}{m}$ is
  sup-respecting. Then it holds that:
  \begin{eqnarray*}
    && \min\nolimits_{\preceq_i} \{ \sup \{ R(b')(j) + \vec{\delta}^{\eta_i}_i
    \mid j \in \interval{m}\ \land\ b' \ll l_j\} \mid \vec{l} \in
    \bigcap_{k\in K} U_k \} \\
    & = & \sup_{k\in K} \min\nolimits_{\preceq_i} \{ \sup
    \{ R(b')(j) + \vec{\delta}^{\eta_i}_i \mid j \in \interval{m}\
    \land\ b' \ll l_j\} \mid \vec{l} \in U_k \}
  \end{eqnarray*}
\end{lemma}

\begin{proof}
  Since all $U_k$ are upward-closed, their intersection can be written
  as
  $\bigcap_{k\in K} U_k = \{\bigsqcup_{k\in K} \vec{l}^k\mid \vec{l}^k\in
  U_k, k\in K\}$ (where suprema of $L^m$ are taken pointwise).  Hence
  the left-hand side of the equation can be rewritten to
  \[
    \min\nolimits_{\preceq_i} \{\sup \{ R(b')(j) +
    \vec{\delta}^{\eta_i}_i \mid j \in \interval{m}\ \land\ b' \ll
    \bigsqcup_{k\in K} l_j^k \} \mid \vec{l}^k\in U_k, k\in K \}
  \]
  We first show that for $j\in\interval{m}$
  \[ \sup\{ R(b')(j) + \vec{\delta}_i^{\eta_i} \mid b'\ll \bigsqcup_{k\in K}
    l_j^k \} = \sup_{k\in K} \sup \{ R(b')(j) + \vec{\delta}_i^{\eta_i} \mid
    b'\ll l_j^k \}\]
  \begin{itemize}
  \item ($\sqsupseteq$) This direction is obvious since $b'\ll l_j^k$
    implies $b'\ll \bigsqcup_{k\in K}l_j^k$. Hence every ordinal
    vector of the form $R(b')(j) + \vec{\delta}_i^{\eta_i}$ which is
    contained in the right-hand side set is automatically a member of
    the left-hand side set.
  \item ($\sqsubseteq$) Let $b'\ll \bigsqcup_{k\in K} l_j^k$.  This
    implies that $b'\ll \bigsqcup_{k\in K} l_j^k = \bigsqcup Y$ where
    $Y = \bigcup_{k\in K} (\cone{l_j^k}\cap B_L)$, since we are in a
    continuous lattice. Then there exists a finite subset
    $Y'\subseteq Y$ such that $b'\sqsubseteq \bigsqcup Y'$
    (see~\cite[Remark on p.~50]{ghklms:continuous-lattices-domains}).

    Since $R$ is sup-respecting we have
    \begin{eqnarray*}
      R(b')(j) + \vec{\delta}_i^{\eta_i} & \preceq &
      \big(\sup_{y\in Y'} R(y)(j)\big) + \vec{\delta}_i^{\eta_i} \\
      & = & \sup_{y\in Y'} \big(R(y)(j) + \vec{\delta}_i^{\eta_i}\big)
      \\
      & \preceq & \sup_{y\in Y} \big(R(y)(j) + \vec{\delta}_i^{\eta_i}\big) \\
      & = & \sup_{k\in K} \sup \{ R(y)(j) +
      \vec{\delta}_i^{\eta_i} \mid y\in Y\cap \cone{l_j^k} \} \\
      & = & \sup_{k\in K} \sup \{ R(b')(j) + \vec{\delta}_i^{\eta_i}
      \mid b'\ll l_j^k \} 
    \end{eqnarray*}
    Note that the first equality is due to the fact that~$Y'$ is
    finite and non-empty.
    
    Since the left-hand side of the equation is the supremum of all
    such $R(b')(j)$ and we have shown that the right-hand side is an
    upper bound, the result follows.
  \end{itemize}
  Now we can conclude by showing that
  \begin{eqnarray*}
    && \min\nolimits_{\preceq_i} \{
    \sup_{j\in\interval{m}} \sup \{ R(b')(j) + \vec{\delta}^{\eta_i}_i
    \mid b' \ll \bigsqcup_{k\in K} l_j^k \} \mid
    \vec{l}^k\in U_k, k\in K \} \\
    & = & \min\nolimits_{\preceq_i} \{ \sup_{j\in\interval{m}} \sup_{k\in K}
    \sup \{ R(b')(j) + \vec{\delta}_i^{\eta_i} \mid b'\ll l_j^k \} \mid
    \vec{l}^k\in U_k, k\in K \} \\
    & = & \min\nolimits_{\preceq_i} \{ \sup_{k\in K} \sup_{j\in\interval{m}}
    \sup \{ R(b')(j) + \vec{\delta}_i^{\eta_i} \mid b'\ll l_j^k \} \mid
    \vec{l}^k\in U_k, k\in K \} \\
    & = & \sup_{k\in K} \min\nolimits_{\preceq_i} \{ \sup_{j\in\interval{m}}
    \sup \{ R(b')(j) + \vec{\delta}_i^{\eta_i} \mid b'\ll l_j \} \mid
    \vec{l}\in U_k \} \\
    & = & \sup_{k\in K} \min\nolimits_{\preceq_i} \{ \sup \{ R(b')(j) +
    \vec{\delta}_i^{\eta_i} \mid j\in\interval{m} \land b'\ll l_j \} \mid
    \vec{l}\in U_k \}
  \end{eqnarray*}
  where the second-last equality is due to complete distributivity. 
\end{proof}

\begin{lemma}[$\Phi_E$ preserves sup-respecting functions]
  \label{le:sup-respecting-preserve}
  Let $L$ be a continuous lattice and let $E$ be a system of equations over $L$
  of the kind $\vec{x} =_{\vec{\eta}} \vec{f}(\vec{x})$. If
  $R\colon B_L \to \interval{m} \to \lift{\asc{L}}{m}$ is sup-respecting,
  then $\Phi_E(R)$ is sup-respecting as well.
\end{lemma}

\begin{proof}
  We assume that $R$ is sup-respecting and $R' = \Phi_E(R)$
  is as follows:
  \[
    R'(b)(i) = \min\nolimits_{\preceq_i}
    \{
    \sup \{ R(b')(j) + \vec{\delta}^{\eta_i}_i 
            \mid (b',j) \in \Amoves{\vec{l}} \}
    \mid \vec{l} \in \Emoves{b,i}
    \}
  \]
  The aim is to show that $R'$ is sup-respecting as well. Let
  $X\subseteq B_L$ be a set of basis elements such that
  $b\sqsubseteq \bigsqcup X$. Note furthermore that $\Emoves{b,i}$ is
  upwards-closed. We first show that
  \[ \bigcap_{b'\in X} \Emoves{b',i} \subseteq \Emoves{b,i} \]
  Let $\vec{l}\in \Emoves{b',i}$ for all $b'\in X$, which means that
  $b'\sqsubseteq f_i(\vec{l})$. So if we take the supremum over all
  $b'\in X$ we obtain
  $b\sqsubseteq \bigsqcup X \sqsubseteq \bigsqcup \{f_i(\vec{l})\} =
  f_i\big(\vec{l}\big)$. Hence $\vec{l}\in \Emoves{b,i}$, as required.

  Now we can apply Lemma~\ref{lem:exchange-min-sup} where $K = X$,
  $U_{b'} = \Emoves{b',i}$ and we obtain:
  \begin{eqnarray*}
    R'(b)(i) & = & \min\nolimits_{\preceq_i} \{ \sup \{
    R(b'')(j) + \vec{\delta}^{\eta_i}_i \mid
    (b'',j)\in\Amoves{\vec{l}}\} \mid \vec{l}\in \Emoves{b,i} \} \\
    & \preceq & \min\nolimits_{\preceq_i} \{ \sup \{
    R(b'')(j) + \vec{\delta}^{\eta_i}_i \mid (b'',j)\in
    \Amoves{\vec{l}}\}
    \mid \vec{l}\in \bigcap_{b'\in X} \Emoves{b',i} \} \\
    & = & \sup_{b'\in X} \min\nolimits_{\preceq_i} \{ \sup \{ R(b'')(j) +
    \vec{\delta}_i^{\eta_i} \mid (b'',j)\in
    \Amoves{\vec{l}} \} \mid
    \vec{l}\in \Emoves{b',i} \} \\
    & = & \sup_{b'\in X} R'(b')(i)
  \end{eqnarray*}
\end{proof}

\subsection{Compositionality for Selections
  (\S~\ref{ssec:selections})}
\label{ssec:compositional-selection}

In order to define selections compositionally we first need a technical lemma
that extends selections to generic elements of the lattice, possibly
not part of the basis.

\begin{lemma}[extending the selection]
  \label{le:selection-extended}
  Let $L$ be a continuous lattice with a basis $B_L$, let
  $f\colon L^m \to L$ be a monotone functions and let
  $\sigma\colon B_L \to \Pow{L^m}$ be a selection for $f$. Define
  $\bar{\sigma}\colon L \to \Pow{L^m}$ by
  $\bar{\sigma}(b) = \sigma(b)$ for $b \in B_L$ and
  $\bar{\sigma}(l) = \{ \bigsqcup_{b \ll l} \vec{l}^b \mid \vec{l}^b
  \in \sigma(b) \}$ for $l \in L \setminus B_L$. Then
  
  \begin{enumerate}
   
  \item for all $\vec{l} \in \bar{\sigma}(l)$ it holds $l \sqsubseteq  f(\vec{l})$;
  \item for all $\vec{l}' \in L^m$, if $l \sqsubseteq f(\vec{l}')$ then there
    exists $\vec{l} \in \bar{\sigma}(l)$ such that
    $\vec{l} \sqsubseteq \vec{l'}$.
  \end{enumerate}
\end{lemma}

\begin{proof}
  For $l \in B_L$, there is nothing to prove since the properties hold
  by definition of selection.

  Let $l \in L \setminus B_L$.  We start with point (1).  Let
  $\vec{l} \in \sigma(l)$, hence
  $\vec{l} = \bigsqcup_{b \ll l} \vec{l}^b$ with
  $\vec{l}^b \in \sigma(b)$ for each $b \ll l$. By the properties of
  selections, for all $b \ll l$, since $\vec{l}^b \in \sigma(l)$ it
  holds $b \sqsubseteq f(\vec{l}^b)$, hence
  \[
    b \sqsubseteq \bigsqcup_{b' \ll l} f(\vec{l}^b) \sqsubseteq
    f(\bigsqcup_{b' \ll l} \vec{l}^b) = f(\vec{l})
  \]
  the last inequality following by monotonicity of $f$. Therefore
  $l = \bigsqcup_{b \ll l} b \sqsubseteq f(\vec{l})$, as desired.

  Concerning point (2), let $\vec{l}' \in L^m$ be such that
  $l \sqsubseteq f(\vec{l}')$. For all $b \ll l$, since
  $b \sqsubseteq f(\vec{l}')$ there is $\vec{l}^b \in \sigma(b)$ such
  that $\vec{l}^b \sqsubseteq \vec{l}'$. Then we can consider
  $\vec{l} = \bigsqcup_{b \ll l} \vec{l}^b \sqsubseteq \vec{l}'$ which
  is in $\bar{\sigma}(l)$ by definition.
\end{proof}

We can now define the selection for a composition of functions.

\begin{lemma}[selection for composition]
  \label{le:selection-composition}
  Let $L$ be a continuous lattice with a basis $B_L$, and let $f\colon L^n \to L$ and
  $f_j\colon L^m \to L$, $j \in \interval{n}$ be monotone functions and
  let $\sigma\colon B_L \to \Pow{L^n}$ and $\sigma_j\colon B_L \to \Pow{L^m}$,
  $j \in \interval{n}$ be the corresponding selections. Consider
  the function $h\colon L^m \to L$ obtained as the composition
  $h(\vec{l}) = f(f_1(\vec{l}), \ldots, f_n(\vec{l}))$. Then
  $\sigma'\colon B_L \to \Pow{L^m}$ defined by
  \begin{center}
    $\sigma'(b) = \{ \bigsqcup_{i=1}^n \vec{l}^i \mid \exists \vec{l} \in
    \sigma(b). \forall i \in \interval{n}.\, \vec{l}^i \in \bar{\sigma}_i(l_i) \}$
  \end{center}
  is a selection for $h$.
\end{lemma}

\begin{proof}
  We show properties (1) and (2) of
  Definition~\ref{de:selection}. Let $b \in B_L$. Concerning (1), let
  $\vec{l}' \in \sigma'(b)$. Hence
  $\vec{l}' = \bigsqcup_{i=1}^n \vec{l}^i$ such that, for some
  $\vec{l} \in \sigma(b)$, for all $i \in \interval{n}$ we have
  $\vec{l}^i \in \bar{\sigma}_i(l_i)$.
  Since $\vec{l}^i \in \bar{\sigma}(l_i)$, by
  Lemma~\ref{le:selection-extended} and monotonicity of $f_i$ we have
  $l_i \sqsubseteq f_i(\vec{l}^i) \sqsubseteq f_i(\bigsqcup _{i=1}^n
  \vec{l}^i) = f_i(\vec{l}')$. Hence
  $\vec{l} \sqsubseteq (f_1(\vec{l}'), \ldots, f_n(\vec{l}'))$ and
  thus, by monotonicity of $f$,
  \[
    f(\vec{l}) \sqsubseteq f(f_1(\vec{l}'), \ldots, f_n(\vec{l}')) = h(\vec{l}')
  \]
  Recalling that $\vec{l} \in \sigma(b)$ and thus
  $b \sqsubseteq f(\vec{l})$ we conclude, by transitivity,
  $b \sqsubseteq h(\vec{l}')$, as desired.

  \smallskip

  Let us focus on property (2). Let $\vec{l} \in L^m$ be such that
  $b \sqsubseteq h(\vec{l}) = h(f_1(\vec{l}), \ldots,
  f_n(\vec{l}))$. Since $\sigma$ is a selection for $f$, there exists
  $\vec{l}' \in \sigma(b)$ such that
  $\vec{l}' \sqsubseteq (f_1(\vec{l}), \ldots, f_n(\vec{l}))$. Now,
  for all $i \in \interval{n}$, since $l_i' \sqsubseteq f_i(\vec{l})$,
  by Lemma~\ref{le:selection-extended}, there is
  $\vec{l}^i \in \bar{\sigma}_i(l_i')$ such that
  $\vec{l}^i \sqsubseteq \vec{l}$. If we let
  $\vec{l}'' = \bigsqcup_{i=1}^n \vec{l}^i$, by definition
  $\vec{l}'' \in \sigma'(b)$ and clearly
  $\vec{l}'' \sqsubseteq \vec{l}$, as desired.  
\end{proof}

\begin{example}
  \label{ex:running-selection-comp}
  Consider again our running example in Example~\ref{ex:running}. The
  selection discussed in Example~\ref{ex:running-selection} is
  computed using the observation in Example~\ref{ex:selection} and
  Lemma~\ref{le:selection-composition}.

  The fact that in this case the selections $\sigma_1$ and $\sigma_2$
  arising from the construction in
  Lemma~\ref{le:selection-composition} are the least ones is not a
  general fact.
  In order to ensure that starting from the least selections of the
  components we get the least selection we need to consider a more
  complex definition of extension (Lemma~\ref{le:selection-extended})
  that is omitted since we favour the use of the logic for symbolic
  $\exists$-moves (\S~\ref{sec:logic-selection}).
\end{example}

\end{document}